\documentclass[journal]{IEEEtran}%
\pdfoutput=1
\usepackage{graphics,amsmath,amsfonts,amscd,revsymb,latexsym,
enumerate,multirow,epsfig}
\usepackage{amsmath}
\usepackage{amsfonts}
\usepackage{amssymb}
\usepackage{graphicx}%
\providecommand{\U}[1]{\protect\rule{.1in}{.1in}}
\providecommand{\U}[1]{\protect\rule{.1in}{.1in}}
\newtheorem{theorem}{Theorem}

\newtheorem{corollary}[theorem]{Corollary}

\newtheorem{definition}{Definition}

\newtheorem{lemma}{Lemma}

\newtheorem{remark}{Remark}

\def\bi{\begin{itemize}}
\def\ei{\end{itemize}}
\def\be{\begin{equation}}
\def\ee{\end{equation}}
\def\bea{\begin{eqnarray}}
\def\eea{\end{eqnarray}}
\def\ben{\begin{eqnarray*}}
\def\een{\end{eqnarray*}}

\def\>{\rangle}
\def\<{\langle}

\newcommand{\1} I

\def\*{\star}

\def\0{{\mathbf{0}}}
\def\1{{\mathbf{1}}}
\def\2{{\mathbf{2}}}
\def\3{{\mathbf{3}}}
\def\4{{\mathbf{4}}}
\def\5{{\mathbf{5}}}
\def\6{{\mathbf{6}}}
\def\7{{\mathbf{7}}}
\def\8{{\mathbf{8}}}
\def\9{{\mathbf{9}}}

\begin{document}

\title{Trading classical communication, quantum communication, and entanglement in
quantum Shannon theory}
\author{Min-Hsiu Hsieh and Mark M. Wilde\thanks{Min-Hsiu Hsieh is with the ERATO-SORST
Quantum Computation and Information Project, Japan Science and Technology
Agency 5-28-3, Hongo, Bunkyo-ku, Tokyo, Japan and Mark M. Wilde conducted this
research with the Centre for Quantum Technologies, National University of
Singapore, 3 Science Drive 2, Singapore 117543, the Electronic Systems
Division, Science Applications International Corporation, 4001 North Fairfax
Drive, Arlington, Virginia, USA\ 22203, and the School of Computer Science,
McGill University, Montreal, Quebec, Canada H3H 2S5 (E-mail:
minhsiuh@gmail.com and mark.wilde@mcgill.ca)}}
\maketitle

\begin{abstract}
We give trade-offs between classical communication, quantum communication, and
entanglement for processing information in the Shannon-theoretic setting. We
first prove a \textquotedblleft unit-resource\textquotedblright\ capacity
theorem that applies to the scenario where only the above three noiseless
resources are available for consumption or generation. The optimal strategy
mixes the three fundamental protocols of teleportation, super-dense coding,
and entanglement distribution. We then provide an achievable rate region and a
matching multi-letter converse for the \textquotedblleft direct
static\textquotedblright\ capacity theorem. This theorem applies to the
scenario where a large number of copies of a noisy bipartite state are
available (in addition to consumption or generation of the above three
noiseless resources). Our coding strategy involves a protocol that we name the
\textit{classically-assisted state redistribution protocol }and the three
fundamental protocols. We finally provide an achievable rate region and a
matching mutli-letter converse for the \textquotedblleft direct
dynamic\textquotedblright\ capacity theorem. This theorem applies to the
scenario where a large number of uses of a noisy quantum channel are available
in addition to the consumption or generation of the three noiseless resources.
Our coding strategy combines the \textit{classically-enhanced father protocol}
with the three fundamental unit protocols.

\end{abstract}

\begin{IEEEkeywords}quantum Shannon theory, quantum communication, classical
communication, entanglement, entanglement-assisted quantum coding, direct
dynamic capacity theorem, direct static capacity theorem\end{IEEEkeywords}

\section{Introduction}

The publication of Shannon's classic article in 1948 formally marks the
beginning of information theory \cite{Shannon48}. Shannon's article states two
fundamental theorems: the source coding theorem and the channel coding
theorem. The source coding theorem concerns processing of a \textit{static}
resource---an information source that emits a symbol from an alphabet where
each symbol occurs with some probability. The proof of the theorem appeals to
the asymptotic setting where many copies of the static resource are available,
i.e., the information source emits a large number of symbols. The result of
the source coding theorem is a tractable lower bound on the compressibility of
the static resource. On the other hand, the channel coding theorem applies to
a \textit{dynamic} resource. An example of a dynamic resource is a noisy
bit-flip channel that flips each input bit with a certain probability. The
proof of the channel coding theorem again appeals to the asymptotic setting
where a sender consumes a large number of independent and identically
distributed (IID)\ uses of the channel to transmit information to a receiver.
The result of the channel coding theorem is a tractable upper bound on the
transmission rate for the dynamic resource.

Quantum Shannon theory has emerged in recent years as the quantum
generalization of Shannon's information theory. Schumacher established a
quantum source coding theorem that is a \textquotedblleft
quantized\textquotedblright\ version of Shannon's source coding theorem
\cite{JS94,Sch95}. Schumacher's static resource is a quantum information
source that emits a given quantum state with a certain probability. Holevo,
Schumacher, and Westmoreland followed by proving that the Holevo information
of a quantum channel is an achievable rate for transmitting classical
information over a noisy quantum channel \cite{Hol98,SW97}. Lloyd, Shor, and
Devetak then proved that the coherent information of a quantum channel is an
achievable rate for transmitting quantum data over a noisy quantum channel
\cite{Lloyd96,Shor02,Devetak03}. Both of these quantum channel coding theorems
exploit a dynamic resource---a noisy quantum channel that connects a sender to
a receiver.

Entanglement is a \textit{static resource} shared between a sender and
receiver. It is \textquotedblleft static\textquotedblright\ because a sender
and receiver cannot exploit entanglement alone to generate either classical
communication or quantum communication or both. However, they can exploit
entanglement and classical communication to communicate quantum
information---this protocol is the well-known teleportation protocol
\cite{BBCJPW93}. The super-dense coding protocol \cite{BW92} doubles the
classical capacity of a noiseless quantum channel by exploiting entanglement
in addition to the use of the noiseless quantum channel. These two protocols
and others demonstrate that entanglement is a valuable resource in quantum
information processing.

Several researchers have shown how to process entanglement in the asymptotic
setting where a large number of identical copies of an entangled state or a
large number of independent uses of a noisy channel are available to generate
entanglement. Bennett \textit{et al}. proved a coding theorem for entanglement
concentration that determines how much entanglement (in terms of maximally
entangled states) a sender and receiver can generate from pure bipartite
states \cite{BBPS96}. The reverse problem of entanglement dilution
\cite{LP99,HW02,HL02} shows that entanglement is not an inter-convertible
static resource (i.e, simulating pure bipartite states from maximally
entangled states requires a sublinear amount of classical communication). A
dynamic resource can also generate entanglement. Devetak proved a coding
theorem that determines how much entanglement a sender and receiver can
generate by sending quantum states through a noisy quantum channel
\cite{Devetak03}.

Quantum Shannon theory began with the aforementioned single-resource coding
theorems
\cite{JS94,Sch95,Hol98,SW97,Lloyd96,Shor02,Devetak03,BBPS96,LP99,HW02,HL02,DW03c}
and has advanced to include double-resource coding theorems---their
corresponding protocols either generate two different resources or they
generate one resource with the help of another
\cite{BDSW96,BSST01,Hol01a,HHHLT01,DW02,DS03,DW03b,DHW03,arx2004shor,DHW05RI}.
The result of each of these scenarios was an achievable two-dimensional
trade-off region for the resources involved in the protocols.

Quantum information theorists have organized many of the existing protocols
into a family tree \cite{DHW03,DHW05RI,ADHW06FQSW}. Furthermore, Devetak
\textit{et al}. proposed the resource inequality framework that establishes
many classical and quantum coding theorems as inter-conversions between
\textit{non-local information resources}~\cite{DHW03,DHW05RI}. The language of
resource inequalities provides structural insights into the relationships
between coding theorems in quantum Shannon theory and greatly simplifies the
development of new coding schemes. An example of one of the resource
trade-offs is the so-called \textquotedblleft father\textquotedblright%
\ achievable rate region. The father protocol exploits a noisy quantum channel
and shared noiseless entanglement to generate noiseless quantum communication.
The father achievable rate region illustrates trade-offs between entanglement
consumption and quantum communication.

In this article, we advance quantum Shannon theory to the triple resource
setting by giving the full triple trade-offs for both the static and dynamic
scenarios. This triple trade-off solution represents one of the most general
scenarios considered in quantum Shannon theory. Here, we study the interplay
of the most important noiseless resources in the theory of quantum
communication: namely, classical communication, quantum communication, and
entanglement, with general noisy resources.

The noisy static resource that we consider here is a shared noisy bipartite
state, and the dynamic resource that we consider is a noisy quantum channel.
We again appeal to the asymptotic setting where a large number of independent
copies or uses of the respective static or dynamic noisy resource are
available. For both the static and dynamic scenarios, we assume that the
sender and receiver either consume or generate noiseless classical
communication, noiseless quantum communication, and noiseless entanglement in
addition to the consumption of the noisy resource. The result is a
three-dimensional achievable rate region that gives trade-offs for the three
noiseless resources in both the static and dynamic scenarios. The rate of a
noiseless resource is negative if a protocol consumes the corresponding
resource, and its rate is positive if a protocol generates the resource. The
above interpretation of a negative rate first appeared in
Refs.~\cite{nature2005horodecki,cmp2007HOW} with the state merging protocol
and with subsequent appearance, for example, in
Refs.~\cite{DHW05RI,YD07QI,devetak:230501}. The present article's solution for
the static and dynamic scenarios contains both negative and positive rates.

Our current formulas characterizing the triple trade-off capacity regions are
alas of a \textquotedblleft multi-letter\textquotedblright\ nature, a problem
that plagues many results in quantum Shannon theory. A multi-letter formula is
one that involves an intractable optimization over an arbitrary number of uses
of a channel or a state, as opposed to a more desirable \textquotedblleft
single-letter\textquotedblright\ formula that involves a tractable
optimization over a single use of a channel or state. In principle, a
multi-letter characterization is an optimal solution, but the multi-letter
nature of our characterization of the capacity region implies that there may
be a slight room for improvement in the formulas when considering an
optimization over a finite number of uses---sometimes suboptimal protocols can
lead to an optimal characterization of a capacity region when taking the limit
over an arbitrary number of uses of a channel or a state.

Despite the multi-letter nature of our characterization in the general case,
we can find examples of shared states and quantum channels for which the
regions single-letterize. We show that the static region single-letterizes for
the special case of an \textquotedblleft erased state,\textquotedblright\ a
state that two parties obtain by sending one half of a maximally entangled
Bell state through an erasure channel. Our argument for single-letterization
is similar to an argument we presented in Ref.~\cite{HW08GFP}\ for the erasure
channel. We also show that the dynamic region single-letterizes for the case
of a qubit dephasing channel. This proof builds on earlier work in
Refs.~\cite{HW08GFP,BHTW10} to show that single-letterization holds. In a
later work~\cite{WH10}, we build on the efforts in Ref.~\cite{HW08GFP,BHTW10}
to give a concise, direct argument for single-letterization of the dynamic
capacity region for the full class of Hadamard channels.

An interesting aspect of this article is that we employ basic topological
arguments and \textit{reductio ad absurdum} arguments in the multi-letter
converse proofs of the triple trade-off capacity theorems. To our knowledge,
the mathematical techniques that we use are different from prior techniques in
the classical information theory literature or the quantum Shannon theoretic
literature (though, there are some connections to the techniques in
Ref.~\cite{DHW05RI}). Many times, we apply well-known results in quantum
Shannon theory to reduce a converse proof to that of a previously known
protocol (for example, we apply the well-known result that forward classical
communication does not increase the quantum capacity of a quantum channel
\cite{BDSW96,BKN98}). For the constructive part of the coding theorems, there
is no need to employ asymptotic arguments such as typical subspace techniques
\cite{NC00}\ or the operator Chernoff bound \cite{WinterPhD} because the
resource inequality framework is sufficient to prove the coding theorems. The
simplicity of our arguments perhaps reflects the maturation of the field of
quantum Shannon theory.

One benefit of the direct dynamic capacity theorem is that we are able to
answer a question concerning the use of entanglement-assisted coding
\cite{BDH06,BDH06IEEE,HBD07,arx2007wildeEA,arx2007wildeEAQCC,HBD08QLDPC,arx2008wildeUQCC,arx2008wildeGEAQCC}%
\ versus the use of teleportation. We show exactly how entanglement-assisted
coding is superior to mere teleportation. We consider this result an important
corollary of the results in this article because it is rare that quantum
Shannon theory gives insight into practical error correction schemes.

We structure this article as follows. In the next section, we establish some
definitions and notation that prove useful for later sections.
Section~\ref{sec:summary}\ briefly summarizes our three main theorems:\ the
unit resource capacity region, the direct static capacity region, and the
direct dynamic capacity theorem. In Section~\ref{sec:unit-resource}, we prove
the optimality of a unit resource capacity region. The unit resource capacity
region does not include a static or dynamic resource, but includes the
resources of noiseless classical communication, noiseless quantum
communication, and noiseless entanglement only. The unit resource capacity
theorem shows that a mixed strategy combining teleportation \cite{BBCJPW93},
super-dense coding \cite{BW92}, and entanglement distribution \cite{DHW03}\ is
optimal whenever a static or dynamic noisy resource is not available.
Section~\ref{sec:static}\ states and proves the direct static capacity
theorem. This theorem determines the trade-offs between the three noiseless
resources when a noisy static resource is available. Section~\ref{sec:dynamic}%
\ states and proves the direct dynamic capacity theorem. This theorem
determines the triple trade-offs when a noisy dynamic resource is available.
We end with a discussion of the results in this article and future open problems.

\section{Definitions and Notation}

\label{sec:notation}We first establish some notation before proceeding to the
main theorems. We review the notation for the three noiseless unit resources
and that for resource inequalities. We establish some notation for handling
geometric objects such as lines, quadrants, and octants in the
three-dimensional space of classical communication, quantum communication, and entanglement.

The three fundamental resources are noiseless classical communication,
noiseless quantum communication, and noiseless entanglement. Let $\left[
c\rightarrow c\right]  $ denote one \textit{cbit} of noiseless forward
classical communication, let $\left[  q\rightarrow q\right]  $ denote one
\textit{qubit} of noiseless forward quantum communication, and let $\left[
qq\right]  $ denote one \textit{ebit} of shared noiseless entanglement
\cite{DHW03,DHW05RI}. The ebit is a maximally entangled state%
\[
\left\vert \Phi^{+}\right\rangle ^{AB}\equiv(\left\vert 00\right\rangle
^{AB}+\left\vert 11\right\rangle ^{AB})/\sqrt{2},
\]
shared between two parties $A$ and $B$ who bear the respective names Alice and
Bob. The ebit $\left[  qq\right]  $ is a unit static resource and both the
cbit $\left[  c\rightarrow c\right]  $ and the qubit $\left[  q\rightarrow
q\right]  $ are unit dynamic resources.

We consider two noisy resources:\ a noisy static resource and a noisy dynamic
resource. Let $\rho^{AB}$ denote the noisy static resource:\ a noisy bipartite
state shared between Alice and Bob. Let $\mathcal{N}^{A^{\prime}\rightarrow
B}$ denote a noisy dynamic resource:\ a noisy quantum channel that connects
Alice to Bob. Throughout this article, the dynamic resource $\mathcal{N}%
^{A^{\prime}\rightarrow B}$ is a completely positive and trace-preserving map
that takes density operators in the Hilbert space of Alice's system
$A^{\prime}$ to Bob's system $B$.

Resource inequalities are a compact, yet rigorous, way to state coding
theorems in quantum Shannon theory \cite{DHW03,DHW05RI}. In this article, we
formulate resource inequalities that consume the above noisy resources and
either consume or generate the noiseless resources. An example from
Refs.~\cite{DHW03,DHW05RI} is the following \textquotedblleft
mother\textquotedblright\ resource inequality:%
\[
\langle\rho^{AB}\rangle+\left\vert Q\right\vert \left[  q\rightarrow q\right]
\geq E\left[  qq\right]  .
\]
It states that a large number $n$\ of copies of the state $\rho^{AB}$ and
$n\left\vert Q\right\vert $ uses of a noiseless qubit channel are sufficient
to generate $nE$ ebits of entanglement while tolerating an arbitrarily
small\ error in the fidelity of the produced ebits. The rates $Q$ and $E$ of
respective qubit channel consumption and entanglement generation are entropic
quantities:%
\begin{align*}
\left\vert Q\right\vert  &  =\frac{1}{2}I\left(  A;E\right)  ,\\
E &  =\frac{1}{2}I\left(  A;B\right)  .
\end{align*}
See Ref.~\cite{Yard05a}\ for definitions of entropy and mutual information.
The entropic quantities are with respect to a state $\left\vert \psi
\right\rangle ^{EAB}$ where $\left\vert \psi\right\rangle ^{EAB}$ is a
purification of the noisy static resource state $\rho^{AB}$ and $E$ is the
purifying reference system (it should be clear when $E$ refers to the
purifying system and when it refers to the rate of entanglement generation).
We take the convention that the rate $Q$ is negative and $E$ is positive
because the protocol consumes quantum communication and generates entanglement
(this convention is the same as in
Refs.~\cite{nature2005horodecki,cmp2007HOW,DHW05RI,YD07QI,devetak:230501}).

We make several geometric arguments throughout this article because the static
and dynamic capacity regions lie in a three-dimensional space with points that
are rate triples $\left(  C,Q,E\right)  $. $C$ represents the rate of
classical communication, $Q$ the rate of quantum communication, and $E$ the
rate of entanglement consumption or generation. Let $L$ denote a line, $Q$ a
quadrant, and $O$ an octant in this space (it should be clear from context
whether $Q$ refers to quantum communication or \textquotedblleft
quadrant\textquotedblright). For example, $L^{-00}$ denotes a line going in
the direction of negative classical communication:%
\[
L^{-00}\equiv\left\{  \alpha\left(  -1,0,0\right)  :\alpha\geq0\right\}  .
\]
$Q^{0+-}$ denotes the quadrant where there is zero classical communication,
generation of quantum communication, and consumption of entanglement:%
\[
Q^{0+-}\equiv\left\{  \alpha\left(  0,1,0\right)  +\beta\left(  0,0,-1\right)
:\alpha,\beta\geq0\right\}  .
\]
$O^{+-+}$ denotes the octant where there is generation of classical
communication, consumption of quantum communication, and generation of
entanglement:%
\[
O^{+-+}\equiv\left\{
\begin{array}
[c]{c}%
\alpha\left(  1,0,0\right)  +\beta\left(  0,-1,0\right)  +\gamma\left(
0,0,1\right) \\
:\alpha,\beta,\gamma\geq0
\end{array}
\right\}  .
\]

It proves useful to have a \textquotedblleft set addition\textquotedblright%
\ operation between two regions $A$ and $B$:%
\[
A+B\equiv\{a+b:a\in A,b\in B\}.
\]
The following relations hold%
\begin{align*}
Q^{0+-}  &  =L^{0+0}+L^{00-},\\
O^{+-+}  &  =L^{+00}+L^{0-0}+L^{00+},
\end{align*}
by using the above definition. Set addition of the same line gives the line
itself, e.g., $L^{+00}+L^{+00}=L^{+00}$. This set equality holds because of
the definition of set addition and the definition of the line. A similar
result also holds for addition of the same quadrant or octant. We define the
set subtraction of two regions $A$ and $B$ as follows:%
\[
A-B\equiv\{a-b:a\in A,b\in B\}.
\]
According to this definition, it follows that $L^{+00}\subseteq L^{+00}%
-L^{+00}=L^{\pm00}$ where $L^{\pm00}$ represents the full line of classical communication.

\section{Summary of Results}

\label{sec:summary}We first provide an accessible overview of the main results
in this article. The interested reader can then delve into later sections of
the article for mathematical details of the proofs.

\subsection{The Unit Resource Capacity Region}

Our first result determines what rates are achievable when there is no noisy
resource---the only resources available are noiseless classical communication,
noiseless quantum communication, and noiseless entanglement. We provide a
three-dimensional \textquotedblleft unit resource\textquotedblright\ capacity
region that lives in a three-dimensional space with points $\left(
C,Q,E\right)  $.

Three important protocols relate the three fundamental noiseless resources.
These protocols are teleportation (TP)\ \cite{BBCJPW93}, super-dense coding
(SD)\ \cite{BW92}, and entanglement distribution (ED)\ \cite{DHW03}. We can
express these three protocols as resource inequalities. The resource
inequality for teleportation is
\begin{equation}
2[c\rightarrow c]+[qq]\geq\lbrack q\rightarrow q], \label{TP}%
\end{equation}
where the meaning of the resource inequality is as before---the protocol
consumes the resources on the left in order to produce the resource on the
right. Super-dense coding corresponds to the following inequality:%
\begin{equation}
\lbrack q\rightarrow q]+[qq]\geq2[c\rightarrow c], \label{SD}%
\end{equation}
and entanglement distribution is as follows:%
\begin{equation}
\lbrack q\rightarrow q]\geq\lbrack qq]. \label{ED}%
\end{equation}
A sender implements ED\ by transmitting half of a locally prepared Bell state
$\left\vert \Phi^{+}\right\rangle $ through a noiseless qubit channel.

In any trade-off problem, we have the achievable rate region and the capacity
region. The \textit{achievable rate region} is the set of all rate triples
that one can achieve with a specific, known protocol. The \textit{capacity
region} divides the line between what is physically achievable and what is
not---there is no method to achieve any point outside the capacity region. We
define it with respect to a given quantum information processing task. In our
development below, we consider the achievable rate region and the capacity
region of the three unit resources of noiseless classical communication,
noiseless quantum communication, and noiseless entanglement.

\begin{definition}
Let $\widetilde{\mathcal{C}}_{\text{U}}$ denote the unit resource achievable
rate region. It consists of all the rate triples $(C,Q,E)$ obtainable from
linear combinations of the above protocols: TP, SD, and ED.
\end{definition}

The development in Section~\ref{sec:unit-resource} below demonstrates that the
achievable rate region $\widetilde{\mathcal{C}}_{\text{U}}$ in the above
definition is equivalent to all rate triples satisfying the following
inequalities:%
\begin{equation}
C+Q+E\leq0,\ \ \ \ Q+E\leq0,\ \ \ \ C+2Q\leq0.\label{utriple}%
\end{equation}

\begin{definition}
The unit resource capacity region $\mathcal{C}_{\text{U}}$ is the closure of
the set of all points $(C,Q,E)$ in the $C,Q,E$ space satisfying the following
resource inequality:%
\begin{equation}
0\geq C[c\rightarrow c]+Q[q\rightarrow q]+E[qq].
\label{eq:unit-capacity-region}%
\end{equation}

\end{definition}

The above notation may seem confusing at first glance until we establish the
convention that a resource with a negative rate implicitly belongs on the
left-hand side of the resource inequality.

Theorem~\ref{ut} below is our first main result, giving the optimal
three-dimensional capacity region for the three unit resources.

\begin{theorem}
\label{ut} The unit resource capacity region $\mathcal{C}_{\text{U}}$ is
equivalent to the unit resource achievable rate region $\widetilde
{\mathcal{C}}_{\text{U}}$:%
\[
\mathcal{C}_{\text{U}}=\widetilde{\mathcal{C}}_{\text{U}}.
\]

\end{theorem}

The complete proof is in Section~\ref{sec:unit-resource}. It involves several
proofs by contradiction that apply to each octant of the $\left(
C,Q,E\right)  $ space. It exploits two postulates: 1)\ ebits alone cannot
generate cbits or qubits and 2)\ cbits alone cannot generate ebits or qubits.

\subsection{Direct Static Capacity Region}

Our second result provides a solution to the scenario when a noisy static
resource is available in addition to the three noiseless resources. We
determine a three-dimensional \textquotedblleft direct
static\textquotedblright\ capacity region that gives multi-letter formulas for
the full trade-off between the three fundamental noiseless resources.

\begin{definition}
The direct static capacity region $\mathcal{C}_{\text{DS}}(\rho^{AB})$ of a
noisy bipartite state $\rho^{AB}$ is a three-dimensional region in the
$(C,Q,E)$ space. It is the closure of the set of all points $(C,Q,E)$
satisfying the following resource inequality:%
\begin{equation}
\langle\rho^{AB}\rangle\geq C[c\rightarrow c]+Q[q\rightarrow q]+E[qq].
\label{eq:direct-static-cap-def}%
\end{equation}

\end{definition}

The rates $C$, $Q$, and $E$ can either be negative or positive with the same
interpretation as in the previous section.

We first introduce a new protocol that proves to be useful in determining the
achievable rate region for the static case. We name this protocol
\textquotedblleft classically-assisted quantum state
redistribution.\textquotedblright

\begin{lemma}
\label{thm:CAM-RI}The following \textquotedblleft classically-assisted quantum
state redistribution\textquotedblright\ resource inequality holds%
\begin{multline}
\langle\rho^{AB}\rangle+\frac{1}{2}I(A^{\prime};E|E^{\prime}X)_{\sigma
}[q\rightarrow q]+I(X;E|B)_{\sigma}[c\rightarrow c]\label{GMP}\\
\geq\frac{1}{2}\left(  I(A^{\prime};B|X)_{\sigma}-I(A^{\prime};E^{\prime
}|X)_{\sigma}\right)  [qq]
\end{multline}
for a static resource $\rho^{AB}$ and for any remote instrument $\mathcal{T}%
^{A\rightarrow A^{\prime}X}$. In the above resource inequality, the state
$\sigma^{XA^{\prime}BEE^{\prime}}$ is defined by%
\begin{equation}
\sigma^{XA^{\prime}BEE^{\prime}}\equiv\widetilde{\mathcal{T}}(\psi
^{ABE}),\label{eq:instrument-state}%
\end{equation}
where $\left\vert \psi\right\rangle \left\langle \psi\right\vert ^{ABE}$ is
some purification of $\rho^{AB}$ and $\widetilde{\mathcal{T}}^{A\rightarrow
A^{\prime}E^{\prime}X}$ is an extension of $\mathcal{T}^{A\rightarrow
A^{\prime}X}$.
\end{lemma}

The above quantities $I(A^{\prime};E|E^{\prime}X)_{\sigma}$, $I(X;E|B)_{\sigma
}$, and $I(A^{\prime};B|X)_{\sigma}-I(A^{\prime};E^{\prime}|X)_{\sigma}$ are
entropic quantities that are taken with respect to the state $\sigma
^{XA^{\prime}BEE^{\prime}}$. These quantities give the rates of resource
consumption or generation in the above protocol. We refer the reader to
Refs.~\cite{DHW05RI,HW08GFP} for definitions of the above entropic quantities
and the definition of a \textquotedblleft quantum
instrument.\textquotedblright\ The above protocol generalizes the mother
protocol \cite{DHW05RI}, noisy teleportation \cite{DHW05RI}, noisy super-dense
coding \cite{DHW05RI}, the entanglement distillation protocol \cite{BDSW96},
and the grandmother protocol \cite{DHW05RI}.

\begin{definition}
\label{def:CASR-achievable}The classically-assisted state redistribution
\textquotedblleft one-shot\textquotedblright\ achievable rate region
$\widetilde{\mathcal{C}}_{\text{CASR}}^{(1)}(\rho^{AB})$ is as follows:%
\[
\widetilde{\mathcal{C}}_{\text{CASR}}^{(1)}(\rho^{AB})\equiv\bigcup
_{\widetilde{\mathcal{T}}}\left(
\begin{array}
[c]{c}%
-I(X;E|B)_{\sigma},-\frac{1}{2}I(A^{\prime};E|E^{\prime}X)_{\sigma},\\
\frac{1}{2}(I(A^{\prime};B|X)_{\sigma}-I(A^{\prime};E^{\prime}|X)_{\sigma})
\end{array}
\right)  ,
\]
where $\sigma$ is defined as above and the union is over all instruments
$\widetilde{\mathcal{T}}$. The classically-assisted state redistribution
achievable rate region $\widetilde{\mathcal{C}}_{\text{CASR}}(\rho^{AB})$ is
the following multi-letter regularization of the one-shot region:%
\begin{equation}
\widetilde{\mathcal{C}}_{\text{CASR}}(\rho^{AB})\equiv\overline{\bigcup
_{k=1}^{\infty}\frac{1}{k}\widetilde{\mathcal{C}}_{\text{CASR}}^{(1)}%
((\rho^{AB})^{\otimes k})}. \label{DSG}%
\end{equation}

\end{definition}

Below we state our second main result, the direct static capacity theorem.

\begin{theorem}
\label{thm:direct-static}The direct static capacity region $\mathcal{C}%
_{\text{DS}}(\rho^{AB})$ is equivalent to the direct static achievable rate
region $\widetilde{\mathcal{C}}_{\emph{DS}}(\rho^{AB})$:%
\[
\mathcal{C}_{\text{DS}}(\rho^{AB})=\widetilde{\mathcal{C}}_{\text{DS}}%
(\rho^{AB}).
\]
The direct static achievable rate region $\widetilde{\mathcal{C}}_{\text{DS}%
}(\rho^{AB})$ is the set addition of the classically-assisted state
redistribution achievable rate region $\widetilde{\mathcal{C}}_{\text{CASR}}$
and the unit resource achievable rate region $\widetilde{\mathcal{C}%
}_{\text{U}}$:%
\begin{equation}
\widetilde{\mathcal{C}}_{\text{DS}}(\rho^{AB})\equiv\widetilde{\mathcal{C}%
}_{\text{CASR}}(\rho^{AB})+\widetilde{\mathcal{C}}_{\text{U}}.
\label{eq:direct-static-ach-def}%
\end{equation}

\end{theorem}

The complete proof is in Section~\ref{sec:static}. The meaning of the theorem
is that it is possible to obtain all achievable points in the direct-static
capacity region by combining only four protocols:\ classically-assisted state
redistribution, SD, TP, and ED.

\subsection{Direct dynamic trade-off}

Our third result is a three-dimensional \textquotedblleft direct
dynamic\textquotedblright\ capacity region that gives the full trade-off
between the three fundamental noiseless resources when a noisy dynamic
resource is available.

\begin{definition}
The direct dynamic capacity region $\mathcal{C}_{\text{DD}}\left(
\mathcal{N}\right)  $ of a noisy channel $\mathcal{N}^{A^{\prime}\rightarrow
B}$ is a three-dimensional region in the $(C,Q,E)$ space defined by the
closure of the set of all points $(C,Q,E)$ satisfying the following resource
inequality:%
\begin{equation}
\langle\mathcal{N}\rangle\geq C[c\rightarrow c]+Q[q\rightarrow q]+E[qq].
\label{eq:direct-dynamic-cap-def}%
\end{equation}

\end{definition}

We first recall a few theorems concerning the classically-enhanced father
protocol \cite{HW08GFP} because this protocol proves useful in determining the
achievable rate region for the dynamic case. Briefly, the classically-enhanced
father protocol gives a way to transmit classical and quantum information over
an entanglement-assisted quantum channel.

\begin{lemma}
\label{gf}The following classically-enhanced father resource inequality holds
\begin{multline}
\langle\mathcal{N}\rangle+\frac{1}{2}I(A;E|X)_{\sigma}[qq]\label{GFP}\\
\geq\frac{1}{2}I(A;B|X)_{\sigma}[q\rightarrow q]+I(X;B)_{\sigma}[c\rightarrow
c],
\end{multline}
for a noisy dynamic resource $\mathcal{N}^{A^{\prime}\rightarrow B}$. In the
above resource inequality, the state $\sigma^{XABE}$ is defined as follows%
\begin{equation}
\sigma^{XABE}\equiv\sum_{x}p\left(  x\right)  \left\vert x\right\rangle
\left\langle x\right\vert ^{X}\otimes U_{\mathcal{N}}(\psi_{x}^{AA^{\prime}}),
\label{DD_sigma}%
\end{equation}
where the states $\psi_{x}^{AA^{\prime}}$ are pure and $U_{\mathcal{N}%
}^{A^{\prime}\rightarrow BE}$ is an isometric extension of $\mathcal{N}$.
\end{lemma}

The classically-enhanced father protocol generalizes the father protocol
\cite{DHW05RI}, classically-enhanced quantum communication \cite{DS03},
entanglement-assisted classical communication \cite{arx2004shor}, classical
communication \cite{Hol98,SW97}, and quantum communication
\cite{Lloyd96,Shor02,Devetak03}.

\begin{definition}
The \textquotedblleft one-shot\textquotedblright\ classically-enhanced father
achievable rate region $\widetilde{\mathcal{C}}_{\text{CEF}}^{(1)}%
(\mathcal{N})$ is as follows:%
\begin{align*}
&  \widetilde{\mathcal{C}}_{\text{CEF}}^{(1)}(\mathcal{N})\\
&  \equiv\bigcup_{\sigma}\left(  I(X,B)_{\sigma},\frac{1}{2}I(A;B|X)_{\sigma
},-\frac{1}{2}I(A;E|X)_{\sigma}\right)  ,
\end{align*}
where $\sigma$ is defined in (\ref{DD_sigma}). The classically-enhanced father
achievable rate region $\widetilde{\mathcal{C}}_{\text{CEF}}(\mathcal{N})$ is
the following multi-letter regularization of the one-shot region:%
\[
\widetilde{\mathcal{C}}_{\text{CEF}}(\mathcal{N})=\overline{\bigcup
_{k=1}^{\infty}\frac{1}{k}\widetilde{\mathcal{C}}_{\text{CEF}}^{(1)}%
(\mathcal{N}^{\otimes k})}.
\]

\end{definition}

We now state our third main result, the direct dynamic capacity theorem.

\begin{theorem}
\label{thm:direct-dynamic-cap}The direct dynamic capacity region
$\mathcal{C}_{\text{DD}}(\mathcal{N})$ is equivalent to the direct dynamic
achievable rate region $\widetilde{\mathcal{C}}_{\text{DD}}(\mathcal{N})$:%
\[
\mathcal{C}_{\text{DD}}(\mathcal{N})=\widetilde{\mathcal{C}}_{\text{DD}%
}(\mathcal{N}).
\]
The direct dynamic achievable rate region $\widetilde{\mathcal{C}}_{\text{DD}%
}(\mathcal{N})$ is the set addition of the classically-enhanced father
achievable rate region and the unit resource achievable rate region:%
\begin{equation}
\widetilde{\mathcal{C}}_{\text{DD}}(\mathcal{N})\equiv\widetilde{\mathcal{C}%
}_{\text{CEF}}(\mathcal{N})+\widetilde{\mathcal{C}}_{\text{U}}. \label{CEF_DD}%
\end{equation}

\end{theorem}

The complete proof is in Section~\ref{sec:dynamic}. The meaning of the theorem
is that it is possible to obtain all achievable points in the direct-dynamic
capacity region by combining only four protocols:\ the classically-enhanced
father protocol, super-dense coding, teleportation, and entanglement distribution.

\section{The triple trade-off between unit resources}

\label{sec:unit-resource}We now consider what rates are achievable when there
is no noisy resource---the only resources available are noiseless classical
communication, noiseless quantum communication, and noiseless entanglement. We
prove our first main result: Theorem~\ref{ut}. Recall that this theorem states
that the three-dimensional \textquotedblleft unit resource\textquotedblright%
\ achievable rate region, involving the three fundamental noiseless resources,
is equivalent to the unit resource capacity region.

In the unit resource capacity theorem, we exploit the following geometric
objects that lie in the $(C,Q,E)$ space:

\begin{enumerate}
\item Teleportation is the point $(-2,1,-1)$. The \textquotedblleft line of
teleportation\textquotedblright\ $L_{\text{TP}}$ is the following set of
points:
\begin{equation}
L_{\text{TP}}\equiv\left\{  \alpha\left(  -2,1,-1\right)  :\alpha
\geq0\right\}  . \label{eq:line-TP}%
\end{equation}

\item Super-dense coding is the point $(2,-1,-1)$. The \textquotedblleft line
of super-dense coding\textquotedblright\ $L_{\text{SD}}$ is the following set
of points:
\begin{equation}
L_{\text{SD}}\equiv\left\{  \beta\left(  2,-1,-1\right)  :\beta\geq0\right\}
. \label{eq:line-SD}%
\end{equation}

\item Entanglement distribution is the point $(0,-1,1)$. The \textquotedblleft
line of entanglement distribution\textquotedblright\ $L_{\text{ED}}$ is the
following set of points:
\begin{equation}
L_{\text{ED}}\equiv\left\{  \gamma\left(  0,-1,1\right)  :\gamma\geq0\right\}
. \label{eq:line-ED}%
\end{equation}

\end{enumerate}

Let $\widetilde{\mathcal{C}}_{\text{U}}$ denote the unit resource achievable
rate region. It consists of all linear combinations of the above protocols:%
\begin{equation}
\widetilde{\mathcal{C}}_{\text{U}}\equiv L_{\text{TP}}+L_{\text{SD}%
}+L_{\text{ED}}. \label{ut_C}%
\end{equation}
The following matrix equation gives all achievable triples $(C,Q,E)$ in
$\widetilde{\mathcal{C}}_{\text{U}}$:%
\[%
\begin{bmatrix}
C\\
Q\\
E
\end{bmatrix}
=%
\begin{bmatrix}
-2 & 2 & 0\\
1 & -1 & -1\\
-1 & -1 & 1
\end{bmatrix}%
\begin{bmatrix}
\alpha\\
\beta\\
\gamma
\end{bmatrix}
,
\]
where $\alpha,\beta,\gamma\geq0$. We can rewrite the above equation with its
matrix inverse:%
\[%
\begin{bmatrix}
\alpha\\
\beta\\
\gamma
\end{bmatrix}
=%
\begin{bmatrix}
-1/2 & -1/2 & -1/2\\
0 & -1/2 & -1/2\\
-1/2 & -1 & 0
\end{bmatrix}%
\begin{bmatrix}
C\\
Q\\
E
\end{bmatrix}
,
\]
in order to express the coefficients $\alpha$, $\beta$, and $\gamma$ as a
function of the rate triples $(C,Q,E)$. The restriction of non-negativity of
$\alpha$, $\beta$, and $\gamma$ gives the following restriction on the
achievable rate triples $(C,Q,E)$:%
\begin{align}
C+Q+E  &  \leq0,\label{utriple1}\\
Q+E  &  \leq0,\label{utriple2}\\
C+2Q  &  \leq0. \label{utriple3}%
\end{align}
The above result implies that the achievable rate region $\widetilde
{\mathcal{C}}_{\text{U}}$ in (\ref{ut_C}) is equivalent to all rate triples
satisfying (\ref{utriple1}-\ref{utriple3}). Figure~\ref{fig:Unit} displays the
full unit resource achievable rate region. \begin{figure}[ptb]
\begin{center}
\includegraphics[width=3.5in]{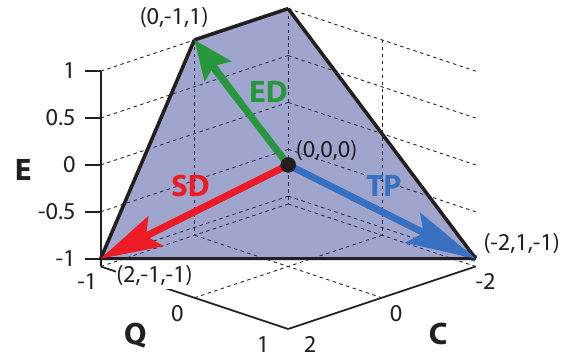}
\end{center}
\caption{(Color online) The above figure depicts the unit resource achievable
region $\widetilde{\mathcal{C}}_{\text{U}}$. The blue, red, and green lines
are the respective lines of teleportation, super-dense coding, and
entanglement distribution defined in (\ref{eq:line-TP}-\ref{eq:line-ED}).
These three lines and their set addition bound the unit resource capacity
region.}%
\label{fig:Unit}%
\end{figure}

Proving Theorem~\ref{ut} involves two steps, traditionally called the
\textit{direct coding theorem} and the \textit{converse theorem}. For this
case, the \textit{direct coding theorem} establishes that the achievable rate
region $\widetilde{\mathcal{C}}_{\text{U}}$ belongs to the capacity region
$\mathcal{C}_{\text{U}}$:%
\[
\widetilde{\mathcal{C}}_{\text{U}}\subseteq\mathcal{C}_{\text{U}}.
\]
The \textit{converse theorem}, on the other hand, establishes the other
inclusion:%
\[
\mathcal{C}_{\text{U}}\subseteq\widetilde{\mathcal{C}}_{\text{U}}.
\]

\subsection{Proof of the Direct Coding Theorem}

The result of the direct coding theorem, that $\widetilde{\mathcal{C}%
}_{\text{U}}\subseteq\mathcal{C}_{\text{U}}$, is immediate from the definition
in (\ref{ut_C})\ of the unit resource achievable rate region $\widetilde
{\mathcal{C}}_{\text{U}}$, the definition in (\ref{eq:unit-capacity-region})
of the unit resource capacity region $\mathcal{C}_{\text{U}}$, and the theory
of resource inequalities in Refs.~\cite{DHW03,DHW05RI}.

\subsection{Proof of the Converse Theorem}

We employ the definition of $\widetilde{\mathcal{C}}_{\text{U}}$ in
(\ref{ut_C}) and consider the eight octants of the $(C,Q,E)$ space
individually in order to prove the converse theorem (that $\mathcal{C}%
_{\text{U}}\subseteq\widetilde{\mathcal{C}}_{\text{U}}$). Let $(\pm,\pm,\pm)$
denote labels for the eight different octants.

It is possible to demonstrate the optimality of each of these three protocols
individually with a contradiction argument (e.g., the contradiction argument
for teleportation is in Ref.~\cite{BBCJPW93}). However, in the converse proof
of Theorem~\ref{ut}, we show that a mixed strategy combining these three
noiseless protocols is optimal.

We accept the following two postulates and exploit them in order to prove the converse:

\begin{enumerate}
\item Entanglement alone cannot generate classical communication or quantum
communication or both.

\item Classical communication alone cannot generate entanglement or quantum
communication or both.
\end{enumerate}

$\boldsymbol{(+,+,+)}$. This octant of $\mathcal{C}_{\text{U}}$ is empty
because a sender and receiver require some resources to implement classical
communication, quantum communication, and entanglement. (They cannot generate
a noiseless resource from nothing!)

$\boldsymbol{(+,+,-)}$. This octant of $\mathcal{C}_{\text{U}}$ is empty
because entanglement alone cannot generate either classical communication or
quantum communication or both.

$\boldsymbol{(+,-,+)}$. The task for this octant is to generate a noiseless
classical channel of $C$ bits and $E$ ebits of entanglement using $|Q|$ qubits
of quantum communication. We thus consider all points of the form $\left(
C,Q,E\right)  $ where $C\geq0$, $Q\leq0$, and $E\geq0$. It suffices to prove
the following inequality:%
\begin{equation}
C+E\leq\left\vert Q\right\vert , \label{pnp}%
\end{equation}
because combining (\ref{pnp}) with $C\geq0$ and $E\geq0$ implies
(\ref{utriple1}-\ref{utriple3}). The achievability of $(C,-|Q|,E)$ implies the
achievability of the point $(C+2E,-|Q|-E,0)$, because we can consume all of
the entanglement with super-dense coding (\ref{SD}). This new point implies
that there is a protocol that consumes $\left\vert Q\right\vert +E$ noiseless
qubit channels to send $C+2E$ classical bits. The following bound then applies%
\[
C+2E\leq\left\vert Q\right\vert +E,
\]
because the Holevo bound \cite{NC00}\ states that we can send only one
classical bit per qubit. The bound in (\ref{pnp}) then follows.

$\boldsymbol{(+,-,-)}$. The task for this octant is to simulate a classical
channel of size $C$ bits using $|Q|$ qubits of quantum communication and $|E|$
ebits of entanglement. We consider all points of the form $\left(
C,Q,E\right)  $ where $C\geq0$, $Q\leq0$, and $E\leq0$. It suffices to prove
the following inequalities:%
\begin{align}
C  &  \leq2|Q|,\label{pnn1}\\
C  &  \leq|Q|+|E|, \label{pnn2}%
\end{align}
because combining (\ref{pnn1}-\ref{pnn2}) with $C\geq0$ implies
(\ref{utriple1}-\ref{utriple3}). The achievability of $(C,-|Q|,-|E|)$ implies
the achievability of $(0,-|Q|+C/2,-|E|-C/2)$, because we can consume all of
the classical communication with teleportation (\ref{TP}). The following bound
applies (quantum communication cannot be positive)%
\[
-\left\vert Q\right\vert +C/2\leq0,
\]
because entanglement alone cannot generate quantum communication. The bound in
(\ref{pnn1}) then follows from the above bound. The achievability of
$(C,-|Q|,-|E|)$ implies the achievability of $(C,-|Q|-|E|,0)$ because we can
consume an extra $\left\vert E\right\vert $ qubit channels with entanglement
distribution (\ref{ED}). The bound in (\ref{pnn2})\ then applies by the same
Holevo bound argument as in the previous octant.

$\boldsymbol{(-,+,+)}$. This octant of $\mathcal{C}_{\text{U}}$ is empty
because classical communication alone cannot generate either quantum
communication or entanglement or both.

$\boldsymbol{(-,+,-)}$. The task for this octant is to simulate a quantum
channel of size $Q$ qubits using $|E|$ ebits of entanglement and $|C|$ bits of
classical communication. We consider all points of the form $\left(
C,Q,E\right)  $ where $C\leq0$, $Q\geq0$, and $E\leq0$. It suffices to prove
the following inequalities:%
\begin{align}
Q  &  \leq\left\vert E\right\vert ,\label{npn1}\\
2Q  &  \leq\left\vert C\right\vert , \label{npn2}%
\end{align}
because combining them with $C\leq0$ implies (\ref{utriple1}-\ref{utriple3}).
The achievability of the point $(-|C|,Q,-|E|)$ implies the achievability of
the point $(-|C|,0,Q-|E|)$, because we can consume all of the quantum
communication for entanglement distribution (\ref{ED}). The following bound
applies (entanglement cannot be positive)%
\[
Q-\left\vert E\right\vert \leq0,
\]
because classical communication alone cannot generate entanglement. The bound
in (\ref{npn1}) follows from the above bound. The achievability of the point
$(-|C|,Q,-|E|)$ implies the achievability of the point $(-|C|+2Q,0,-Q-|E|)$,
because we can consume all of the quantum communication for super-dense coding
(\ref{SD}). The following bound applies (classical communication cannot be
positive)%
\[
-|C|+2Q\leq0,
\]
because entanglement alone cannot create classical communication. The bound in
(\ref{npn2}) follows from the above bound.

$\boldsymbol{(-,-,+)}$. The task for this octant is to create $E$ ebits of
entanglement using $|Q|$ qubits of quantum communication and $|C|$ bits of
classical communication. We consider all points of the form $\left(
C,Q,E\right)  $ where $C\leq0$, $Q\leq0$, and $E\geq0$. It suffices to prove
the following inequality:%
\begin{equation}
E\leq\left\vert Q\right\vert , \label{nnp}%
\end{equation}
because combining it with $Q\leq0$ and $C\leq0$ implies (\ref{utriple1}%
-\ref{utriple3}). The achievability of $(-|C|,-|Q|,E)$ implies the
achievability of $(-|C|-2E,-|Q|+E,0)$, because we can consume all of the
entanglement with teleportation (\ref{TP}). The following bound applies
(quantum communication cannot be positive)%
\[
-|Q|+E\leq0,
\]
because classical communication alone cannot generate quantum communication.
The bound in (\ref{nnp}) follows from the above bound.

$\boldsymbol{(-,-,-)}$. $\widetilde{\mathcal{C}}_{\text{U}}$ completely
contains this octant.

\section{Classically-Assisted Quantum State Redistribution}

Before discussing the direct static trade-off, we overview the
\textit{classically-assisted quantum state redistribution protocol}~introduced
in Section~\ref{sec:summary}. This protocol generates entanglement with the
help of classical communication, quantum communication, and a noisy bipartite
state. It employs techniques from Winter's instrument compression theorem
\cite{Winter01a}, Devetak and Winter's \textquotedblleft classical compression
with quantum side information\textquotedblright\ theorem \cite{DW02}, and the
quantum state redistribution protocol
\cite{devetak:230501,DY06QI,YD07QI,ye:030302,O08}. Thus, we do not give a
full, detailed proof of Lemma~\ref{thm:CAM-RI} (the coding theorem for the
protocol), but instead resort to the resource inequality framework
\cite{DHW05RI}\ for a simple proof of the coding theorem. A full exposition of
this protocol will appear in Ref.~\cite{HW10}. For now, our simple proof of
Lemma~\ref{thm:CAM-RI} appears below.

\begin{proof}
[Proof of Lemma~\ref{thm:CAM-RI}]Consider a state $\rho^{AB}$ shared between
Alice and Bob. Let $\psi^{ABE}$ denote the purification of this state. A
remote instrument $T^{A\rightarrow A^{\prime}X_{B}E^{\prime}}$ acting on this
state produces the state $\sigma^{X_{B}A^{\prime}BEE^{\prime}}$ where%
\[
\sigma^{X_{B}A^{\prime}BEE^{\prime}}\equiv T^{A\rightarrow A^{\prime}%
X_{B}E^{\prime}}\left(  \psi^{ABE}\right)  .
\]
There exists a protocol, \textit{instrument compression with quantum side
information} (ICQSI), that exploits the techniques from
Refs.~\cite{Winter01a,DW02}. It implements the following resource inequality:%
\begin{multline*}
\left\langle \rho^{AB}\right\rangle +I\left(  X_{B};E|B\right)  _{\sigma
}\left[  c\rightarrow c\right]  +H\left(  X_{B}|BE\right)  _{\sigma}\left[
cc\right]  \\
\geq\langle\overline{\Delta}^{X\rightarrow X_{A}X_{B}}\circ T:\rho^{A}\rangle,
\end{multline*}
where $\left[  cc\right]  $ represents the resource of one bit of common
randomness and $\overline{\Delta}^{X\rightarrow X_{A}X_{B}}$ denotes a
classical channel that transmits classical information to Alice and Bob. ICQSI
is similar to Winter's instrument compression protocol in the sense that Alice
and Bob are exploiting classical communication and common randomness to
simulate the action of a quantum instrument, but the difference is that Alice
does not need to send as much classical information as Winter's instrument
compression protocol. Bob can exploit his quantum side information to learn
something about the classical information that Alice is transmitting. The
static version of the HSW\ coding theorem from Ref.~\cite{DW02}\ shows that
Bob can learn $nI\left(  X_{B};B\right)  $ bits about the classical
information that Alice is transmitting, so that she does not have to transmit
the full $nI\left(  X_{B};EB\right)  _{\sigma}$ bits required by the
instrument compression protocol, but instead transmits only $nI\left(
X_{B};E|B\right)  _{\sigma}$ classical bits. It then follows that%
\begin{multline*}
\left\langle \rho^{AB}\right\rangle +I\left(  X_{B};E|B\right)  _{\sigma
}\left[  c\rightarrow c\right]  +H\left(  X_{B}|BE\right)  _{\sigma}\left[
cc\right]  \\
\geq\langle\overline{\Delta}^{X\rightarrow X_{A}X_{B}}(\sigma^{XA^{\prime
}BEE^{\prime}})\rangle,
\end{multline*}
because simulating the quantum instrument is the same as actually performing it.

We now apply the quantum state redistribution protocol discussed in
Refs.~\cite{devetak:230501,DY06QI,YD07QI,ye:030302,O08}. This protocol is
useful here because it makes efficient use of both systems $A^{\prime}$ and
$E^{\prime}$ that Alice possesses after the action of the instrument. We
specifically apply the \textquotedblleft reversed\textquotedblright\ version
of state redistribution outlined on the right hand side of Figure~3 of
Ref.~\cite{DY06QI}, with the substitutions $R\leftrightarrow E$,
$B\leftrightarrow E^{\prime}$, $C\leftrightarrow A^{\prime}$,
$A\leftrightarrow B$. Finally, we can also apply Theorem~4.12 from
Ref.~\cite{DHW05RI}, that shows how convex combinations of static resources
are related to conditioning on classical variables, to get the following
resource inequality:%
\begin{align*}
&  \langle\overline{\Delta}^{X\rightarrow X_{A}X_{B}}(\sigma^{X_{B}A^{\prime}%
})\rangle+\frac{1}{2}I\left(  A^{\prime};E|E^{\prime}X_{B}\right)  _{\sigma
}\left[  q\rightarrow q\right] \\
&  \geq\frac{1}{2}\left(  I\left(  A^{\prime};B|X_{B}\right)  _{\sigma
}-I\left(  A^{\prime};E|X_{B}\right)  _{\sigma}\right)  \left[  qq\right]  .
\end{align*}
Combining the above two resource inequalities gives the following resource
inequality:%
\begin{multline*}
\left\langle \rho^{AB}\right\rangle +I\left(  X_{B};E|B\right)  _{\sigma
}\left[  c\rightarrow c\right]  +H\left(  X_{B}|BE\right)  _{\sigma}\left[
cc\right] \\
+\frac{1}{2}I\left(  A^{\prime};E|E^{\prime}X_{B}\right)  _{\sigma}\left[
q\rightarrow q\right] \\
\geq\frac{1}{2}\left(  I\left(  A^{\prime};B|X_{B}\right)  _{\sigma}-I\left(
A^{\prime};E^{\prime}|X_{B}\right)  _{\sigma}\right)  \left[  qq\right]  .
\end{multline*}
We can then derandomize the protocol and eliminate the common randomness via
Corollary 4.8 of Ref.~\cite{DHW05RI}\ because the output resource is pure. The
result is the following resource inequality:%
\begin{multline*}
\left\langle \rho^{AB}\right\rangle +I\left(  X_{B};E|B\right)  _{\sigma
}\left[  c\rightarrow c\right]  +\frac{1}{2}I\left(  A^{\prime};E|E^{\prime
}X_{B}\right)  _{\sigma}\left[  q\rightarrow q\right] \\
\geq\frac{1}{2}\left(  I\left(  A^{\prime};B|X_{B}\right)  _{\sigma}-I\left(
A^{\prime};E^{\prime}|X_{B}\right)  _{\sigma}\right)  \left[  qq\right]  .
\end{multline*}

\end{proof}

We obtain the classically-assisted state redistribution \textquotedblleft
one-shot\textquotedblright\ achievable rate region $\widetilde{\mathcal{C}%
}_{\text{CASR}}^{(1)}(\rho^{AB})$ and its regularization $\widetilde
{\mathcal{C}}_{\text{CASR}}(\rho^{AB})$ in
Definition~\ref{def:CASR-achievable}\ by performing the protocol with all
possible instruments and taking the union over the resulting achievable rates.

\begin{corollary}
The following \textquotedblleft grandmother\textquotedblright\ protocol of
Ref.~\cite{DHW05RI} results from combining the classically-assisted state
redistribution resource inequality with entanglement distribution of $\frac
{n}{2}I\left(  A^{\prime};E^{\prime}|X_{B}\right)  _{\sigma}$ extra qubits of
quantum communication:%
\begin{multline*}
\left\langle \rho^{AB}\right\rangle +I\left(  X_{B};E|B\right)  _{\sigma
}\left[  c\rightarrow c\right]  +\frac{1}{2}I\left(  A^{\prime};EE^{\prime
}|X_{B}\right)  _{\sigma}\left[  q\rightarrow q\right] \\
\geq\frac{1}{2}I\left(  A^{\prime};B|X_{B}\right)  _{\sigma}\left[  qq\right]
.
\end{multline*}
The above resource inequality requires slightly less classical communication
than the grandmother protocol from Ref.~\cite{DHW05RI}.
\end{corollary}

\section{Direct Static Trade-off}

\label{sec:static}In this section, we prove Theorem~\ref{thm:direct-static},
the direct static capacity theorem. Recall that this theorem determines what
rates are achievable when a sender and receiver consume a noisy static
resource. The parties additionally consume or generate noiseless classical
communication, noiseless quantum communication, and noiseless entanglement.
The theorem determines the three-dimensional \textquotedblleft direct
static\textquotedblright\ capacity region that gives the full trade-off
between the three fundamental noiseless resources when a noisy static resource
is available.

\subsection{Proof of the Direct Coding Theorem}

The direct coding theorem is the statement that the direct static capacity
region contains the direct static achievable rate region:%
\[
\widetilde{\mathcal{C}}_{\text{DS}}(\rho^{AB})\subseteq\mathcal{C}_{\text{DS}%
}(\rho^{AB}).
\]
It follows directly from combining the classically-assisted state
redistribution resource inequality (\ref{GMP}) with TP (\ref{TP}), SD
(\ref{SD}), and ED (\ref{ED}) and considering the definition of the direct
static achievable rate region in (\ref{eq:direct-static-ach-def})\ and the
definition of the direct static capacity region
in\ (\ref{eq:direct-static-cap-def}).

\subsection{Proof of the Converse Theorem}

The converse theorem is the statement that the direct static achievable rate
region $\widetilde{\mathcal{C}}_{\text{DS}}(\rho^{AB})$ contains the direct
static capacity region $\mathcal{C}_{\text{DS}}(\rho^{AB})$:%
\[
\mathcal{C}_{\text{DS}}(\rho^{AB})\subseteq\widetilde{\mathcal{C}}_{\text{DS}%
}(\rho^{AB}).
\]
In order to prove it, we consider one octant of the $\left(  C,Q,E\right)  $
space at a time and use the notation from Section~\ref{sec:notation}. We omit
writing $\rho^{AB}$ in what follows and instead write $\rho$ to denote the
noisy bipartite state $\rho^{AB}$.

The converse proof exploits relations among the previously known capacity
regions: the classically-assisted state redistribution, the mother
\cite{DHW03,DHW05RI}, noisy super-dense coding \cite{HHHLT01,DHW05RI}, noisy
teleportation \cite{DHW03,DHW05RI}, and entanglement distillation
\cite{DHW05RI}. We illustrate the relation between these protocols in
Figure~\ref{fig:Mother-Side}.

\begin{figure*}[ptb]
\begin{center}
\includegraphics[width=1.0\textwidth]{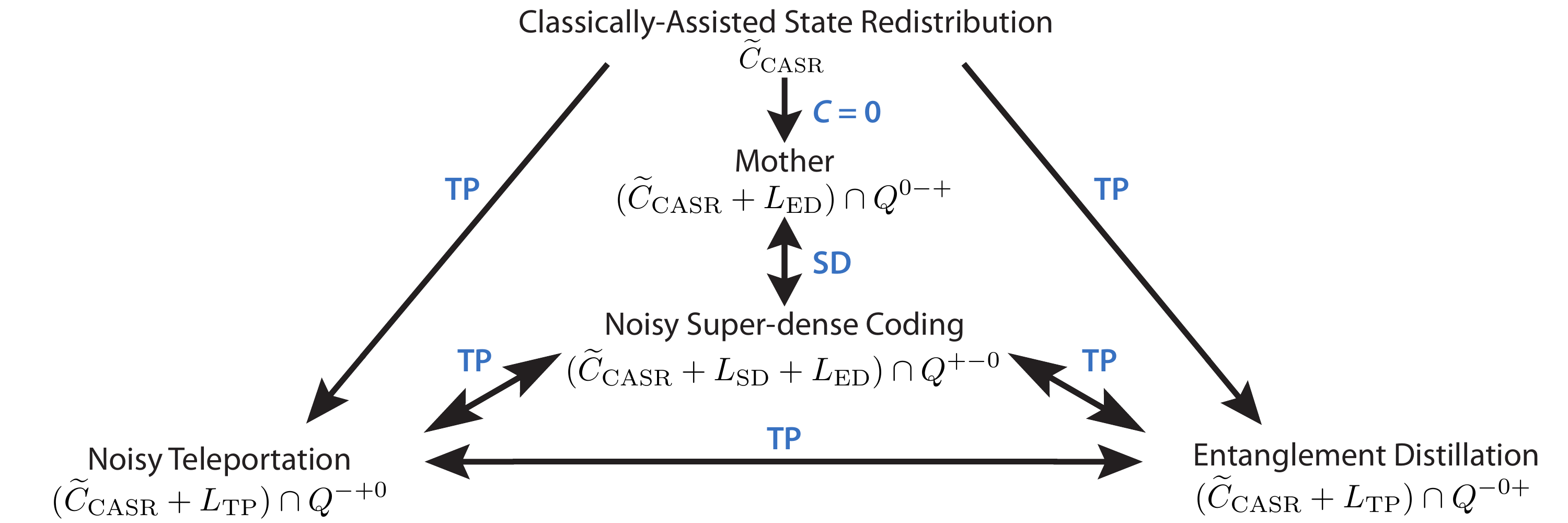}
\end{center}
\caption{The above figure depicts the relations between the capacity regions
for the static case. Any bi-directional arrow represents a bijection between
two capacity regions. Any one-way arrow represents an injection from one
capacity region to another.}%
\label{fig:Mother-Side}%
\end{figure*}

The mother protocol consumes a noisy static resource and noiseless quantum
communication to generate noiseless entanglement. The mother's achievable rate
region $\widetilde{\mathcal{C}}^{0-+}(\rho)$ and capacity region
$\mathcal{C}_{\text{MP}}\left(  \rho\right)  $ lie in the $Q^{0-+}$ quadrant
of the $\left(  C,Q,E\right)  $ space:%
\begin{align}
\widetilde{\mathcal{C}}_{\text{DS}}^{0-+}\left(  \rho\right)   &
=\widetilde{\mathcal{C}}_{\text{DS}}\left(  \rho\right)  \cap Q^{0-+}%
,\label{eq:def-ach-DS-0-+}\\
\mathcal{C}_{\text{DS}}^{0-+}\left(  \rho\right)   &  =\mathcal{C}_{\text{DS}%
}\left(  \rho\right)  \cap Q^{0-+}. \label{eq:def-cap-DS-0-+}%
\end{align}
The above relations establish that the mother's respective regions are special
cases of the direct static achievable rate region $\widetilde{\mathcal{C}%
}_{\text{DS}}\left(  \rho\right)  $\ and the direct static capacity region
$\mathcal{C}_{\text{DS}}\left(  \rho\right)  $. The mother capacity theorem
states that the mother's achievable rate region $\widetilde{\mathcal{C}%
}_{\text{MP}}\left(  \rho\right)  $ is the same as its capacity region:%
\begin{equation}
\widetilde{\mathcal{C}}_{\text{DS}}^{0-+}\left(  \rho\right)  =\mathcal{C}%
_{\text{DS}}^{0-+}\left(  \rho\right)  . \label{eq:mother-capacity}%
\end{equation}
The mother's achievable rate region $\widetilde{\mathcal{C}}_{\text{DS}}%
^{0-+}(\rho)$ is a special case of the classically-assisted state
redistribution achievable rate region $\widetilde{\mathcal{C}}_{\text{CASR}%
}\left(  \rho\right)  $\ (combined with entanglement distribution) where there
is no consumption of classical communication:%
\begin{equation}
\widetilde{\mathcal{C}}_{\text{DS}}^{0-+}\left(  \rho\right)  =(\widetilde
{\mathcal{C}}_{\text{CASR}}\left(  \rho\right)  +L_{\text{ED}})\cap Q^{0-+}.
\label{eq:mother-CAM}%
\end{equation}

The noisy super-dense coding protocol consumes a noisy static resource and
noiseless quantum communication to generate noiseless classical communication.
Its achievable rate region $\widetilde{\mathcal{C}}_{\text{DS}}^{+-0}\left(
\rho\right)  $ and capacity region $\mathcal{C}_{\text{DS}}^{+-0}\left(
\rho\right)  $ lie in the $Q^{+-0}$ quadrant of the $\left(  C,Q,E\right)  $
space:%
\begin{align}
\widetilde{\mathcal{C}}_{\text{DS}}^{+-0}\left(  \rho\right)   &
=\widetilde{\mathcal{C}}_{\text{DS}}\left(  \rho\right)  \cap Q^{+-0}%
,\label{eq:def-ach-DS-+-0}\\
\mathcal{C}_{\text{DS}}^{+-0}\left(  \rho\right)   &  =\mathcal{C}_{\text{DS}%
}\left(  \rho\right)  \cap Q^{+-0}. \label{eq:def-cap-DS-+-0}%
\end{align}
The above relations establish that its respective regions are special cases of
the direct static achievable rate region $\widetilde{\mathcal{C}}_{\text{DS}%
}\left(  \rho\right)  $\ and the direct static capacity region $\mathcal{C}%
_{\text{DS}}\left(  \rho\right)  $. Its capacity theorem states that NSD's
achievable rate region is equivalent to its capacity region:%
\begin{equation}
\widetilde{\mathcal{C}}_{\text{DS}}^{+-0}\left(  \rho\right)  =\mathcal{C}%
_{\text{DS}}^{+-0}\left(  \rho\right)  . \label{eq:NSD-capacity}%
\end{equation}
Its achievable rate region $\widetilde{\mathcal{C}}_{\text{DS}}^{+-0}\left(
\rho\right)  $\ is obtainable from the mother's achievable rate region
$\widetilde{\mathcal{C}}_{\text{DS}}^{0-+}\left(  \rho\right)  $ by combining
it with super-dense coding and keeping the points with zero entanglement:%
\begin{align}
\widetilde{\mathcal{C}}_{\text{DS}}^{+-0}\left(  \rho\right)   &  =\left(
\widetilde{\mathcal{C}}_{\text{DS}}^{0-+}\left(  \rho\right)  +L_{\text{SD}%
}\right)  \cap Q^{+-0},\nonumber\\
&  =\left(  ((\widetilde{\mathcal{C}}_{\text{CASR}}\left(  \rho\right)
+L_{\text{ED}})\cap Q^{0-+})+L_{\text{SD}}\right)  \cap Q^{+-0},
\label{eq:mother-NSD}%
\end{align}
where the second equality comes from (\ref{eq:mother-CAM}).

The noisy teleportation protocol consumes a noisy static resource and
noiseless classical communication to generate noiseless quantum communication.
Its achievable rate region $\widetilde{\mathcal{C}}_{\text{DS}}^{-+0}\left(
\rho\right)  $ and capacity region $\mathcal{C}_{\text{DS}}^{-+0}\left(
\rho\right)  $ lie in the $Q^{-+0}$ quadrant of the $\left(  C,Q,E\right)  $
space:%
\begin{align}
\widetilde{\mathcal{C}}_{\text{DS}}^{-+0}\left(  \rho\right)   &
=\widetilde{\mathcal{C}}_{\text{DS}}\left(  \rho\right)  \cap Q^{-+0}%
,\label{eq:def-ach-DS--+0}\\
\mathcal{C}_{\text{DS}}^{-+0}\left(  \rho\right)   &  =\mathcal{C}_{\text{DS}%
}\left(  \rho\right)  \cap Q^{-+0}. \label{eq:def-cap-DS--+0}%
\end{align}
The above relations establish that its respective regions are special cases of
the direct static achievable rate region $\widetilde{\mathcal{C}}_{\text{DS}%
}\left(  \rho\right)  $\ and the direct static capacity region $\mathcal{C}%
_{\text{DS}}\left(  \rho\right)  $. Its capacity theorem states that its
achievable rate region is equivalent to its capacity region:%
\begin{equation}
\widetilde{\mathcal{C}}_{\text{DS}}^{-+0}\left(  \rho\right)  =\mathcal{C}%
_{\text{DS}}^{-+0}\left(  \rho\right)  . \label{eq:NTP-capacity}%
\end{equation}
Its achievable rate region $\widetilde{\mathcal{C}}_{\text{DS}}^{-+0}\left(
\rho\right)  $\ is obtainable from the classically-assisted state
redistribution's achievable rate region $\widetilde{\mathcal{C}}_{\text{CASR}%
}\left(  \rho\right)  $ by combining it with teleportation and keeping the
points with zero entanglement:%
\begin{equation}
\widetilde{\mathcal{C}}_{\text{DS}}^{-+0}\left(  \rho\right)  =(\widetilde
{\mathcal{C}}_{\text{CASR}}\left(  \rho\right)  +L_{\text{TP}})\cap Q^{-+0}.
\label{eq:mother-NTP}%
\end{equation}

The entanglement distillation protocol consumes a noisy static resource and
noiseless classical communication to generate noiseless entanglement. Its
achievable rate region $\widetilde{\mathcal{C}}_{\text{DS}}^{-0+}(\rho)$ and
capacity region $\mathcal{C}_{\text{DS}}^{-0+}(\rho)$ lie in the $Q^{-0+}$
quadrant of the $(C,Q,E)$ space:%
\begin{align}
\widetilde{\mathcal{C}}_{\text{DS}}^{-0+}(\rho)  &  =\widetilde{\mathcal{C}%
}_{\text{DS}}(\rho)\cap Q^{-0+},\label{eq:def-ach-DS--0+}\\
\mathcal{C}_{\text{DS}}^{-0+}(\rho)  &  =\mathcal{C}_{\text{DS}}(\rho)\cap
Q^{-0+}. \label{eq:def-cap-DS--0+}%
\end{align}
The above relations show that its respective regions are special cases of the
direct static achievable rate region $\widetilde{\mathcal{C}}_{\text{DS}}%
(\rho)$ and the direct static capacity region $\mathcal{C}_{\text{DS}}(\rho)$.
Its capacity theorem states that its achievable rate region is equivalent to
its capacity region:%
\begin{equation}
\widetilde{\mathcal{C}}_{\text{DS}}^{-0+}(\rho)=\mathcal{C}_{\text{DS}}%
^{-0+}(\rho). \label{eq:ED-capacity}%
\end{equation}
Its achievable rate region $\widetilde{\mathcal{C}}_{\text{DS}}^{-0+}(\rho)$
is obtainable from the classically-assisted state redistribution's achievable
rate region $\widetilde{\mathcal{C}}_{\text{CASR}}(\rho)$ by combining it with
teleportation and keeping the points with zero quantum communication:%
\begin{equation}
\widetilde{\mathcal{C}}_{\text{DS}}^{-0+}(\rho)=(\widetilde{\mathcal{C}%
}_{\text{CASR}}(\rho)+L_{\text{TP}})\cap Q^{-0+}. \label{eq:mother-ED}%
\end{equation}

\subsubsection{Entanglement-Generating Octants}

We first consider the four octants with corresponding protocols that generate
entanglement, i.e., those of the form $\left(  \pm,\pm,+\right)  $. The proof
of one octant is trivial and the proof of another appears in Ref.~\cite{HW10}.
The proofs of the remaining two octants are similar to each other---they
follow by consuming all of the entanglement in those octants and resorting to
the previously established converse theorems for two-dimensional quadrants.
Geometrically, the technique is to project all points in an octant into a
quadrant with a known capacity theorem and exploit that converse proof of the
corresponding quadrant.

$\boldsymbol{(+,+,+)}$. This octant is empty because a noisy static resource
alone cannot generate a dynamic resource.

$\boldsymbol{(-,-,+)}$. The full proof of the converse for this octant appears
in Appendix~\ref{sec:converse_--+_static}. There, we obtain the following
bounds on the one-shot, one-instrument capacity region:%
\begin{align*}
E &  \leq I(A\rangle BX)_{\sigma}+\left\vert Q\right\vert ,\\
\left\vert C\right\vert +2\left\vert Q\right\vert  &  \geq I\left(
X;E|B\right)  _{\sigma}+I\left(  A^{\prime};E|E^{\prime}X\right)  _{\sigma}\\
E &  \leq\left\vert C\right\vert +\left\vert Q\right\vert +I\left(  A\rangle
BX\right)  _{\sigma}-I\left(  X;E|B\right)  _{\sigma},
\end{align*}
where the entropies are with respect to the state in
(\ref{eq:instrument-state}). The above set of inequalities is equivalent to a
translation of the unit resource capacity region to the achievable rate triple
of the classically-assisted state redistribution protocol. The result is that
the capacity region for this octant is within an achievable rate region
consisting of all rates achieved by the regularized classically-assisted state
redistribution protocol combined with the unit protocols:%
\[
\ \mathcal{C}_{\text{DS}}^{--+}\left(  \rho\right)  \subseteq\widetilde
{\mathcal{C}}_{\text{CASR}}\left(  \rho\right)  +\widetilde{\mathcal{C}%
}_{\text{U}}.
\]

$\boldsymbol{(+,-,+)}$. This octant exploits the projection technique with
super-dense coding. Let%
\[
\mathcal{C}_{\text{DS}}^{+-+}\left(  \rho\right)  \equiv\mathcal{C}%
_{\text{DS}}\left(  \rho\right)  \cap O^{+-+},
\]
and recall the definition of $\mathcal{C}_{\text{DS}}^{+-0}\left(
\rho\right)  $ in (\ref{eq:def-cap-DS-+-0}). We exploit the line of
super-dense coding $L_{\text{SD}}$ as defined in (\ref{eq:line-SD}). Define
the following maps:%
\begin{align*}
f  &  :S\rightarrow(S+L_{\text{SD}})\cap Q^{+-0},\\
\hat{f}  &  :S\rightarrow(S-L_{\text{SD}})\cap O^{+-+}.
\end{align*}
The map $f$ translates the set $S$ in the super-dense coding direction and
keeps the points that lie on the $Q^{+-0}$ quadrant. The map $\hat{f}$, in a
sense, undoes the effect of $f$ by moving the set $S$ back to the octant
$O^{+-+}$.

The inclusion $\mathcal{C}_{\text{DS}}^{+-+}\left(  \rho\right)  \subseteq
\hat{f}(f(\mathcal{C}_{\text{DS}}^{+-+}\left(  \rho\right)  ))$ holds because%
\begin{align}
&  \mathcal{C}_{\text{DS}}^{+-+}\left(  \rho\right) \nonumber\\
&  =\mathcal{C}_{\text{DS}}^{+-+}\left(  \rho\right)  \cap O^{+-+}\nonumber\\
&  \subseteq(((\mathcal{C}_{\text{DS}}^{+-+}\left(  \rho\right)
+L_{\text{SD}})\cap Q^{+-0})-L_{\text{SD}})\cap O^{+-+}\nonumber\\
&  =(f(\mathcal{C}_{\text{DS}}^{+-+}\left(  \rho\right)  )-L_{\text{SD}})\cap
O^{+-+}\nonumber\\
&  =\hat{f}(f(\mathcal{C}_{\text{DS}}^{+-+}\left(  \rho\right)  )).
\label{DSpnp1}%
\end{align}
The first set equivalence is obvious from the definition of $\mathcal{C}%
_{\text{DS}}^{+-+}\left(  \rho\right)  $. The first inclusion follows from the
following logic.\ Pick any point $a\equiv(C,Q,E)\in\mathcal{C}_{\text{DS}%
}^{+-+}\left(  \rho\right)  \cap O^{+-+}$ and a particular point
$b\equiv\left(  2E,-E,-E\right)  \in L_{\text{SD}}$. It follows that
$a+b=\left(  C+2E,Q-E,0\right)  \in(\mathcal{C}_{\text{DS}}^{+-+}\left(
\rho\right)  +L_{\text{SD}})\cap Q^{+-0}$. We then pick a point $-b=\left(
-2E,E,E\right)  \in-L_{\text{SD}}$. It follows that $a+b-b\in(((\mathcal{C}%
_{\text{DS}}^{+-+}\left(  \rho\right)  +L_{\text{SD}})\cap Q^{+-0}%
)-L_{\text{SD}})\cap O^{+-+}$ and that $a+b-b=(C,Q,E)=a$. The first inclusion
thus holds because every point in $\mathcal{C}_{\text{DS}}^{+-+}\left(
\rho\right)  \cap O^{+-+}$ is in $(((\mathcal{C}_{\text{DS}}^{+-+}\left(
\rho\right)  +L_{\text{SD}})\cap Q^{+-0})-L_{\text{SD}})\cap O^{+-+}$. The
second set equivalence follows from the definition of $f$ and the third set
equivalence follows from the definition of $\hat{f}$.

It is operationally clear that the following inclusion holds
\begin{equation}
f(\mathcal{C}_{\text{DS}}^{+-+}\left(  \rho\right)  )\subseteq\mathcal{C}%
_{\text{DS}}^{+-0}\left(  \rho\right)  , \label{DSpnp2}%
\end{equation}
because the mapping $f$ converts any achievable point $a\in\mathcal{C}%
_{\text{DS}}^{+-+}\left(  \rho\right)  $ to an achievable point in
$\mathcal{C}_{\text{DS}}^{+-0}\left(  \rho\right)  $ by consuming all of the
entanglement at point $a$ with super-dense coding.

The converse proof of the noisy super-dense coding protocol \cite{DHW05RI} is
useful for us:%
\begin{equation}
\mathcal{C}_{\text{DS}}^{+-0}\left(  \rho\right)  \subseteq\widetilde
{\mathcal{C}}_{\text{DS}}^{+-0}\left(  \rho\right)  . \label{DSpnp4}%
\end{equation}
Recall the relation in\ (\ref{eq:mother-NSD}) between its achievable rate
region $\widetilde{\mathcal{C}}_{\text{DS}}^{+-0}\left(  \rho\right)  $ and
the classically-assisted state redistribution achievable rate region
$\widetilde{\mathcal{C}}_{\text{CASR}}\left(  \rho\right)  $. The following
set inclusion holds
\begin{equation}
\widetilde{\mathcal{C}}_{\text{DS}}^{+-0}\left(  \rho\right)  \subseteq
(\widetilde{\mathcal{C}}_{\text{CASR}}\left(  \rho\right)  +L_{\text{SD}})\cap
Q^{+-0}, \label{DSpnp5}%
\end{equation}
by dropping the intersection with $Q^{0-+}$ in (\ref{eq:mother-NSD}).

The inclusion $\hat{f}(\widetilde{\mathcal{C}}_{\text{DS}}^{+-0}\left(
\rho\right)  )\subseteq\widetilde{\mathcal{C}}_{\text{DS}}^{+-+}\left(
\rho\right)  $ holds because%
\begin{align}
&  \hat{f}(\widetilde{\mathcal{C}}_{\text{DS}}^{+-0}\left(  \rho\right)
)\nonumber\\
&  \subseteq(((\widetilde{\mathcal{C}}_{\text{CASR}}\left(  \rho\right)
+L_{\text{SD}})\cap Q^{+-0})-L_{\text{SD}})\cap O^{+-+}\nonumber\\
&  \subseteq((\widetilde{\mathcal{C}}_{\text{CASR}}\left(  \rho\right)
+L_{\text{SD}})-L_{\text{SD}})\cap O^{+-+}\nonumber\\
&  =((\widetilde{\mathcal{C}}_{\text{CASR}}\left(  \rho\right)  +L_{\text{SD}%
})\cap O^{+-+})\cup\nonumber\\
&  \ \ \ \ \ \ \ \ \ \ \ \ \ \ \ \ \ \ \ \ \ \ \ \ \ \ \ \ \ \ ((\widetilde
{\mathcal{C}}_{\text{CASR}}\left(  \rho\right)  -L_{\text{SD}})\cap
O^{+-+})\nonumber\\
&  \subseteq\widetilde{\mathcal{C}}_{\text{DS}}^{+-+}\left(  \rho\right)  .
\label{DSpnp3}%
\end{align}
The first inclusion follows from (\ref{DSpnp5}). The second inclusion follows
by dropping the intersection with $Q^{+-0}$. The second set equivalence
follows because $(\widetilde{\mathcal{C}}_{\text{CASR}}\left(  \rho\right)
+L_{\text{SD}})-L_{\text{SD}}=(\widetilde{\mathcal{C}}_{\text{CASR}}\left(
\rho\right)  +L_{\text{SD}})\cup(\widetilde{\mathcal{C}}_{\text{CASR}}\left(
\rho\right)  -L_{\text{SD}})$, and the last inclusion follows because
$(\widetilde{\mathcal{C}}_{\text{CASR}}\left(  \rho\right)  -L_{\text{SD}%
})\cap O^{+-+}=\left(  0,0,0\right)  $.

Putting (\ref{DSpnp1}), (\ref{DSpnp2}), (\ref{DSpnp4}), and (\ref{DSpnp3})
together, the inclusion $\mathcal{C}_{\text{DS}}^{+-+}\left(  \rho\right)
\subseteq\widetilde{\mathcal{C}}_{\text{DS}}^{+-+}\left(  \rho\right)  $ holds
because%
\begin{multline*}
\mathcal{C}_{\text{DS}}^{+-+}\left(  \rho\right)  \subseteq\hat{f}%
(f(\mathcal{C}_{\text{DS}}^{+-+}\left(  \rho\right)  ))\\
\subseteq\hat{f}(\mathcal{C}_{\text{DS}}^{+-0}\left(  \rho\right)
)\subseteq\hat{f}(\widetilde{\mathcal{C}}_{\text{DS}}^{+-0}\left(
\rho\right)  )\subseteq\widetilde{\mathcal{C}}_{\text{DS}}^{+-+}\left(
\rho\right)  .
\end{multline*}
The above inclusion is the statement of the converse theorem for this octant.

The proof for the other entanglement generating octant $\left(  -,+,+\right)
$ follows similarly to the above octant by exploiting the same projection
technique with teleportation. The proof appears in
Appendix~\ref{sec:-++static}.

\subsubsection{Entanglement-Consuming Octants}

We now consider all the octants with corresponding protocols that consume
entanglement, i.e., those of the form $\left(  \pm,\pm,-\right)  $. The proofs
for two of the octants are trivial and the proofs for the two non-trivial
octants each contain an additivity lemma that shows how to relate their
converse proofs to the converse proofs of a quadrant.

$\boldsymbol{(+,+,-)}$. This octant is empty because a noisy static resource
assisted by noiseless entanglement cannot generate a dynamic resource (the two
static resources cannot generate classical communication or quantum
communication or both).

$\boldsymbol{(-,-,-)}$. $\widetilde{\mathcal{C}}_{\text{DS}}$ completely
contains this octant.

$\boldsymbol{(+,-,-)}$. Define $\Phi^{|E|}$ to be a state of size $|E|$ ebits.
We first need the following lemma.

\begin{lemma}
\label{lem:additivity-+-0}The following inclusion holds%
\[
\mathcal{C}_{\text{DS}}^{+-0}(\rho\otimes\Phi^{|E|})\subseteq\widetilde
{\mathcal{C}}_{\text{DS}}^{+-0}(\rho)+\widetilde{\mathcal{C}}_{\text{DS}%
}^{+-0}(\Phi^{|E|}).
\]

\end{lemma}

\begin{IEEEproof}
Super-dense coding induces a linear bijection $f:\widetilde{\mathcal{C}}_{\text{DS}%
}^{0-+}\left(  \rho\right)  \rightarrow\widetilde{\mathcal{C}}_{\text{DS}}^{+-0}\left(
\rho\right)  $ between the mother achievable region $\widetilde{\mathcal{C}}_{\text{DS}%
}^{0-+}\left(  \rho\right)  $ and the noisy dense coding achievable region
$\widetilde{\mathcal{C}}_{\text{DS}}^{+-0}\left(  \rho\right)  $. The bijection $f$
behaves as follows for every point $(0,Q,E)\in\widetilde{\mathcal{C}}_{\text{DS}}%
^{0-+}\left(  \rho\right)  $:%
\[
f:(0,Q,E)\rightarrow(2E,Q-E,0).
\]
The following relation holds%
\begin{equation}
f(\widetilde{\mathcal{C}}_{\text{DS}}^{0-+}\left(  \rho\right)  )=\widetilde
{\mathcal{C}}_{\text{DS}}^{+-0}\left(  \rho\right)  , \label{DSpnn2}%
\end{equation}
because applying dense coding to the mother resource inequality gives noisy
dense coding \cite{DHW05RI}.
The inclusion $\mathcal{C}_{\text{DS}}^{0-+}(\rho\otimes\Phi^{|E|})\subseteq
f^{-1}(\widetilde{\mathcal{C}}_{\text{DS}}^{+-0}(\rho)+\widetilde{\mathcal{C}}_{\text{DS}}%
^{+-0}(\Phi^{|E|}))$ holds because%
\begin{align*}
\mathcal{C}_{\text{DS}}^{0-+}(\rho\otimes\Phi^{|E|})  &  =\mathcal{C}_{\text{DS}}^{0-+}%
(\rho)+(0,0,E)\\
&  \subseteq \mathcal{C}_{\text{DS}}^{0-+}(\rho)+\mathcal{C}_{\text{DS}}^{0-+}(\Phi^{|E|})\\
&  =\widetilde{\mathcal{C}}_{\text{DS}}^{0-+}(\rho)+\widetilde{\mathcal{C}}_{\text{DS}}^{0-+}%
(\Phi^{|E|})\\
&  =f^{-1}(\widetilde{\mathcal{C}}_{\text{DS}}^{+-0}(\rho))+f^{-1}(\widetilde
{\mathcal{C}}_{\text{DS}}^{+-0}(\Phi^{|E|}))\\
&  =f^{-1}(\widetilde{\mathcal{C}}_{\text{DS}}^{+-0}(\rho)+\widetilde{\mathcal{C}}_{\text{DS}%
}^{+-0}(\Phi^{|E|})).
\end{align*}
The first set equivalence follows because the capacity region of the noisy
resource state $\rho$ combined with a rate $E$ maximally entangled state is
equivalent to a translation of the capacity region of the noisy resource state
$\rho$. The first inclusion follows because the capacity region of a rate $E$
maximally entangled state contains the rate triple $(0,0,E)$. The second set
equivalence follows from the mother capacity theorem in
(\ref{eq:mother-capacity}), the third set equivalence from (\ref{DSpnn2}), and
the last from linearity of the map $f$. The above inclusion implies the
following one:
\[
f(\mathcal{C}_{\text{DS}}^{0-+}(\rho\otimes\Phi^{|E|}))\subseteq\widetilde{C}%
_{\text{DS}}^{+-0}(\rho)+\widetilde{\mathcal{C}}_{\text{DS}}^{+-0}(\Phi^{|E|}).
\]
The lemma follows because%
\begin{align*}
f(\mathcal{C}_{\text{DS}}^{0-+}(\rho\otimes\Phi^{|E|}))  &  =f(\widetilde{\mathcal{C}}_{\text{DS}%
}^{0-+}(\rho\otimes\Phi^{|E|}))\\
&  =\widetilde{\mathcal{C}}_{\text{DS}}^{+-0}(\rho\otimes\Phi^{|E|})\\
&  =\mathcal{C}_{\text{DS}}^{+-0}(\rho\otimes\Phi^{|E|}),
\end{align*}
where we apply the mother capacity theorem in (\ref{eq:mother-capacity}) and
the NSD\ capacity theorem in (\ref{eq:NSD-capacity}).
\end{IEEEproof}

We now begin the converse proof for this octant. Observe that%
\begin{equation}
\widetilde{\mathcal{C}}_{\text{DS}}^{+-0}(\Phi^{|E|})=\widetilde{\mathcal{C}%
}_{\text{U}}^{+-E}. \label{eq:+-0phi=+-EU}%
\end{equation}
Hence for all $E\leq0$,%
\begin{equation}
\mathcal{C}_{\text{DS}}^{+-E}(\rho)=\mathcal{C}_{\text{DS}}^{+-0}(\rho
\otimes\Phi^{|E|})\subseteq\widetilde{\mathcal{C}}_{\text{DS}}^{+-0}%
(\rho)+\widetilde{\mathcal{C}}_{\text{U}}^{+-E}, \label{DSpnn3}%
\end{equation}
where we apply Lemma~\ref{lem:additivity-+-0} and (\ref{eq:+-0phi=+-EU}).
Thus, the inclusion $\mathcal{C}_{\text{DS}}^{+--}(\rho)\subseteq
\widetilde{\mathcal{C}}_{\text{DS}}^{+--}(\rho)$ holds because%
\begin{align*}
\mathcal{C}_{\text{DS}}^{+--}(\rho)  &  =\bigcup_{E\leq0}\mathcal{C}%
_{\text{DS}}^{+-E}(\rho)\\
&  \subseteq\bigcup_{E\leq0}\left(  \widetilde{\mathcal{C}}_{\text{DS}}%
^{+-0}(\rho)+\widetilde{\mathcal{C}}_{\text{U}}^{+-E}\right) \\
&  =(\widetilde{\mathcal{C}}_{\text{DS}}^{+-0}(\rho)+\widetilde{\mathcal{C}%
}_{\text{U}})\cap O^{+--}\\
&  \subseteq(\widetilde{\mathcal{C}}_{\text{CASR}}(\rho)+\widetilde
{\mathcal{C}}_{\text{U}})\cap O^{+--}\\
&  =\widetilde{\mathcal{C}}_{\text{DS}}^{+--}(\rho).
\end{align*}
The first set equivalence holds by definition. The first inclusion follows
from (\ref{DSpnn3}). The second set equivalence follows because $\bigcup
_{E\leq0}\widetilde{\mathcal{C}}_{\text{U}}^{+-E}=\widetilde{\mathcal{C}%
}_{\text{U}}\cap O^{+--}$. The second inclusion holds because $\widetilde
{\mathcal{C}}_{\text{DS}}^{+-0}(\rho)$ is equivalent to noisy dense coding and
the classically-assisted state redistribution combined with the unit resource
region generates noisy dense coding. The above inclusion $\mathcal{C}%
_{\text{DS}}^{+--}(\rho)\subseteq\widetilde{\mathcal{C}}_{\text{DS}}%
^{+--}(\rho)$ is the statement of the converse theorem for this octant.

The proof for the other entanglement consuming octant $\left(  -,+,-\right)  $
follows similarly to the above proof and appears in
Appendix~\ref{sec:-+-static}.

\section{Direct dynamic trade-off}

\label{sec:dynamic}In this section, we prove
Theorem~\ref{thm:direct-dynamic-cap}, the direct dynamic capacity theorem.
Recall that this theorem determines what rates are achievable when a sender
and receiver consume a noisy dynamic resource. They additionally consume or
generate noiseless classical communication, noiseless quantum communication,
and noiseless entanglement. The theorem determines the three-dimensional
\textquotedblleft direct dynamic\textquotedblright\ capacity region, giving
the full trade-off between the three fundamental noiseless resources when a
noisy dynamic resource is available.

\subsection{Proof of the Direct Coding Theorem}

The direct coding theorem is the statement that the direct dynamic capacity
region $\mathcal{C}_{\text{DD}}(\mathcal{N})$\ contains the direct dynamic
achievable rate region $\widetilde{\mathcal{C}}_{\text{DD}}(\mathcal{N})$:%
\[
\widetilde{\mathcal{C}}_{\text{DD}}(\mathcal{N})\subseteq\mathcal{C}%
_{\text{DD}}(\mathcal{N}).
\]
It follows immediately from combining the classically-enhanced father resource
inequality in (\ref{GFP}) with the unit resource inequalities and considering
the definition of $\widetilde{\mathcal{C}}_{\text{DD}}(\mathcal{N})$ in
Theorem~\ref{thm:direct-dynamic-cap} and the definition of $\mathcal{C}%
_{\text{DD}}(\mathcal{N})$ in (\ref{eq:direct-dynamic-cap-def}).

\subsection{Proof of the Converse Theorem}

The statement of the converse theorem is that the direct dynamic achievable
rate region $\widetilde{\mathcal{C}}_{\text{DD}}(\mathcal{N})$ contains the
direct dynamic capacity region:%
\[
\mathcal{C}_{\text{DD}}(\mathcal{N})\subseteq\widetilde{\mathcal{C}%
}_{\text{DD}}(\mathcal{N}).
\]
We consider one octant of the $\left(  C,Q,E\right)  $ space at a time in
order to prove the converse theorem.

The converse theorem exploits the converse proofs of several other protocols:
the father protocol \cite{DHW03,DHW05RI}, classically-enhanced quantum
communication and entanglement generation \cite{DS03}, classically-assisted
quantum communication and entanglement generation \cite{BKN98}, and
entanglement-assisted classical communication \cite{BSST01,Hol01a,arx2004shor}%
. It also exploits some information-theoretic arguments. We briefly review
each of these protocols and their relation to the classically-enhanced father
achievable rate region $\widetilde{\mathcal{C}}_{\text{CEF}}(\mathcal{N})$ in
what follows. Figure~\ref{fig:Father-Side} illustrates the relation between
these protocols.

\begin{figure*}[ptb]
\begin{center}
\includegraphics[width=1.0\textwidth]{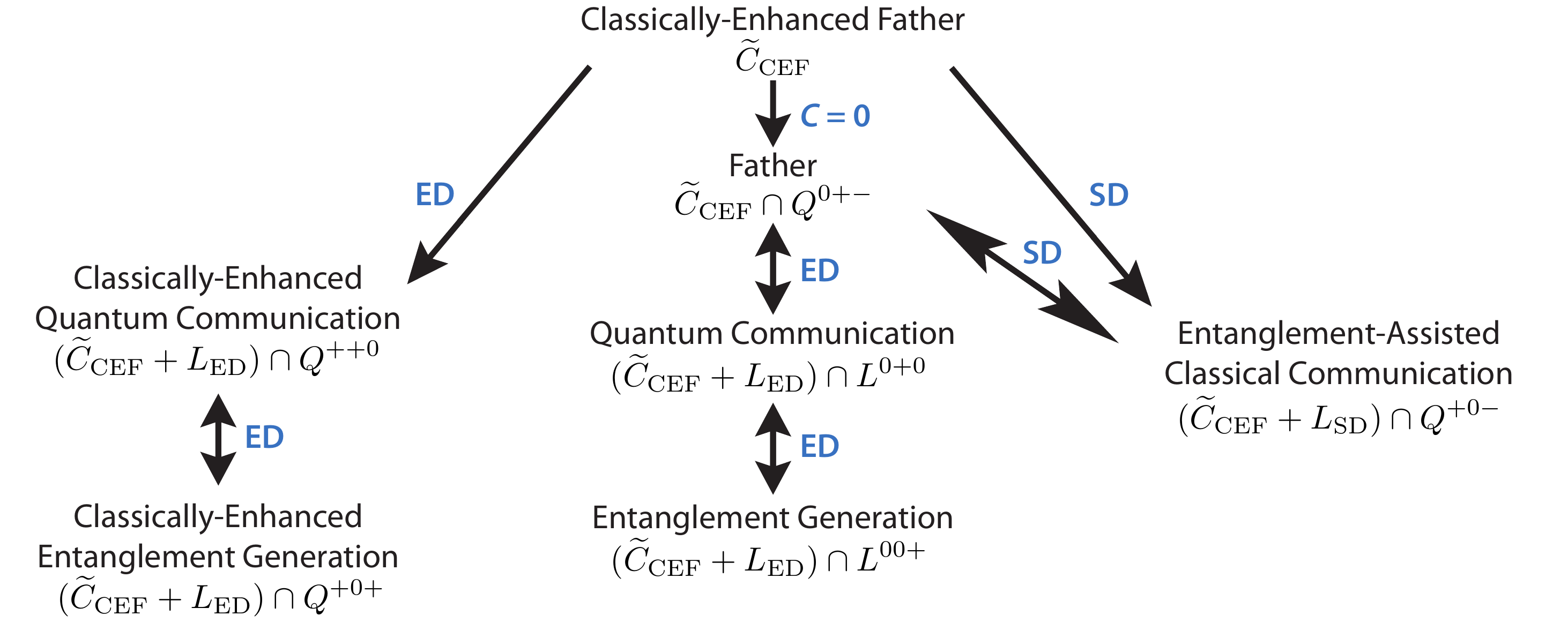}
\end{center}
\caption{The above figure depicts the relations between the capacity regions
for the dynamic case. Any bi-directional arrow represents a bijection between
two capacity regions. Any one-way arrow represents an injection from one
capacity region to another.}%
\label{fig:Father-Side}%
\end{figure*}

The father protocol consumes a noisy dynamic resource and noiseless
entanglement to generate noiseless quantum communication. The father's
achievable rate region $\widetilde{\mathcal{C}}_{\text{DD}}^{0+-}%
(\mathcal{N})$ and capacity region $\mathcal{C}_{\text{DD}}^{0+-}%
(\mathcal{N})$ lie in the $Q^{0+-}$ quadrant of the $(C,Q,E)$ space:%
\begin{align}
\widetilde{\mathcal{C}}_{\text{DD}}^{0+-}(\mathcal{N})  &  =\widetilde
{\mathcal{C}}_{\text{DD}}(\mathcal{N})\cap Q^{0+-},\\
\mathcal{C}_{\text{DD}}^{0+-}(\mathcal{N})  &  =\mathcal{C}_{\text{DD}%
}(\mathcal{N})\cap Q^{0+-}.
\end{align}
The above relations shows that the father's respective regions are special
cases of the direct dynamic achievable rate region $\widetilde{\mathcal{C}%
}_{\text{DD}}(\mathcal{N})$ and the direct dynamic capacity region
$\mathcal{C}_{\text{DD}}(\mathcal{N})$. The father capacity theorem states
that its capacity region $\mathcal{C}_{\text{DD}}^{0+-}(\mathcal{N})$\ is
equivalent to its achievable rate region $\widetilde{\mathcal{C}}_{\text{DD}%
}^{0+-}(\mathcal{N})$ \cite{DHW05RI}:%
\begin{equation}
\mathcal{C}_{\text{DD}}^{0+-}(\mathcal{N})=\widetilde{\mathcal{C}}_{\text{DD}%
}^{0+-}(\mathcal{N}). \label{eq:father-capacity}%
\end{equation}
The father achievable rate region $\widetilde{\mathcal{C}}_{\text{DD}}%
^{0+-}(\mathcal{N})$ is a special case of the classically-enhanced father
protocol where there is no classical communication \cite{HW08GFP}:
\begin{equation}
\widetilde{\mathcal{C}}_{\text{DD}}^{0+-}(\mathcal{N})=\widetilde{\mathcal{C}%
}_{\text{CEF}}(\mathcal{N})\cap Q^{0+-}. \label{CEF_FP}%
\end{equation}

The classically-enhanced quantum communication protocol consumes a noisy
dynamic resource to generate noiseless classical communication and noiseless
quantum communication. Its achievable rate region $\widetilde{\mathcal{C}%
}_{\text{DD}}^{++0}(\mathcal{N})$ and capacity region $\mathcal{C}_{\text{DD}%
}^{++0}(\mathcal{N})$ lie in the $Q^{++0}$ quadrant of the $(C,Q,E)$ space:%
\begin{align}
\widetilde{\mathcal{C}}_{\text{DD}}^{++0}(\mathcal{N})  &  =\widetilde
{\mathcal{C}}_{\text{DD}}(\mathcal{N})\cap Q^{++0},\\
\mathcal{C}_{\text{DD}}^{++0}(\mathcal{N})  &  =\mathcal{C}_{\text{DD}%
}(\mathcal{N})\cap Q^{++0}.
\end{align}
The above relations establish that its respective regions are special cases of
the direct dynamic achievable rate region $\widetilde{\mathcal{C}}_{\text{DD}%
}(\mathcal{N})$ and the direct dynamic capacity region $\mathcal{C}%
_{\text{DD}}(\mathcal{N})$. The classically-enhanced quantum communication
capacity theorem states that the achievable rate region $\widetilde
{\mathcal{C}}_{\text{DD}}^{++0}(\mathcal{N})$ is equivalent to the capacity
region $\mathcal{C}_{\text{DD}}^{++0}(\mathcal{N})$ \cite{DS03}:%
\begin{equation}
\widetilde{\mathcal{C}}_{\text{DD}}^{++0}\left(  \mathcal{N}\right)
=\mathcal{C}_{\text{DD}}^{++0}\left(  \mathcal{N}\right)  .
\label{eq:class-enh-q-comm-capacity}%
\end{equation}
Its achievable rate region $\widetilde{\mathcal{C}}_{\text{DD}}^{++0}%
(\mathcal{N})$ is obtainable from the classically-enhanced father's achievable
rate region $\widetilde{\mathcal{C}}_{\text{CEF}}(\mathcal{N})$ by combining
it with entanglement distribution and keeping the points with zero
entanglement \cite{HW08GFP}:
\begin{equation}
\widetilde{\mathcal{C}}_{\text{DD}}^{++0}(\mathcal{N})=(\widetilde
{\mathcal{C}}_{\text{CEF}}(\mathcal{N})+L_{\text{ED}})\cap Q^{++0}.
\label{CEF_CQ}%
\end{equation}

The classically-enhanced entanglement generation protocol consumes a noisy
dynamic resource to generate noiseless classical communication and noiseless
entanglement. Its achievable rate region $\widetilde{\mathcal{C}}_{\text{DD}%
}^{+0+}(\mathcal{N})$ and capacity region $\mathcal{C}_{\text{DD}}%
^{+0+}(\mathcal{N})$ lie in the $Q^{+0+}$ quadrant of the $(C,Q,E)$ space:%
\begin{align}
\widetilde{\mathcal{C}}_{\text{DD}}^{+0+}(\mathcal{N})  &  =\widetilde
{\mathcal{C}}_{\text{DD}}(\mathcal{N})\cap Q^{+0+},\label{CE_AR}\\
\mathcal{C}_{\text{DD}}^{+0+}(\mathcal{N})  &  =\mathcal{C}_{\text{DD}%
}(\mathcal{N})\cap Q^{+0+}. \label{CE_CR}%
\end{align}
The above relations establish that its respective regions are special cases of
the direct dynamic achievable rate region $\widetilde{\mathcal{C}}_{\text{DD}%
}(\mathcal{N})$ and the direct dynamic capacity region $\mathcal{C}%
_{\text{DD}}(\mathcal{N})$. Its capacity theorem states that the achievable
rate region $\widetilde{\mathcal{C}}_{\text{DD}}^{+0+}(\mathcal{N})$ is
equivalent to the capacity region $\mathcal{C}_{\text{DD}}^{+0+}(\mathcal{N})$
\cite{DS03}:%
\begin{equation}
\widetilde{\mathcal{C}}_{\text{DD}}^{+0+}\left(  \mathcal{N}\right)
=\mathcal{C}_{\text{DD}}^{+0+}\left(  \mathcal{N}\right)  . \label{DDp0pCR}%
\end{equation}
Its achievable rate region $\widetilde{\mathcal{C}}_{\text{DD}}^{+0+}%
(\mathcal{N})$ is obtainable from the classically-enhanced father's achievable
rate region $\widetilde{\mathcal{C}}_{\text{CEF}}(\mathcal{N})$ by combining
it with entanglement distribution and keeping the points with zero
entanglement \cite{HW08GFP}:
\begin{equation}
\widetilde{\mathcal{C}}_{\text{DD}}^{+0+}(\mathcal{N})=(\widetilde
{\mathcal{C}}_{\text{CEF}}(\mathcal{N})+L_{\text{ED}})\cap Q^{+0+}.
\label{CEF_CE}%
\end{equation}

The entanglement-assisted classical communication protocol consumes a noisy
dynamic resource and noiseless entanglement to generate noiseless classical
communication. Its achievable rate region $\widetilde{\mathcal{C}}_{\text{DD}%
}^{+0-}(\mathcal{N})$ and capacity region $\mathcal{C}_{\text{DD}}%
^{+0-}(\mathcal{N})$ lie in the $Q^{+0-}$ quadrant of the $(C,Q,E)$ space:%
\begin{align}
\widetilde{\mathcal{C}}_{\text{DD}}^{+0-}(\mathcal{N})  &  =\widetilde
{\mathcal{C}}_{\text{DD}}(\mathcal{N})\cap Q^{+0-},\\
\mathcal{C}_{\text{DD}}^{+0-}(\mathcal{N})  &  =\mathcal{C}_{\text{DD}%
}(\mathcal{N})\cap Q^{+0-}.
\end{align}
The above relations establish that its respective regions are special cases of
the direct dynamic achievable rate region $\widetilde{\mathcal{C}}_{\text{DD}%
}(\mathcal{N})$ and the direct dynamic capacity region $\mathcal{C}%
_{\text{DD}}(\mathcal{N})$. Its capacity theorem states that the achievable
rate region $\widetilde{\mathcal{C}}_{\text{DD}}^{+0-}(\mathcal{N})$ is
equivalent to the capacity region $\mathcal{C}_{\text{DD}}^{+0-}(\mathcal{N})$
\cite{BSST01,Hol01a,arx2004shor,DHW05RI}:%
\begin{equation}
\widetilde{\mathcal{C}}_{\text{DD}}^{+0-}\left(  \mathcal{N}\right)
=\mathcal{C}_{\text{DD}}^{+0-}\left(  \mathcal{N}\right)  .
\label{EAC_capacity}%
\end{equation}
Its achievable rate region $\widetilde{\mathcal{C}}_{\text{DD}}^{+0-}%
(\mathcal{N})$ is obtainable from the classically-enhanced father's achievable
rate region $\widetilde{\mathcal{C}}_{\text{CEF}}(\mathcal{N})$ by combining
it with the super-dense coding protocol and keeping the points with zero
quantum communication \cite{HW08GFP}:%
\begin{equation}
\widetilde{\mathcal{C}}_{\text{DD}}^{+0-}(\mathcal{N})=(\widetilde
{\mathcal{C}}_{\text{CEF}}(\mathcal{N})+L_{\text{SD}})\cap Q^{+0-}.
\label{CEF_EAC}%
\end{equation}

Forward classical communication does not increase the entanglement generation
capacity or the quantum communication capacity \cite{BDSW96,BKN98}. Thus,
there is a simple relation between the achievable rate region $\widetilde
{\mathcal{C}}_{\text{DD}}^{-0+}\left(  \mathcal{N}\right)  $ and the
achievable entanglement generation capacity region $\widetilde{\mathcal{C}%
}_{\text{DD}}^{00+}\left(  \mathcal{N}\right)  $:%
\[
\widetilde{\mathcal{C}}_{\text{DD}}^{-0+}\left(  \mathcal{N}\right)
=\widetilde{\mathcal{C}}_{\text{DD}}^{00+}\left(  \mathcal{N}\right)
+L^{-00},
\]
where $L^{-00}\equiv\left\{  \lambda\left(  -1,0,0\right)  :\lambda
\geq0\right\}  $. The converse proof of the entanglement generation capacity
region states that the achievable rate region $\widetilde{\mathcal{C}%
}_{\text{DD}}^{00+}\left(  \mathcal{N}\right)  $ is optimal \cite{Devetak03}:%
\[
\widetilde{\mathcal{C}}_{\text{DD}}^{00+}\left(  \mathcal{N}\right)
=\mathcal{C}_{\text{DD}}^{00+}\left(  \mathcal{N}\right)  \text{.}%
\]
The converse proof of the region $\widetilde{\mathcal{C}}_{\text{DD}}%
^{-0+}\left(  \mathcal{N}\right)  $ then follows%
\begin{multline}
\mathcal{C}_{\text{DD}}^{-0+}\left(  \mathcal{N}\right)  =\mathcal{C}%
_{\text{DD}}^{00+}\left(  \mathcal{N}\right)  +L^{-00}%
\label{eq:converse-forward-class-comm-ent-gen}\\
=\widetilde{\mathcal{C}}_{\text{DD}}^{00+}\left(  \mathcal{N}\right)
+L^{-00}=\widetilde{\mathcal{C}}_{\text{DD}}^{-0+}\left(  \mathcal{N}\right)
.
\end{multline}
The entanglement generation achievable rate region results from combining the
father achievable rate region in (\ref{CEF_FP}) with entanglement distribution
\cite{DHW03}:
\begin{equation}
\widetilde{\mathcal{C}}_{\text{DD}}^{00+}\left(  \mathcal{N}\right)
=((\widetilde{\mathcal{C}}_{\text{CEF}}\left(  \mathcal{N}\right)  \cap
Q^{0+-})+L_{\text{ED}})\cap L^{00+}. \label{eq:CAEG-from-father-ED}%
\end{equation}

Similar results for the classically-assisted quantum communication capacity
region $\mathcal{C}_{\text{DD}}^{-+0}(\mathcal{N})$ hold
\begin{multline}
\mathcal{C}_{\text{DD}}^{-+0}(\mathcal{N})=\mathcal{C}_{\text{DD}}%
^{0+0}(\mathcal{N})+L^{-00}\\
=\widetilde{\mathcal{C}}_{\text{DD}}^{0+0}(\mathcal{N})+L^{-00}=\widetilde
{\mathcal{C}}_{\text{DD}}^{-+0}(\mathcal{N}). \label{DDNP0CR}%
\end{multline}
The quantum communication achievable rate region results from combining the
father achievable rate region in (\ref{CEF_FP}) with entanglement distribution
\cite{DHW03}:
\begin{equation}
\widetilde{\mathcal{C}}_{\text{DD}}^{0+0}\left(  \mathcal{N}\right)
=((\widetilde{\mathcal{C}}_{\text{CEF}}\left(  \mathcal{N}\right)  \cap
Q^{0+-})+L_{\text{ED}})\cap L^{0+0}. \label{eq:CAQG-from-father-ED}%
\end{equation}

\subsubsection{Octants That Generate Quantum Communication}

We first consider all of the octants with corresponding protocols that
generate quantum communication, i.e., octants of the form $\left(  \pm
,+,\pm\right)  $. The proof of one of the octants is the converse proof in
Ref.~\cite{HW08GFP}. The proofs of two of the remaining three octants are
similar to the entanglement-generating octants from the static case. The
similarity holds because these octants generate the noiseless version of the
noisy dynamic resource. The proof of the last remaining octant is different
from any we have seen so far. We discuss later how its proof gives insight
into the question of using entanglement-assisted coding versus teleportation.

$\boldsymbol{(+,+,-)}$. The converse theorem for this octant is the converse
for entanglement-assisted communication of classical and quantum information.
(Ref.~\cite{HW08GFP} contains the proof). We briefly recall these inequalities
here:%
\begin{align*}
C+2Q  &  \leq I\left(  AX;B\right)  _{\sigma},\\
Q  &  \leq I\left(  A\rangle BX\right)  _{\sigma}+\left\vert E\right\vert ,\\
C+Q  &  \leq I\left(  X;B\right)  _{\sigma}+I\left(  A\rangle BX\right)
+\left\vert E\right\vert ,
\end{align*}
where $\sigma$ is a state of the form in (\ref{DD_sigma}). The above
inequalities form the one-shot, one-state region. Interestingly, the above set
of inequalities is a translation of the unit resource capacity region to the
classically-enhanced father protocol.

$\boldsymbol{(+,+,+)}$. This octant exploits the projection technique with
entanglement distribution. Let%
\[
\mathcal{C}_{\text{DD}}^{+++}\left(  \mathcal{N}\right)  \equiv\mathcal{C}%
_{\text{DD}}\left(  \mathcal{N}\right)  \cap O^{+++}%
\]
and recall the definition of $\mathcal{C}_{\text{DD}}^{+0+}(\mathcal{N})$ in
(\ref{CE_CR}). We exploit the line of entanglement distribution $L_{\text{ED}%
}$ as defined in (\ref{eq:line-ED}). Define the following maps:%
\begin{align*}
f  &  :S\rightarrow(S+L_{\text{ED}})\cap Q^{+0+},\\
\hat{f}  &  :S\rightarrow(S-L_{\text{ED}})\cap O^{+++}.
\end{align*}
The map $f$ translates the set $S$ in the entanglement distribution direction
and keeps the points that lie on the $Q^{+0+}$ quadrant. The map $\hat{f}$, in
a sense, undoes the effect of $f$ by moving the set $S$ back to the $O^{+++}$ octant.

The inclusion $\mathcal{C}_{\text{DD}}^{+++}\left(  \mathcal{N}\right)
\subseteq\hat{f}(f(\mathcal{C}_{\text{DD}}^{+++}\left(  \mathcal{N}\right)
))$ holds because%
\begin{align}
&  \mathcal{C}_{\text{DD}}^{+++}\left(  \mathcal{N}\right) \nonumber\\
&  =\mathcal{C}_{\text{DD}}^{+++}\left(  \mathcal{N}\right)  \cap
O^{+++}\nonumber\\
&  \subseteq(((\mathcal{C}_{\text{DD}}^{+++}\left(  \mathcal{N}\right)
+L_{\text{ED}})\cap Q^{+0+})-L_{\text{ED}})\cap O^{+++}\nonumber\\
&  =(f(\mathcal{C}_{\text{DD}}^{+++}\left(  \mathcal{N}\right)  )-L_{\text{ED}%
})\cap O^{+++}\nonumber\\
&  =\hat{f}(f(\mathcal{C}_{\text{DD}}^{+++}\left(  \mathcal{N}\right)  )).
\label{DDppp1}%
\end{align}
The first set equivalence is obvious from the definition of $\mathcal{C}%
_{\text{DD}}^{+++}\left(  \mathcal{N}\right)  $. The first inclusion follows
from the following logic. Pick any point $a\equiv(C,Q,E)\in\mathcal{C}%
_{\text{DD}}^{+++}\left(  \mathcal{N}\right)  $ and a particular point
$b\equiv\left(  0,-Q,Q\right)  \in L_{\text{ED}}$. It follows that
$a+b=\left(  C,0,E+Q\right)  \in(\mathcal{C}_{\text{DD}}^{+++}\left(
\mathcal{N}\right)  +L_{\text{ED}})\cap Q^{+0+}$. Pick a point $-b=\left(
0,Q,-Q\right)  \in-L_{\text{ED}}$. It follows that $a+b-b\in(((\mathcal{C}%
_{\text{DD}}^{+++}\left(  \mathcal{N}\right)  +L_{\text{ED}})\cap
Q^{+0+})-L_{\text{ED}})\cap O^{+++}$ and that $a+b-b=\left(  C,Q,E\right)
=a$. The first inclusion thus holds because every point in $\mathcal{C}%
_{\text{DD}}^{+++}\left(  \mathcal{N}\right)  \cap O^{+++}$ is in
$(((\mathcal{C}_{\text{DD}}^{+++}\left(  \mathcal{N}\right)  +L_{\text{ED}%
})\cap Q^{+0+})-L_{\text{ED}})\cap O^{+++}$. The second set equivalence
follows from the definition of $f$ and the third set equivalence follows from
the definition of $\hat{f}$.

It is operationally clear that the following inclusion holds%
\begin{equation}
f(\mathcal{C}_{\text{DD}}^{+++}\left(  \mathcal{N}\right)  )\subseteq
\mathcal{C}_{\text{DD}}^{+0+}\left(  \mathcal{N}\right)  , \label{DDppp2}%
\end{equation}
because the mapping $f$ converts any achievable point $a\in\mathcal{C}%
_{\text{DD}}^{+++}\left(  \mathcal{N}\right)  $ to an achievable point in
$\mathcal{C}_{\text{DD}}^{+0+}\left(  \mathcal{N}\right)  $ by consuming all
of the quantum communication at point $a$ with entanglement distribution.

The converse proof of the classically-enhanced entanglement generation
protocol \cite{DS03} states that the following inclusion holds%
\begin{equation}
\mathcal{C}_{\text{DD}}^{+0+}(\mathcal{N})\subseteq\widetilde{\mathcal{C}%
}_{\text{DD}}^{+0+}(\mathcal{N}). \label{eq:converse-CEEG}%
\end{equation}

The inclusion $\hat{f}(\widetilde{\mathcal{C}}_{\text{DD}}^{+0+}\left(
\mathcal{N}\right)  )\subseteq\widetilde{\mathcal{C}}_{\text{DD}}^{+++}\left(
\mathcal{N}\right)  $ holds because%
\begin{align}
&  \hat{f}(\widetilde{\mathcal{C}}_{\text{DD}}^{+0+}\left(  \mathcal{N}%
\right)  )\nonumber\\
&  =(((\widetilde{\mathcal{C}}_{\text{CEF}}\left(  \mathcal{N}\right)
+L_{\text{ED}})\cap Q^{+0+})-L_{\text{ED}})\cap O^{+++}\nonumber\\
&  \subseteq((\widetilde{\mathcal{C}}_{\text{CEF}}\left(  \mathcal{N}\right)
+L_{\text{ED}})-L_{\text{ED}})\cap O^{+++}\nonumber\\
&  =((\widetilde{\mathcal{C}}_{\text{CEF}}\left(  \mathcal{N}\right)
+L_{\text{ED}})\cap O^{+++})\cup\nonumber\\
&  \ \ \ \ \ \ \ \ \ \ \ \ \ \ \ \ \ \ \ \ \ \ \ \ \ \ \ \ \ \ \ ((\widetilde
{\mathcal{C}}_{\text{CEF}}\left(  \mathcal{N}\right)  -L_{\text{ED}})\cap
O^{+++})\nonumber\\
&  \subseteq\widetilde{\mathcal{C}}_{\text{DD}}^{+++}\left(  \mathcal{N}%
\right)  . \label{DDppp3}%
\end{align}
The first set equivalence follows from (\ref{CEF_CE}) the definition of
$\hat{f}$. The first inclusion follows by dropping the intersection with
$Q^{+0+}$. The second set equivalence follows because $(\widetilde
{\mathcal{C}}_{\text{CEF}}\left(  \mathcal{N}\right)  +L_{\text{ED}%
})-L_{\text{ED}}=(\widetilde{\mathcal{C}}_{\text{CEF}}\left(  \mathcal{N}%
\right)  +L_{\text{ED}})\cup(\widetilde{\mathcal{C}}_{\text{CEF}}\left(
\mathcal{N}\right)  -L_{\text{ED}})$, and the second inclusion holds by
(\ref{CEF_DD}) and because $(\widetilde{\mathcal{C}}_{\text{CEF}}\left(
\mathcal{N}\right)  -L_{\text{ED}})\cap O^{+++}=\left(  0,0,0\right)  $.

Putting (\ref{DDppp1}), (\ref{DDppp2}), (\ref{eq:converse-CEEG}), and
(\ref{DDppp3}) together, we have the following inclusion:%
\begin{multline*}
\mathcal{C}_{\text{DD}}^{+++}\left(  \mathcal{N}\right)  \subseteq\hat
{f}(f(\mathcal{C}_{\text{DD}}^{+++}\left(  \mathcal{N}\right)  ))\\
\subseteq\hat{f}(\mathcal{C}_{\text{DD}}^{+0+}\left(  \mathcal{N}\right)
)\subseteq\hat{f}(\widetilde{\mathcal{C}}_{\text{DD}}^{+0+}\left(
\mathcal{N}\right)  )\subseteq\widetilde{\mathcal{C}}_{\text{DD}}^{+++}\left(
\mathcal{N}\right)  .
\end{multline*}
The above inclusion $\mathcal{C}_{\text{DD}}^{+++}\left(  \mathcal{N}\right)
\subseteq\widetilde{\mathcal{C}}_{\text{DD}}^{+++}\left(  \mathcal{N}\right)
$ is the statement of the converse theorem for this octant.

The proof of the octant $\left(  -,+,+\right)  $ is similar to the above proof
and appears in Appendix~\ref{sec:-++dynamic}.

$\boldsymbol{(-,+,-)}$. This octant is one of the more interesting octants for
the direct dynamic trade-off. Interestingly, this octant gives insight into
the question of the use of teleportation versus entanglement-assisted quantum
error correction
\cite{BDH06,BDH06IEEE,HBD07,arx2007wildeEA,arx2007wildeEAQCC,HBD08QLDPC,arx2008wildeUQCC,arx2008wildeGEAQCC}%
. We discuss these insights in Section~\ref{sec:EA-insights}.

The converse proof for this octant follows from a \textit{reductio ad
absurdum} argument, reminiscent of similar arguments that appeared in
Ref.~\cite{HW08GFP}. We use this technique to obtain two bounds for this
octant, obviating the need to consider the technique used for the previous
octants. Figure~\ref{fig:octant-+-}\ illustrates the bounds for this octant.

First, consider that the bound $Q\leq I\left(  A\rangle BX\right)  $ applies
to all points in the $Q^{-+0}$ quadrant because forward classical
communication does not increase the quantum communication capacity of a
quantum channel \cite{BDSW96,BKN98}. Consider combining achievable points
along the boundary of the above bound with the inverse of the entanglement
distribution protocol. This procedure outlines the following bound:%
\begin{equation}
Q\leq I\left(  A\rangle BX\right)  +\left\vert E\right\vert .
\label{eq:-+-bound1}%
\end{equation}
The bound in (\ref{eq:-+-bound1}) applies to all points in this octant. Were
it not so, then one could combine points outside (\ref{eq:-+-bound1}) with
entanglement distribution and beat the bound $Q\leq I\left(  A\rangle
BX\right)  $ that applies to all points in the $Q^{-+0}$ quadrant.

Next, consider that the bound $2Q\leq I\left(  AX;B\right)  $ applies to all
points in the $Q^{0+-}$ quadrant. This bound follows from the limit on the
entanglement-assisted quantum capacity of a quantum channel
\cite{BSST01,DHW05RI,HW08GFP}. Consider combining achievable points along the
boundary of the above bound with the inverse of the super-dense coding
protocol. This procedure outlines the following bound:%
\begin{equation}
2Q\leq I\left(  AX;B\right)  +\left\vert C\right\vert . \label{eq:-+-bound2}%
\end{equation}
The bound in (\ref{eq:-+-bound2}) applies to all points in this octant. Were
it not so, then one could combine points inside (\ref{eq:-+-bound1}) but
outside (\ref{eq:-+-bound2}) with super-dense coding and beat the bound
$2Q\leq I\left(  AX;B\right)  $ that applies to all points in the $Q^{0+-}$ quadrant.

We can achieve all points specified by the above boundaries simply by
combining the classically-enhanced father protocol with teleportation or
entanglement distribution. Thus, the statement of the converse holds for this
octant:%
\[
\mathcal{C}_{\text{DD}}^{-+-}\left(  \mathcal{N}\right)  \subseteq
\widetilde{\mathcal{C}}_{\text{DD}}^{-+-}\left(  \mathcal{N}\right)  .
\]
Figure~\ref{fig:octant-+-}\ illustrates these arguments.

We can also obtain these bounds with a direct, information-theoretic argument.
This argument appears in Appendix~\ref{sec:converse_-+-_dynamic}.%
%TCIMACRO{\FRAME{ftbpFU}{3.4904in}{2.591in}{0pt}{\Qcb{(Color online) The above
%figure displays the capacity region in the octant $\left(  -,+,-\right)  $,
%that corresponds to classical- and entanglement-assisted quantum
%communication. Plane I corresponds to the bound in (\ref{eq:-+-bound2}), and
%plane II\ corresponds to the bound in (\ref{eq:-+-bound1}). There cannot be
%any points outside plane II. If there were, we could combine such points with
%entanglement distribution to beat the bound on the quantum communication
%capacity. There cannot be any points inside plane II but outside plane I. If
%there were, we could combine such points with super-dense coding to beat the
%entanglement-assisted quantum capacity bound. Planes I and II intersect at the
%line of teleportation. We can achieve all points inside bounds
%(\ref{eq:-+-bound1}) and (\ref{eq:-+-bound2}) simply by combining the father
%protocol (the black dot) with teleportation and entanglement distribution.}%
%}{\Qlb{fig:octant-+-}}{octant-+-.pdf}{\special{ language "Scientific Word";
%type "GRAPHIC";  maintain-aspect-ratio TRUE;  display "USEDEF";
%valid_file "F";  width 3.4904in;  height 2.591in;  depth 0pt;
%original-width 5.4129in;  original-height 4.0067in;  cropleft "0";
%croptop "1";  cropright "1";  cropbottom "0";
%filename 'octant-+-.pdf';file-properties "XNPEU";}}}%
%BeginExpansion
\begin{figure}
[ptb]
\begin{center}
\includegraphics[
natheight=4.006700in,
natwidth=5.412900in,
height=2.591in,
width=3.4904in
]%
{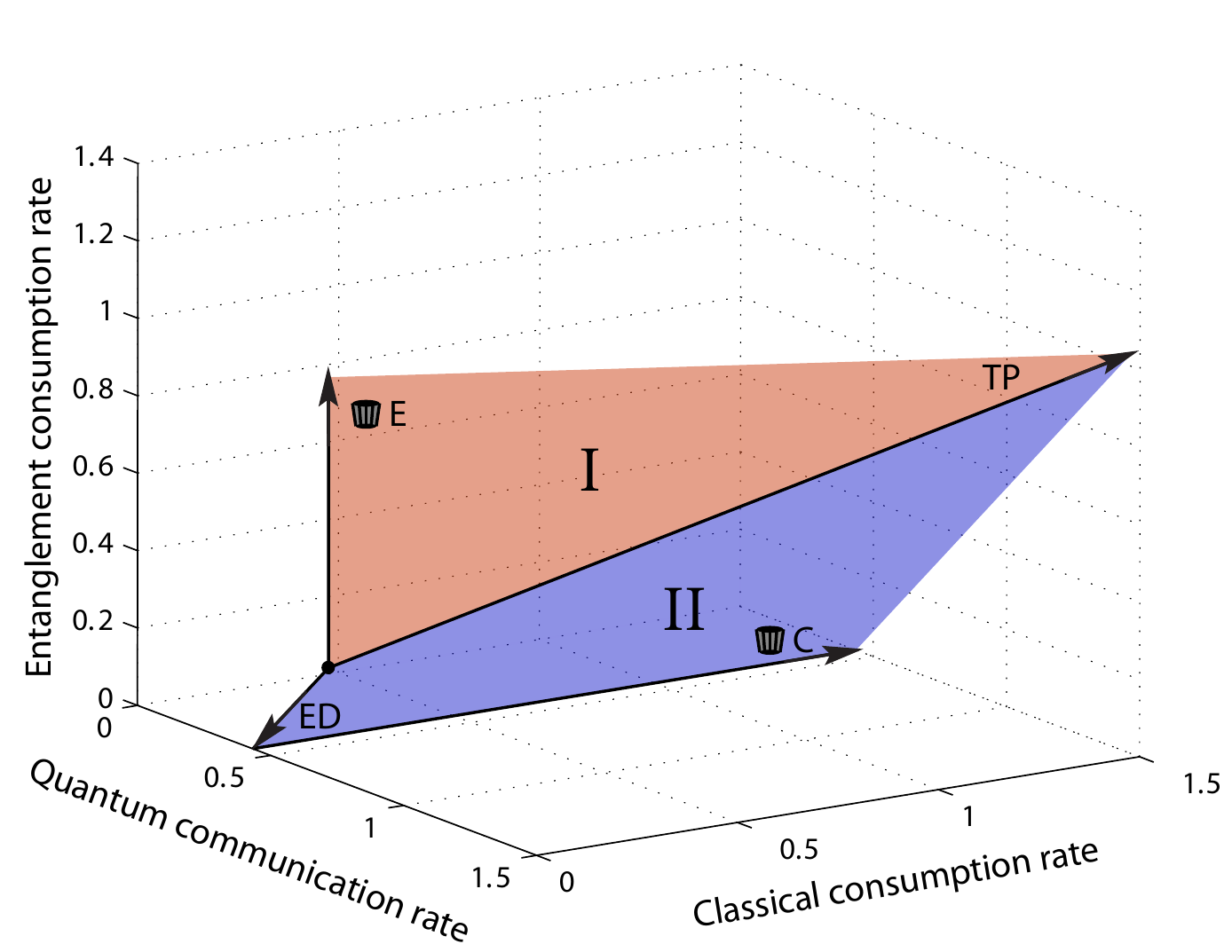}%
\caption{(Color online) The above figure displays the capacity region in the
octant $\left(  -,+,-\right)  $, that corresponds to classical- and
entanglement-assisted quantum communication. Plane I corresponds to the bound
in (\ref{eq:-+-bound2}), and plane II\ corresponds to the bound in
(\ref{eq:-+-bound1}). There cannot be any points outside plane II. If there
were, we could combine such points with entanglement distribution to beat the
bound on the quantum communication capacity. There cannot be any points inside
plane II but outside plane I. If there were, we could combine such points with
super-dense coding to beat the entanglement-assisted quantum capacity bound.
Planes I and II intersect at the line of teleportation. We can achieve all
points inside bounds (\ref{eq:-+-bound1}) and (\ref{eq:-+-bound2}) simply by
combining the father protocol (the black dot) with teleportation and
entanglement distribution.}%
\label{fig:octant-+-}%
\end{center}
\end{figure}
%EndExpansion

\begin{remark}
The two bounds in (\ref{eq:-+-bound1})\ and (\ref{eq:-+-bound2}) are the same
as two bounds that apply to the octant $O^{++-}$\ (corresponding to
entanglement-assisted communication of classical and quantum information
\cite{HW08GFP}). Thus, these bounds are extensions of the same bounds from the
$O^{++-}$ octant.
\end{remark}

\subsubsection{Octants That Consume Quantum Communication}

We now consider the four octants with corresponding protocols that consume
quantum communication, i.e., octants of the form $\left(  \pm,-,\pm\right)  $.
The proof of one of the octants is trivial. The proofs of two of the remaining
three octants are similar to the proofs of the entanglement-consuming octants
from the static case. The similarity holds because these octants consume the
noiseless version of the noisy dynamic resource. The proof of the other octant
exploits information theoretic arguments.

$\boldsymbol{(-,-,-)}$. $\widetilde{\mathcal{C}}_{\text{DD}}\left(
\mathcal{N}\right)  $ completely contains this octant.

$\boldsymbol{(+,-,+)}$. We use a trick similar to the $(+,-,-)$ octant for the
static case. Let $\text{id}^{\otimes|Q|}$ denote a noiseless qubit channel of
size $|Q|$ qubits. We need the following additivity lemma.

\begin{lemma}
\label{lem:additivity-+0+}The following inclusion holds%
\[
\mathcal{C}_{\text{DD}}^{+0+}(\mathcal{N}\otimes\text{id}^{\otimes
|Q|})\subseteq\widetilde{\mathcal{C}}_{\text{DD}}^{+0+}(\mathcal{N}%
)+\widetilde{\mathcal{C}}_{\text{DD}}^{+0+}(\text{id}^{\otimes|Q|}).
\]

\end{lemma}

\begin{IEEEproof}
Entanglement distribution induces a linear bijection $f:\widetilde
{\mathcal{C}}_{\text{DD}}^{++0}\left(  \mathcal{N}\right)  \rightarrow\widetilde
{\mathcal{C}}_{\text{DD}}^{+0+}\left(  \mathcal{N}\right)  $ between the
classically-enhanced quantum communication achievable region $\widetilde
{\mathcal{C}}_{\text{DD}}^{++0}$ and the classically-enhanced entanglement generation
achievable region $\widetilde{\mathcal{C}}_{\text{DD}}^{+0+}$. The linear bijection $f$
behaves as follows for every point $\left(  R,Q,0\right)  \in\widetilde
{\mathcal{C}}_{\text{DD}}^{++0}$:%
\[
f:\left(  R,Q,0\right)  \rightarrow\left(  R,0,Q\right)  .
\]
The following relation holds%
\begin{equation}
f(\widetilde{\mathcal{C}}_{\text{DD}}^{++0}\left(  \mathcal{N}\right)  )=\widetilde
{\mathcal{C}}_{\text{DD}}^{+0+}\left(  \mathcal{N}\right)  ,
\label{eq:ent-dist-CEQC-CEEG}%
\end{equation}
because applying entanglement distribution to the classically-enhanced quantum
communication resource inequality gives classically-enhanced entanglement
generation \cite{DS03}.
The inclusion $\mathcal{C}_{\text{DD}}^{++0}(\mathcal{N}\otimes$id$^{\otimes
|Q|})\subseteq f^{-1}(\widetilde{\mathcal{C}}_{\text{DD}}^{+0+}(\mathcal{N}%
)+\widetilde{\mathcal{C}}_{\text{DD}}^{+0+}($id$^{\otimes|Q|}))$ holds because%
\begin{align*}
&  \mathcal{C}_{\text{DD}}^{++0}(\mathcal{N}\otimes\text{id}^{\otimes|Q|})\\
&  =\mathcal{C}_{\text{DD}}^{++0}(\mathcal{N})+(0,Q,0)\\
&  \subseteq \mathcal{C}_{\text{DD}}^{++0}(\mathcal{N})+\mathcal{C}_{\text{DD}}^{++0}%
(\text{id}^{\otimes|Q|})\\
&  =\widetilde{\mathcal{C}}_{\text{DD}}^{++0}(\mathcal{N})+\widetilde{\mathcal{C}}_{\text{DD}%
}^{++0}(\text{id}^{\otimes|Q|})\\
&  =f^{-1}(\widetilde{\mathcal{C}}_{\text{DD}}^{+0+}(\mathcal{N}))+f^{-1}(\widetilde
{\mathcal{C}}_{\text{DD}}^{+0+}(\text{id}^{\otimes|Q|}))\\
&  =f^{-1}(\widetilde{\mathcal{C}}_{\text{DD}}^{+0+}(\mathcal{N})+\widetilde
{\mathcal{C}}_{\text{DD}}^{+0+}(\text{id}^{\otimes|Q|})).
\end{align*}
The first set equivalence follows because the capacity region of the noisy
channel $\mathcal{N}$ combined with a rate $Q$ noiseless qubit channel is
equivalent to a translation of the capacity region of the noisy channel
$\mathcal{N}$. The first inclusion follows because the capacity region of a
rate $Q$ noiseless qubit channel contains the rate triple $(0,Q,0)$. The
second set equivalence follows from the classically-enhanced quantum
communication capacity theorem in (\ref{eq:class-enh-q-comm-capacity}), the
third set equivalence from (\ref{eq:ent-dist-CEQC-CEEG}), and the fourth set
equivalence from linearity of the map $f$. The above inclusion implies the
following one:
\[
f(\mathcal{C}_{\text{DD}}^{++0}(\mathcal{N}\otimes\text{id}^{\otimes|Q|}))\subseteq
\widetilde{\mathcal{C}}_{\text{DD}}^{+0+}(\mathcal{N})+\widetilde{\mathcal{C}}_{\text{DD}}%
^{+0+}(\text{id}^{\otimes|Q|}).
\]
The lemma follows because%
\begin{align*}
f(\mathcal{C}_{\text{DD}}^{++0}(\mathcal{N}\otimes\text{id}^{\otimes|Q|}))  &
=f(\widetilde{\mathcal{C}}_{\text{DD}}^{++0}(\mathcal{N}\otimes\text{id}^{\otimes
|Q|}))\\
&  =\widetilde{\mathcal{C}}_{\text{DD}}^{+0+}(\mathcal{N}\otimes\text{id}^{\otimes
|Q|})\\
&  =\mathcal{C}_{\text{DD}}^{+0+}(\mathcal{N}\otimes\text{id}^{\otimes|Q|}),
\end{align*}
where we apply the relations in (\ref{eq:class-enh-q-comm-capacity}),
(\ref{eq:ent-dist-CEQC-CEEG}), and (\ref{DDp0pCR}).
\end{IEEEproof}

We now begin the converse proof of this octant. Observe that
\begin{equation}
\widetilde{\mathcal{C}}_{\text{DD}}^{+0+}(\text{id}^{\otimes|Q|}%
)=\widetilde{\mathcal{C}}_{\text{U}}^{+Q+}. \label{eq:+0+qubit=+Q+U}%
\end{equation}
Hence, for all $Q\leq0$,
\begin{equation}
\mathcal{C}_{\text{DD}}^{+Q+}(\mathcal{N})=\mathcal{C}_{\text{DD}}%
^{+0+}(\mathcal{N}\otimes\text{id}^{\otimes|Q|})\subseteq\widetilde
{\mathcal{C}}_{\text{DD}}^{+0+}(\mathcal{N})+\widetilde{\mathcal{C}}%
_{\text{U}}^{+Q+}, \label{DDpnp1}%
\end{equation}
where we apply Lemma~\ref{lem:additivity-+0+} and (\ref{eq:+0+qubit=+Q+U}).
The inclusion $\mathcal{C}_{\text{DD}}^{+-+}(\mathcal{N})\subseteq
\widetilde{\mathcal{C}}_{\text{DD}}^{+-+}$ follows because%
\begin{align*}
\mathcal{C}_{\text{DD}}^{+-+}(\mathcal{N})  &  =\bigcup_{Q\leq0}%
\mathcal{C}_{\text{DD}}^{+Q+}(\mathcal{N})\\
&  \subseteq\bigcup_{Q\leq0}(\widetilde{\mathcal{C}}_{\text{DD}}%
^{+0+}(\mathcal{N})+\widetilde{\mathcal{C}}_{\text{U}}^{+Q+})\\
&  =(\widetilde{\mathcal{C}}_{\text{DD}}^{+0+}(\mathcal{N})+\widetilde
{\mathcal{C}}_{\text{U}})\cap O^{+-+}\\
&  \subseteq(\widetilde{\mathcal{C}}_{\text{CEF}}(\mathcal{N})+\widetilde
{\mathcal{C}}_{\text{U}})\cap O^{+-+}\\
&  =\widetilde{\mathcal{C}}_{\text{DD}}^{+-+}(\mathcal{N}).
\end{align*}
The first and equivalence hold by definition, the first inclusion follows from
(\ref{DDpnp1}), the second set equivalence follows because $\bigcup_{Q\leq
0}\widetilde{\mathcal{C}}_{\text{U}}^{+Q+}=\widetilde{\mathcal{C}}_{\text{U}%
}\cap O^{+-+}$, and the second inclusion follows because combining the
classically-enhanced father region with entanglement distribution gives the
region for classically-enhanced entanglement generation. The above inclusion
$\mathcal{C}_{\text{DD}}^{+-+}(\mathcal{N})\subseteq\widetilde{\mathcal{C}%
}_{\text{DD}}^{+-+}$ is the statement of the converse theorem for this octant.

The proof of the $\left(  -,-,+\right)  $ octant is similar to the proof of
the above octant and appears in Appendix~\ref{sec:--+dynamic}.

$\boldsymbol{(+,-,-)}$. The proof exploits information-theoretic arguments
that show the following bounds apply to all rate triples $\left(
C,-\left\vert Q\right\vert ,-\left\vert E\right\vert \right)  $\ in this
octant:%
\begin{align}
C  &  \leq I\left(  AX;B\right)  +2\left\vert Q\right\vert ,\label{eq:+--1}\\
C  &  \leq I\left(  X;B\right)  +I\left(  A\rangle BX\right)  +\left\vert
Q\right\vert +\left\vert E\right\vert . \label{eq:+--2}%
\end{align}%
%TCIMACRO{\FRAME{ftbpFU}{3.5405in}{1.9683in}{0pt}{\Qcb{(Color online) The most
%general protocol for classical communication with the help of a noisy channel,
%noiseless entanglement, and noiseless quantum communication. Alice wishes to
%communicate a classical message $M$ to Bob. She shares entanglement with Bob
%in the form of maximally entangled states. Her half of the entanglement is in
%system $T_{A}$ and Bob's half is in the system $T_{B}$. Alice performs some
%CPTP\ encoding map $\QTR{cal}{E}$ on her classical message and her half of the
%entanglement. The output of this encoder is a quantum state in some register
%$A_{1}$ and a large number of systems $A^{\prime n}$ that are input to the
%channel. She transmits $A^{\prime n}$ through the noisy channel and the system
%$A_{1}$ over noiseless qubit channels. Bob receives the outputs $B^{n}$ of the
%channel and the system $A_{1}$ from the noiseless qubit channels. He combines
%these with his half of the entanglement and decodes the classical message that
%Alice transmits.}}{\Qlb{fig:EQA-class}}{eqa-class-operation.pdf}%
%{\special{ language "Scientific Word";  type "GRAPHIC";
%maintain-aspect-ratio TRUE;  display "USEDEF";  valid_file "F";
%width 3.5405in;  height 1.9683in;  depth 0pt;  original-width 8.2469in;
%original-height 5.0998in;  cropleft "0";  croptop "1";  cropright "1";
%cropbottom "0";  filename 'EQA-class-operation.pdf';file-properties "XNPEU";}%
%}}%
%BeginExpansion
\begin{figure}
[ptb]
\begin{center}
\includegraphics[
natheight=5.099800in,
natwidth=8.246900in,
height=1.9683in,
width=3.5405in
]%
{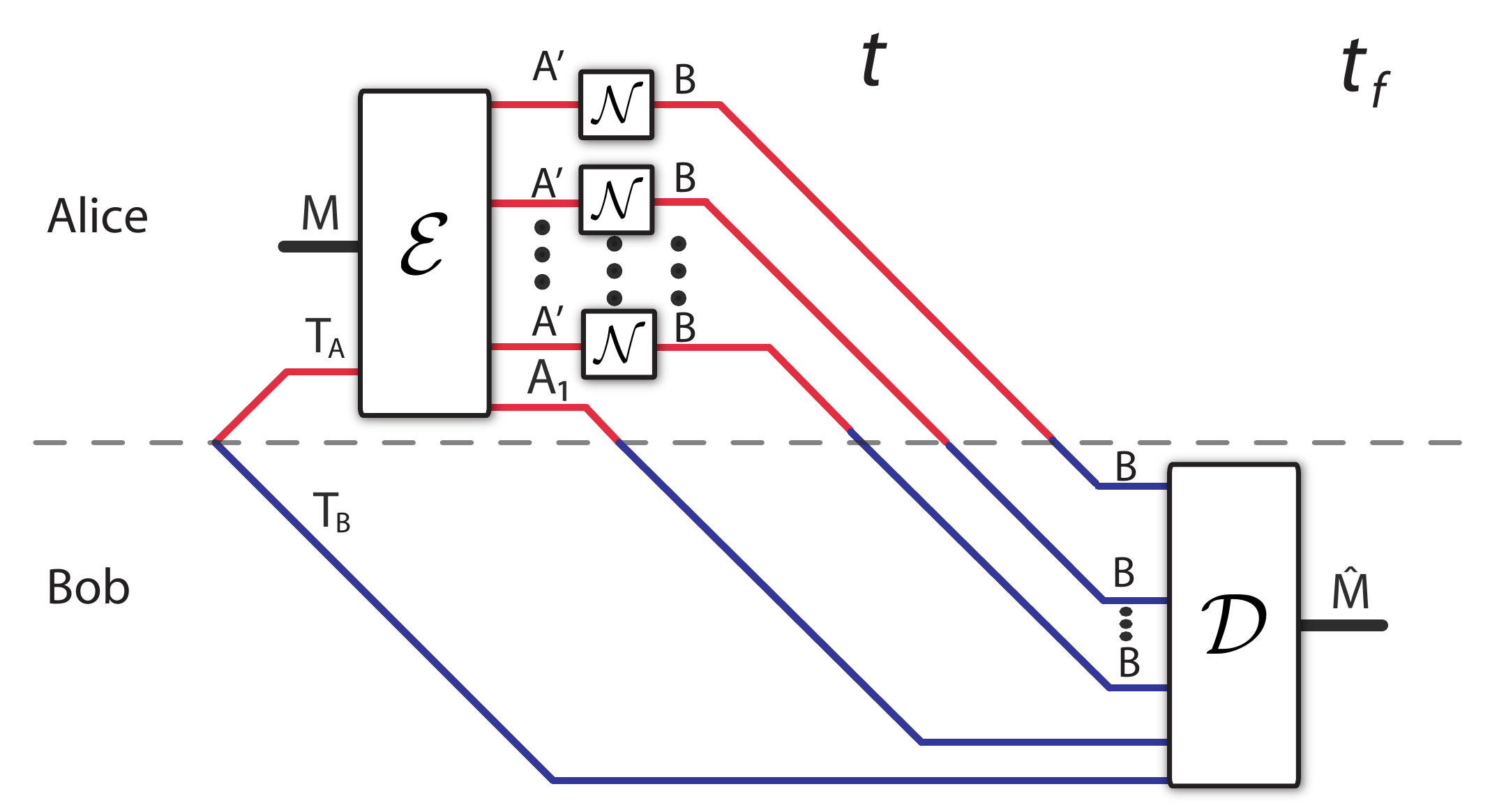}%
\caption{(Color online) The most general protocol for classical communication
with the help of a noisy channel, noiseless entanglement, and noiseless
quantum communication. Alice wishes to communicate a classical message $M$ to
Bob. She shares entanglement with Bob in the form of maximally entangled
states. Her half of the entanglement is in system $T_{A}$ and Bob's half is in
the system $T_{B}$. Alice performs some CPTP\ encoding map $\mathcal{E}$ on
her classical message and her half of the entanglement. The output of this
encoder is a quantum state in some register $A_{1}$ and a large number of
systems $A^{\prime n}$ that are input to the channel. She transmits $A^{\prime
n}$ through the noisy channel and the system $A_{1}$ over noiseless qubit
channels. Bob receives the outputs $B^{n}$ of the channel and the system
$A_{1}$ from the noiseless qubit channels. He combines these with his half of
the entanglement and decodes the classical message that Alice transmits.}%
\label{fig:EQA-class}%
\end{center}
\end{figure}
%EndExpansion

Figure~\ref{fig:EQA-class}\ depicts the most general protocol that generates
classical communication with the help of a noisy channel, noiseless
entanglement, and noiseless quantum communication. The state at the beginning
of the protocol is as follows:%
\[
\omega^{MM^{\prime}T_{A}T_{B}}\equiv\frac{1}{M}\sum_{m}\left\vert
m\right\rangle \left\langle m\right\vert ^{M}\otimes\left\vert m\right\rangle
\left\langle m\right\vert ^{M^{\prime}}\otimes\Phi^{T_{A}T_{B}}.
\]
Alice performs an encoding according to some CPTP\ map $\mathcal{E}%
^{M^{\prime}T_{A}\rightarrow A^{\prime n}A_{1}}$ that takes the classical
system $M^{\prime}$ of size $2^{nC}$\ and the quantum system $T_{A}$ as input.
The map $\mathcal{E}^{M^{\prime}T_{A}\rightarrow A^{\prime n}A_{1}}$ outputs
some systems $A^{\prime n}$ and a quantum state in a register $A_{1}$. The
Schmidt rank of the shared entanglement is $2^{nE}$ and the size of the
quantum system is $2^{nQ}$. The state after the encoder is as follows:%
\[
\omega^{MA^{\prime n}A_{1}T_{B}}\equiv\mathcal{E}^{M^{\prime}T_{A}\rightarrow
AA_{1}}(\omega^{MM^{\prime}T_{A}T_{B}}).
\]
She then transmits the above state through many uses of the noisy channel
$\mathcal{N}$, producing the following state at time $t$:%
\begin{equation}
\omega^{MB^{n}A_{1}T_{B}}\equiv\mathcal{N}^{A^{\prime n}\rightarrow B^{n}%
}(\omega^{MA^{\prime n}A_{1}T_{B}}). \label{eq:EQA-class-state}%
\end{equation}
Note that the state at this point is near to being a state of the form in
(\ref{DD_sigma}), if we define $A\equiv A_{1}T_{B}$ (more on this later).
Finally, Bob combines all of his systems and passes them through a
CPTP\ decoding map $\mathcal{D}^{B^{n}A_{1}T_{B}\rightarrow\hat{M}}$ that
produces his estimate $\hat{M}$\ of the classical message $M$. A protocol is
$\epsilon$-good if the probability of decoding the classical message
incorrectly is low:%
\begin{equation}
\Pr\left\{  \hat{M}\neq M\right\}  \leq\epsilon.
\label{eq:good-EQA-class-code}%
\end{equation}

We first prove the bound in (\ref{eq:+--1}). Consider the following chain of inequalities:%

\begin{align*}
nC &  =H\left(  M\right)  \\
&  =I(M;\hat{M})+H(M|\hat{M})\\
&  \leq I(M;\hat{M})+n\delta^{\prime}\\
&  \leq I\left(  M;A_{1}B^{n}T_{B}\right)  _{\omega}+n\delta^{\prime}\\
&  =I\left(  M;B^{n}|A_{1}T_{B}\right)  _{\omega}+I\left(  M;A_{1}%
T_{B}\right)  _{\omega}+n\delta^{\prime}%
\end{align*}%
\begin{align*}
&  =I\left(  A_{1}T_{B}M;B^{n}\right)  _{\omega}-I\left(  A_{1}T_{B}%
;B^{n}\right)  _{\omega}\\
&  \ \ \ \ \ \ \ +I\left(  M;A_{1}T_{B}\right)  _{\omega}+n\delta^{\prime}\\
&  \leq I\left(  A_{1}T_{B}M;B^{n}\right)  _{\omega}+I\left(  M;A_{1}%
T_{B}\right)  _{\omega}+n\delta^{\prime}\\
&  =I\left(  A_{1}T_{B}M;B^{n}\right)  _{\omega}+I\left(  M;T_{B}\right)
_{\omega}\\
&  \ \ \ \ \ \ \ +I\left(  A_{1};T_{B}|M\right)  _{\omega}+n\delta^{\prime}\\
&  =I\left(  A_{1}T_{B}M;B^{n}\right)  _{\omega}+H\left(  A_{1}|M\right)
_{\omega}\\
&  \ \ \ \ \ \ \ -H\left(  A_{1}|T_{B}M\right)  _{\omega}+n\delta^{\prime}\\
&  \leq I\left(  AM;B^{n}\right)  _{\omega}+n2\left\vert Q\right\vert
+n\delta^{\prime}.
\end{align*}
The first equality follows because the classical message $M$\ is uniform, and
the second equality follows by a straightforward entropic manipulation. The
first inequality follows from the condition in (\ref{eq:good-EQA-class-code}%
)\ and Fano's inequality \cite{CT91}\ with $\delta^{\prime}\equiv\epsilon
C+H_{2}\left(  \epsilon\right)  /n$. The second inequality follows from the
quantum data processing inequality \cite{NC00}. The third and fourth
equalities follow by expanding the quantum mutual information with the chain
rule. The third inequality follows because $I\left(  A_{1}T_{B};B^{n}\right)
_{\omega}\geq0$. The fifth equality follows by expanding the mutual
information $I\left(  M;A_{1}T_{B}\right)  _{\omega}$ with the chain rule. The
last equality follows because $I\left(  M;T_{B}\right)  _{\omega}=0$ for this
protocol and $I\left(  A_{1};T_{B}|M\right)  _{\omega}=H\left(  A_{1}%
|M\right)  _{\omega}-H\left(  A_{1}|T_{B}M\right)  _{\omega}$. The final
inequality follows from the definition $A\equiv A_{1}T_{B}$ and because
$H\left(  A_{1}|M\right)  _{\omega}\leq nQ$ and $\left\vert H\left(
A_{1}|T_{B}M\right)  _{\omega}\right\vert \leq nQ$.

We now prove the bound in (\ref{eq:+--2}). Consider the following chain of
inequalities:%
\begin{align*}
nC  &  \leq I\left(  A_{1}T_{B}M;B^{n}\right)  _{\omega}+I\left(  A_{1}%
;T_{B}|M\right)  _{\omega}+n\delta^{\prime}\\
&  =I\left(  M;B^{n}\right)  _{\omega}+I\left(  A_{1}T_{B};B^{n}|M\right)
_{\omega}+H\left(  A_{1}|M\right)  _{\omega}\\
&  \ \ \ \ \ +H\left(  T_{B}|M\right)  _{\omega}-H\left(  A_{1}T_{B}|M\right)
_{\omega}+n\delta^{\prime}\\
&  =I\left(  M;B^{n}\right)  _{\omega}+I\left(  A_{1}T_{B}\rangle
B^{n}M\right)  _{\omega}+H\left(  A_{1}|M\right)  _{\omega}\\
&  \ \ \ \ \ +H\left(  T_{B}|M\right)  _{\omega}+n\delta^{\prime}\\
&  \leq I\left(  M;B^{n}\right)  _{\omega}+I\left(  A\rangle B^{n}M\right)
_{\omega}+n\left\vert Q\right\vert +n\left\vert E\right\vert +n\delta^{\prime
}.
\end{align*}
The first inequality follows from the fifth equality above and the fact that
$I\left(  M;T_{B}\right)  _{\omega}=0$ for this protocol. The first equality
follows by applying the chain rule for quantum mutual information to $I\left(
A_{1}T_{B}M;B^{n}\right)  _{\omega}$ and by expanding the mutual information
$I\left(  A_{1};T_{B}|M\right)  _{\omega}$. The second equality follows by
noting that%
\[
I\left(  A_{1}T_{B};B^{n}|M\right)  _{\omega}=I\left(  A_{1}T_{B}\rangle
B^{n}M\right)  _{\omega}+H\left(  A_{1}T_{B}|M\right)  _{\omega}.
\]
The last inequality follows from the definition $A\equiv A_{1}T_{B}$ and the
fact that the entropies $H\left(  A_{1}|M\right)  _{\omega}$ and $H\left(
T_{B}|M\right)  _{\omega}$ are less than the logarithm of the dimensions of
the respective systems $A_{1}$ and $T_{B}$.

We should make some final statements concerning this proof. The state in
(\ref{eq:EQA-class-state}) as we have defined it is not quite a state of the
form in (\ref{DD_sigma}) because the encoder could be a general CPTP\ map.
Though, a few arguments demonstrate that a collection of isometric maps works
just as well as a general CPTP\ map, and it then follows that the state in
(\ref{eq:EQA-class-state}) is of the form in (\ref{DD_sigma}). First, consider
that a general CPTP\ map applied to a classical-quantum state of the form in
(\ref{eq:EQA-class-state}) merely acts as a collection of CPTP\ maps
$\{\mathcal{E}_{m}^{T_{A}\rightarrow AA_{1}}\}$ indexed by the classical
message $m$ (see Section~2.3.7 of Ref.~\cite{Yard05a}). Each of these
CPTP\ maps has an isometric extension to some purifying system $E^{\prime}$.
Alice can then perform a complete von Neumann measurement of this system,
producing a classical system $Y$. By the same arguments as in Theorem~7.8 of
Ref.~\cite{DHW05RI} and Appendix~E of Ref.~\cite{HW08GFP}, the protocol can
only improve under this measurement, so that it is sufficient to consider
isometric encoders. Thus, the state in (\ref{eq:EQA-class-state}) is a state
of the form in (\ref{DD_sigma}) and this concludes the converse proof for this octant.

\subsection{Discussion}

\label{sec:EA-insights}The proof of the octant $O^{-+-}$ is one of the more
interesting octants in the direct dynamic trade-off. Its proof directly
answers the following question concerning the use of entanglement-assisted
quantum error correction
\cite{BDH06,BDH06IEEE,HBD07,arx2007wildeEA,arx2007wildeEAQCC,HBD08QLDPC,arx2008wildeUQCC,arx2008wildeGEAQCC}%
:

\begin{quote}
\textit{Why even use entanglement-assisted coding if teleportation is a way to
consume entanglement for the purpose of transmitting quantum information}?
\end{quote}

The proof of the octant $O^{-+-}$ gives a practical answer to the above
question by showing exactly how entanglement-assisted quantum error correction
is useful. We illustrate our arguments by considering the specific case of a
dephasing qubit channel with dephasing parameter $p=0.2$. The quantum capacity
of this channel is around 0.5 qubits per channel use.

First, let us suppose that classical communication is a free resource. Then we
can project the boundary of the octant $O^{-+-}$ and the line of teleportation
into the quadrant $Q^{0+-}$ to compare entanglement-assisted quantum coding to
teleportation. Figure~\ref{fig:EA-vs-TP-both-plots}(a)\ illustrates this
projection. From this figure, we observe that the superior strategy is to
combine quantum communication (LSD) with teleportation or to combine the
classically-enhanced father protocol (CEF) with teleportation. If we do not
take advantage of coding quantum information over the channel, we have to
consume more entanglement in order to achieve the same amount of quantum
communication. The figure demonstrates that a na\"{\i}ve strategy employing
teleportation alone must consume around 0.5 more ebits per channel use for the
same amount of quantum communication that one can obtain by combining LSD\ or
CEF\ with teleportation---this result holds for the $0.2$ qubit dephasing
channel. In general, the extra amount of entanglement consumption necessary is
equal to the quantum capacity of the channel.%
%TCIMACRO{\FRAME{ftbpFU}{6.5518in}{2.8357in}{0pt}{\Qcb{(Color online) (a) The
%figure on the left shows the projection of the octant $O^{-+-}$ into the
%quadrant $Q^{0+-}$. It shows the projection of the line of teleportation in
%this plane and the projection of the line of teleportation starting at the
%classically-enhanced father achievable point. The plot demonstrates that a
%na\"{\i}ve strategy that only teleports and does not take advantage of channel
%coding requires a higher rate of entanglement consumption in order to achieve
%a given rate of quantum communication. The better strategy is to employ
%channel coding (either quantum coding or entanglement-assisted quantum
%coding). (b) The figure on the right shows the full capacity region of a
%strategy that exploits channel coding and the full capacity region of one that
%does not take advantage of channel coding. The latter capacity region is
%always strictly inside the former whenever the channel has a positive
%entanglement-assisted quantum capacity, demonstrating that it is more useful
%to take advantage of channel coding.}}{\Qlb{fig:EA-vs-TP-both-plots}%
%}{ea-vs-tp-both-plots.pdf}{\special{ language "Scientific Word";
%type "GRAPHIC";  maintain-aspect-ratio TRUE;  display "USEDEF";
%valid_file "F";  width 6.5518in;  height 2.8357in;  depth 0pt;
%original-width 11.1665in;  original-height 4.8066in;  cropleft "0";
%croptop "1";  cropright "1";  cropbottom "0";
%filename 'EA-vs-TP-both-plots.pdf';file-properties "XNPEU";}}}%
%BeginExpansion
\begin{figure*}
[ptb]
\begin{center}
\includegraphics[
natheight=4.806600in,
natwidth=11.166500in,
height=2.8357in,
width=6.5518in
]%
{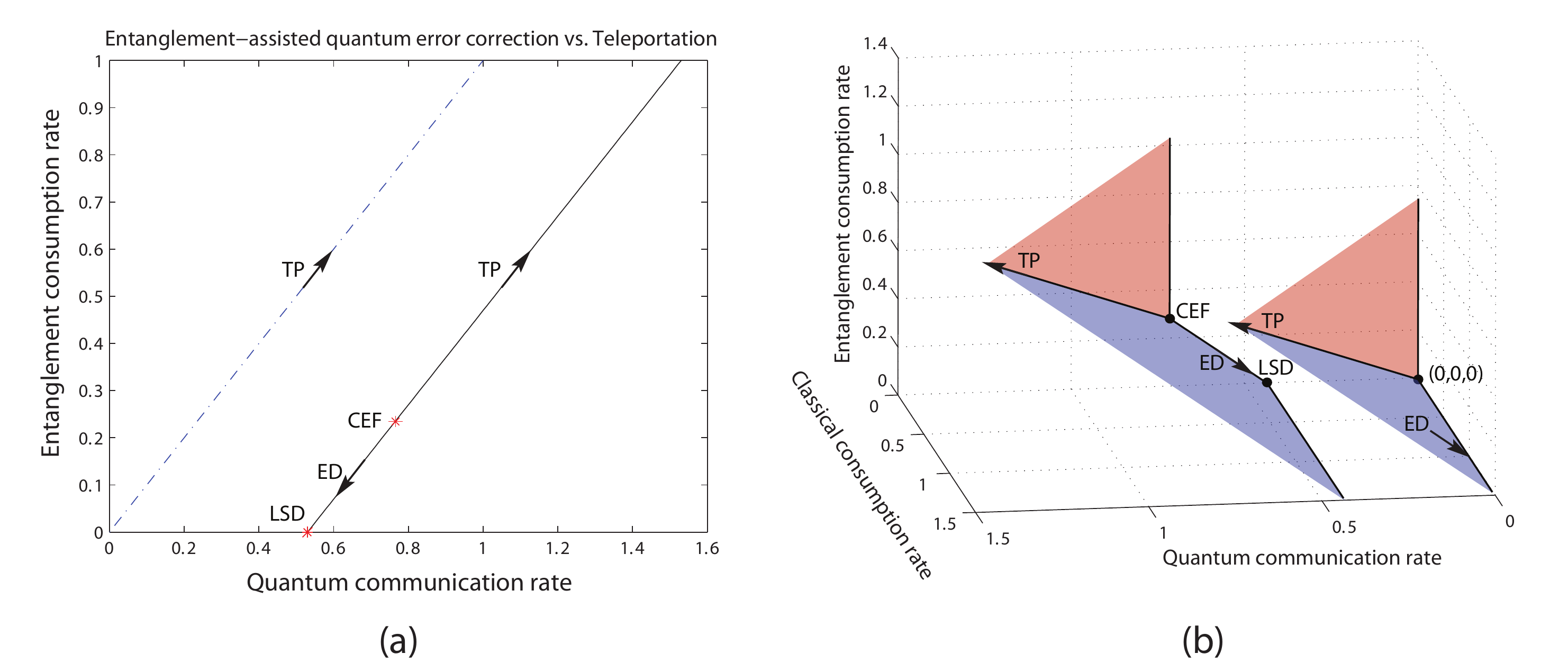}%
\caption{(Color online) (a) The figure on the left shows the projection of the
octant $O^{-+-}$ into the quadrant $Q^{0+-}$. It shows the projection of the
line of teleportation in this plane and the projection of the line of
teleportation starting at the classically-enhanced father achievable point.
The plot demonstrates that a na\"{\i}ve strategy that only teleports and does
not take advantage of channel coding requires a higher rate of entanglement
consumption in order to achieve a given rate of quantum communication. The
better strategy is to employ channel coding (either quantum coding or
entanglement-assisted quantum coding). (b) The figure on the right shows the
full capacity region of a strategy that exploits channel coding and the full
capacity region of one that does not take advantage of channel coding. The
latter capacity region is always strictly inside the former whenever the
channel has a positive entanglement-assisted quantum capacity, demonstrating
that it is more useful to take advantage of channel coding.}%
\label{fig:EA-vs-TP-both-plots}%
\end{center}
\end{figure*}
%EndExpansion

Next, let us consider the case when classical communication is not free. Then
we must consider the full achievable rate region in the octant $O^{-+-}$.
Figure~\ref{fig:EA-vs-TP-both-plots}(b) depicts this scenario by showing both
the achievable rate region that combines the classically-enhanced father
protocol with teleportation and entanglement distribution and the achievable
rate region that combines teleportation and entanglement distribution only. We
observe that the second achievable rate region is strictly inside the first
one whenever the channel has a positive entanglement-assisted quantum
capacity. Thus, the best strategy is not merely to teleport, but it is to
combine channel coding (the classically-enhanced father protocol) with the two
unit protocols of teleportation and entanglement distribution. Furthermore,
with channel coding, one consumes less entanglement or classical communication
in order to achieve a given rate of quantum communication (this improvement
can be dramatic if the entanglement-assisted quantum capacity of the channel
is large).

Sometimes, quantum Shannon theory gives insight into practical quantum error
correction schemes. Devetak's proof of the quantum channel coding theorem
shows that codes with a CSS-like structure are good enough for achieving
capacity \cite{Devetak03}. Another case occurs with the classically-enhanced
father protocol \cite{HW08GFP},\ regarding the structure of
classically-enhanced entanglement-assisted quantum codes, and yet another
occurs in multiple-access quantum coding \cite{YHD05MQAC}, regarding the
structure of multiple-access quantum codes. This octant proves to be another
case where quantum Shannon theory gives some interesting guidelines for the
optimal strategy of a quantum error correction scheme.

\section{Single-Letter Examples}

In this section, we give an example of a shared state for which the static
region single-letterizes and an example of a noisy channel for which the
dynamic region single-letterizes. Single-letterization implies that we do not
have to consider many copies or many channel uses to compute the respective
capacity regions---we only have to consider one copy or channel use, implying
that the computation of the region is a tractable problem.

\subsection{Static Case}

We consider an \textquotedblleft erased state\textquotedblright\ as our
example. We first show that the ($-,-,+$) octant single-letterizes. Then the
full static region single-letterizes by our previous arguments above.

Suppose that the state that Alice and Bob share is the following erased
version of a maximally entangled Bell state:%
\[
\rho^{AB}\equiv\left(  1-\epsilon\right)  \left(  \Phi^{+}\right)
^{AB}+\epsilon\pi^{A}\otimes\left\vert e\right\rangle \left\langle
e\right\vert ^{B},
\]
where%
\[
\left\vert \Phi^{+}\right\rangle ^{AB}\equiv\frac{1}{\sqrt{2}}\left(
\left\vert 00\right\rangle ^{AB}+\left\vert 11\right\rangle ^{AB}\right)  .
\]
This state arises from sending Alice's share of the state $\left\vert \Phi
^{+}\right\rangle ^{AB}$ through an erasure channel that acts as%
\[
\sigma\rightarrow\left(  1-\epsilon\right)  \sigma+\epsilon\left\vert
e\right\rangle \left\langle e\right\vert .
\]
In what follows, all entropies are with respect to the state $\rho$:%
\begin{align*}
H\left(  A\right)   &  =1,\\
H\left(  B\right)   &  =1-\epsilon+H_{2}\left(  \epsilon\right)  ,\\
H\left(  AB\right)   &  =H\left(  E\right)  =\epsilon+H_{2}\left(
\epsilon\right)  .
\end{align*}
Then the following information quantities appearing in the mother
protocol~\cite{DHW05RI} and entanglement distillation~\cite{BDSW96} are as
follows:%
\begin{align*}
I\left(  A\rangle B\right)   &  =1-2\epsilon,\\
\frac{1}{2}I\left(  A;B\right)   &  =1-\epsilon,\\
\frac{1}{2}I\left(  A;E\right)   &  =\epsilon.
\end{align*}

Let us first consider the plane of classically-assisted entanglement
distillation. We can achieve the point $\left(  2\epsilon,0,1-2\epsilon
\right)  $ by the hashing protocol \cite{BDSW96}\ (the classical communication
rate required to achieve this distillation yield is $I\left(  A;E\right)
=2\epsilon$). The rate of entanglement distillation can be no higher than
$1-2\epsilon$, which one can actually prove via the quantum capacity theorem
(the maximally entangled state maximizes the coherent information and
classical communication does not increase the entanglement generation
capacity). Thus, we know that the bound $E\leq1-2\epsilon$ applies for all
$C\geq2\epsilon$ and $Q=0$. Now we should prove that time-sharing between the
origin and the point $\left(  2\epsilon,0,1-2\epsilon\right)  $ is an optimal strategy.

Consider a scheme of entanglement distillation for an erased state with
erasure parameter $\epsilon$. If each party has $n$ halves of the shared
states, then Bob shares $n(1-\epsilon)$ ebits with Alice and the environment
shares $n\epsilon$ ebits with Alice (for the case of large $n$). From these
$n(1-\epsilon)$ shared ebits, Alice and Bob can perform local operations and
forward classical communication to distill $n(1-2\epsilon)$ logical ebits, by
the entanglement distillation result for the erased state. This implies an
optimal \textquotedblleft decoding ratio\textquotedblright\ of $n(1-2\epsilon
)$ decoded ebits for the $n(1-\epsilon)$ physical ebits: $(1-2\epsilon
)/(1-\epsilon)$. Now let us consider some strategy for the erased state that
mixes between the forward classical communication rate of $2\epsilon$ and no
forward classical communication. Suppose that they can achieve the rate triple
$(\lambda2\epsilon,0,\lambda(1-2\epsilon)+\delta)$ where $\delta$ is some
small positive number and $0\leq\lambda\leq1$ (so that this rate triple
represents any point that beats the time-sharing limit). Now if they share $n$
erased states, Alice and Bob share $n(1-\epsilon)$ ebits and the environment
again shares $n\epsilon$ of them with Alice. But this time, Alice and Bob are
not allowed to perform forward classical communication on $n\left(
1-\lambda\right)  (1-\epsilon)$ of them (or a subspace of them of this size).
Thus, these states are not available for decoding ebits. This leaves
$n(1-\epsilon)-n\left(  1-\lambda\right)  (1-\epsilon)=n\lambda(1-\epsilon)$
qubits available for decoding the ebits. If Alice and Bob could decode
$n(\lambda(1-2\epsilon)+\delta)$ logical ebits by local operations and forward
classical communication, this would contradict the optimality of the above
\textquotedblleft decoding ratio\textquotedblright\ because $n(\lambda
(1-2\epsilon)+\delta)/(n\lambda(1-\epsilon))=(1-2\epsilon)/(1-\epsilon
)+\delta/\lambda(1-\epsilon)$ is greater than the optimal decoding ratio
$(1-2\epsilon)/(1-\epsilon)$. Therefore, they must only be able to decode
$n(1-\lambda)(1-2\epsilon)$ logical ebits. This proves that time-sharing
between entanglement distillation and the origin is an optimal strategy for
the erased state.

The above argument then gives that the following region in the $-0+$ plane is
optimal (note that we keep the convention that the rate $R$ is positive even
though the protocol consumes it):%
\begin{align*}
E &  \leq1-2\epsilon\text{ if }C\geq2\epsilon,\\
E &  \leq\frac{C}{2\epsilon}\left(  1-2\epsilon\right)  \text{ if }%
C<2\epsilon.
\end{align*}
We can then obtain a bound on the whole ($-,-,+$) octant by extending this
region by \textquotedblleft inverse\textquotedblright\ teleportation. That is,
the above region, combined with inverse teleportation, gives a bound on all
points in the grandmother octant. Were it not so, then one could combine
points outside this bound with teleportation and achieve points outside the
above region, contradicting the optimality of the region.

For achievability, we can achieve all points in the ($-,-,+$) static octant of
the erased state by combining the mother point $\left(  \epsilon
,0,1-\epsilon\right)  $ with teleportation $\left(  2Q,-Q,-Q\right)  $ and the
wasting of classical communication and quantum communication.
Figure~\ref{fig:grandmother}\ plots this region for the case of an erased
state with parameter $\epsilon=1/4$. It follows that the full region
single-letterizes for the case of an erased state, by our characterization of
the static region in Theorem~\ref{thm:direct-static}.%

%TCIMACRO{\FRAME{ftbpFU}{3.4411in}{2.303in}{0pt}{\Qcb{(Color online) Plot of
%the ($-,-,+$) octant of the static capacity region for the case of an
%\textquotedblleft erased state.\textquotedblright\ The region does not exhibit
%a trade-off and time-sharing between the mother protocol (the point labeled
%\textquotedblleft Mother\textquotedblright), entanglement distillation (the
%point labeled \textquotedblleft Edisti\textquotedblright), and the origin is
%the optimal strategy.}}{\Qlb{fig:grandmother}}{erased-state-grandmother.pdf}%
%{\special{ language "Scientific Word";  type "GRAPHIC";
%maintain-aspect-ratio TRUE;  display "USEDEF";  valid_file "F";
%width 3.4411in;  height 2.303in;  depth 0pt;  original-width 5.7804in;
%original-height 3.8536in;  cropleft "0";  croptop "1";  cropright "1";
%cropbottom "0";
%filename 'erased-state-grandmother.pdf';file-properties "XNPEU";}}}%
%BeginExpansion
\begin{figure}
[ptb]
\begin{center}
\includegraphics[
natheight=3.853600in,
natwidth=5.780400in,
width=3.4411in
]%
{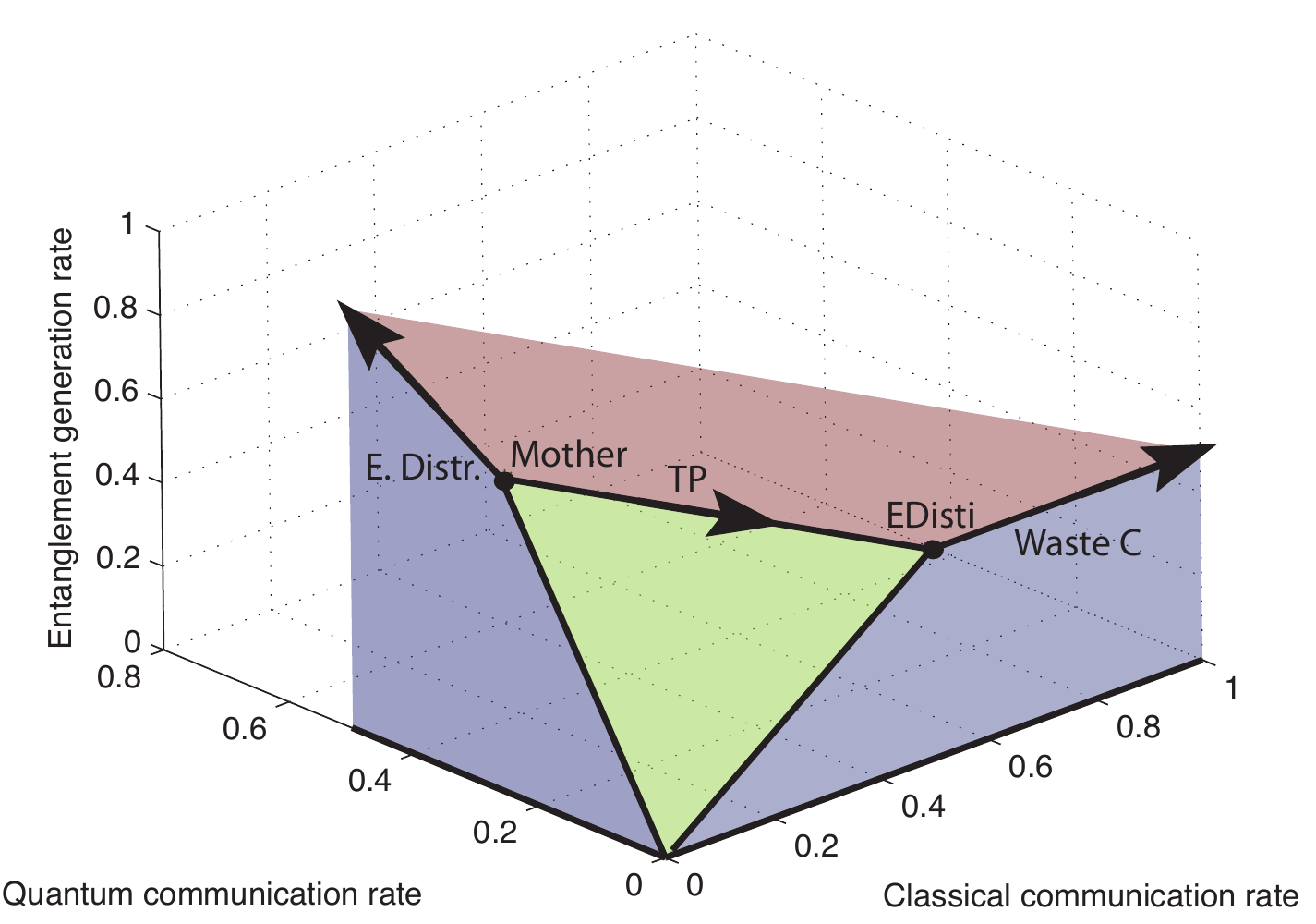}%
\caption{(Color online) Plot of the ($-,-,+$) octant of the static capacity
region for the case of an \textquotedblleft erased state.\textquotedblright%
\ The region does not exhibit a trade-off and time-sharing between the mother
protocol (the point labeled \textquotedblleft Mother\textquotedblright),
entanglement distillation (the point labeled \textquotedblleft
Edisti\textquotedblright), and the origin is the optimal strategy.}%
\label{fig:grandmother}%
\end{center}
\end{figure}
%EndExpansion

\subsection{Dynamic Case}

Our example for the dynamic case is the qubit dephasing channel with dephasing
parameter $p$. We showed in Refs.~\cite{BHTW10,HW08GFP} that the $\left(
+,+,-\right)  $ octant single-letterizes for this channel. The
characterization of the dynamic capacity region in
Theorem~\ref{thm:direct-dynamic-cap}\ shows that the classically-enhanced
father protocol combined with the unit protocols achieves the full region.
Thus, it is sufficient to determine the classically-enhanced father over one
use of the channel, and we thus only need a single use to get the full
capacity region. Though, Ref.~\cite{WH10} in fact gives a simplified, direct
argument that this is the case.%
%TCIMACRO{\FRAME{ftbpFU}{3.4411in}{2.5901in}{0pt}{\Qcb{(Color online) Plot of
%the full dynamic capacity region for all octants for the case of a qubit
%dephasing channel with dephasing parameter $p=0.2$. We can see that it is
%merely the unit resource capacity region translated along the
%classically-enhanced father trade-off curve.}}{\Qlb{fig:full-triple-dynamic}%
%}{full-triple-dynamic.pdf}{\special{ language "Scientific Word";
%type "GRAPHIC";  maintain-aspect-ratio TRUE;  display "USEDEF";
%valid_file "F";  width 3.4411in;  height 2.5901in;  depth 0pt;
%original-width 6.7464in;  original-height 5.0669in;  cropleft "0";
%croptop "1";  cropright "1";  cropbottom "0";
%filename 'full-triple-dynamic.pdf';file-properties "XNPEU";}}}%
%BeginExpansion
\begin{figure}
[ptb]
\begin{center}
\includegraphics[
natheight=5.066900in,
natwidth=6.746400in,
height=2.5901in,
width=3.4411in
]%
{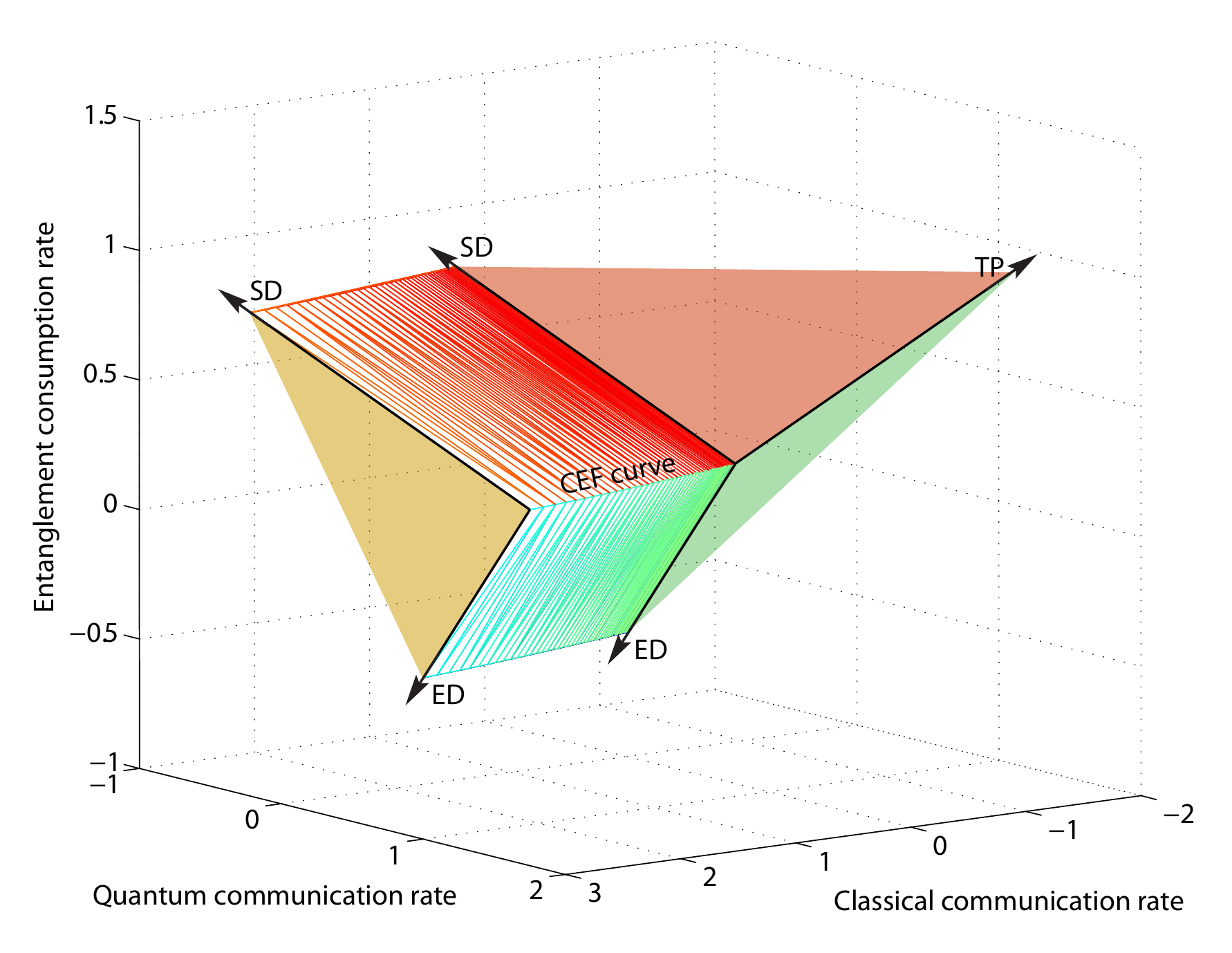}%
\caption{(Color online) Plot of the full dynamic capacity region for all
octants for the case of a qubit dephasing channel with dephasing parameter
$p=0.2$. We can see that it is merely the unit resource capacity region
translated along the classically-enhanced father trade-off curve.}%
\label{fig:full-triple-dynamic}%
\end{center}
\end{figure}
%EndExpansion

\section{Conclusion}

We have provided a unifying treatment of many of the important results in
quantum Shannon theory. Our first result is a solution of the unit resource
capacity region---the optimal strategy mixes super-dense coding,
teleportation, and entanglement distribution. Our next result is the full
triple trade-off for the static scenario where a sender and receiver share a
noisy bipartite state. The coding strategy combines the classically-assisted
state redistribution protocol with the three unit protocols. Our last result
is a solution of the direct dynamic capacity theorem---the scenario where a
sender and receiver have access to a large number of independent uses of a
noisy quantum channel. The coding strategy combines the classically-enhanced
father protocol with the three unit protocols.

The discussion in Section~\ref{sec:EA-insights} demonstrates another case
where quantum Shannon theory has practical implications for quantum error
correction schemes. We are able to determine how one benefits from
entanglement-assisted coding versus teleportation.

Our work was originally inspired by the work in \cite{AH03GRSP} in which the
authors solved a triple trade-off problem called generalized remote state
preparation (GRSP). The relation between our capacity regions and theirs is
yet unknown due to incompatible definitions of a resource \cite{DHW05RI}. The
GRSP uses ``pseudo-resources'' that resemble our definition of a resource but
fail to satisfy the quasi-i.i.d. requirement. We can possibly remedy this by
generalizing our definition of a resource.

In this article, we have discussed the triple trade-off scenario for when a
protocol consumes a noisy resource to generate noiseless resources. An
interesting open research question is the triple trade-off scenario for when a
protocol generates or \textit{simulates} a noisy resource rather than consumes
it. A special case of this type of triple trade-off is the quantum reverse
Shannon theorem, because the protocol corresponding to it consumes classical
communication and entanglement to simulate a noisy channel
\cite{BDHSW05,ADHW06FQSW,Devetak06Duality}. The discussion in the last section
of Ref.~\cite{BSST01}\ speaks of the usefulness of the quantum reverse Shannon
theorem and its role in simplifying quantum Shannon theory. One could imagine
several other protocols that would arise as special cases of the triple
trade-off for simulating a noisy resource, but the usefulness of such triple
trade-offs is unclear to us at this point.

An interesting open research question concerns the triple trade-offs for the
static and dynamic cases where the noiseless resources are instead public
classical communication, private classical communication, and a secret key. We
expect the proof strategies to be similar to those in this article, but the
capacity regions should be different from those found here. A useful protocol
is the publicly-enhanced secret-key-assisted private classical communication
protocol \cite{HW09}, an extension of the private father protocol
\cite{hsieh:042306}. This protocol gives the initial steps for finding the
full triple trade-off of the dynamic case. The static case should employ
previously found protocols, such as that for secret key distillation. As a
last suggestion, one might also consider using these techniques for
determining the optimal sextuple trade-offs in multiple-access coding
\cite{YHD05MQAC,itit2008hsieh}\ and broadcast channel coding
\cite{YHD2006,DH2006}.

Another interesting open research question concerning our results is the
long-standing \textquotedblleft single-letterizable\textquotedblright\ issue
that plagues most capacity results in quantum Shannon theory. Our capacity
formulas are regularized expressions---the implication of regularization is
that the evaluation of the rate region with a regularized expression is
intractable, requiring an optimization over an infinite number of uses of a
channel for the dynamic case and over all instruments for the static case.
Additionally, prior work shows that two different regularized capacity
expressions can coincide asymptotically \cite{YHD05MQAC}, even though the
corresponding finite capacity formulas trace out different finite
approximations of the rate region. Thus, regularized results in quantum
Shannon theory present a problem for determining the true characterization of
a capacity region. We have shown examples of states and channels for which the
regions single-letterize, but there is always the possibility of uncovering
some formulas for the region that give a single-letter characterization.

\section{Acknowledgements}

The authors thank Igor Devetak, Patrick Hayden, and Debbie Leung for initial
discussions during the development of this project, and the anonymous referee
for bringing the classically-assisted quantum state redistribution protocol to
our attention. M.M.W. acknowledges support from the National Research
Foundation \& Ministry of Education, Singapore, from the internal research and
development grant SAIC-1669 of Science Applications International Corporation,
and from the MDEIE (Qu\'{e}bec) PSR-SIIRI international collaboration grant.
M.M.W.~also acknowledges the hospitality of the ERATO-SORST\ project, where he
and M.-H.H. were able to complete this research.

\appendices

\section{$\boldsymbol{(-,-,+)}$ Octant of the Direct Static Capacity Region}

\label{sec:converse_--+_static}The converse proof for this octant corresponds
to the classically-assisted state redistribution protocol. We employ an
information-theoretic argument.

We consider the most general classically-assisted state redistribution
protocol for proving the converse theorem (illustrated in
Figure~\ref{fig:class-assist-mother}).%
%TCIMACRO{\FRAME{ftbpFU}{3.3399in}{2.2407in}{0pt}{\Qcb{(Color online) The
%figure above illustrates the most general protocol for classically- and
%quantum-communication-assisted entanglement distillation. Alice, Bob, and the
%reference system share a state $\left\vert \psi\right\rangle ^{A^{n}B^{n}%
%E^{n}}$. Alice performs a quantum instrument $\QTR{cal}{T}^{A^{n}\rightarrow
%A^{\prime}A_{1}ME^{\prime}}$ on her system $A^{n}$. Alice transmits $M$ and
%$A_{1}$ to Bob. Bob performs a decoding operation $\QTR{cal}{D}^{B^{n}%
%MA^{\prime}\rightarrow B_{1}}$ that outputs the system $B_{1}$. The result of
%the protocol is a state close to the maximally entangled state $\Phi
%^{A_{1}B_{1}}$ on systems $A_{1}$ and $B_{1}$.}}{\Qlb{fig:class-assist-mother}%
%}{class-assisted-mother.pdf}{\special{ language "Scientific Word";
%type "GRAPHIC";  maintain-aspect-ratio TRUE;  display "USEDEF";
%valid_file "F";  width 3.3399in;  height 2.2407in;  depth 0pt;
%original-width 8.8332in;  original-height 5.8997in;  cropleft "0";
%croptop "1";  cropright "1";  cropbottom "0";
%filename 'class-assisted-mother.pdf';file-properties "XNPEU";}%
%}}%
%BeginExpansion
\begin{figure}
[ptb]
\begin{center}
\includegraphics[
natheight=5.899700in,
natwidth=8.833200in,
height=2.2407in,
width=3.3399in
]%
{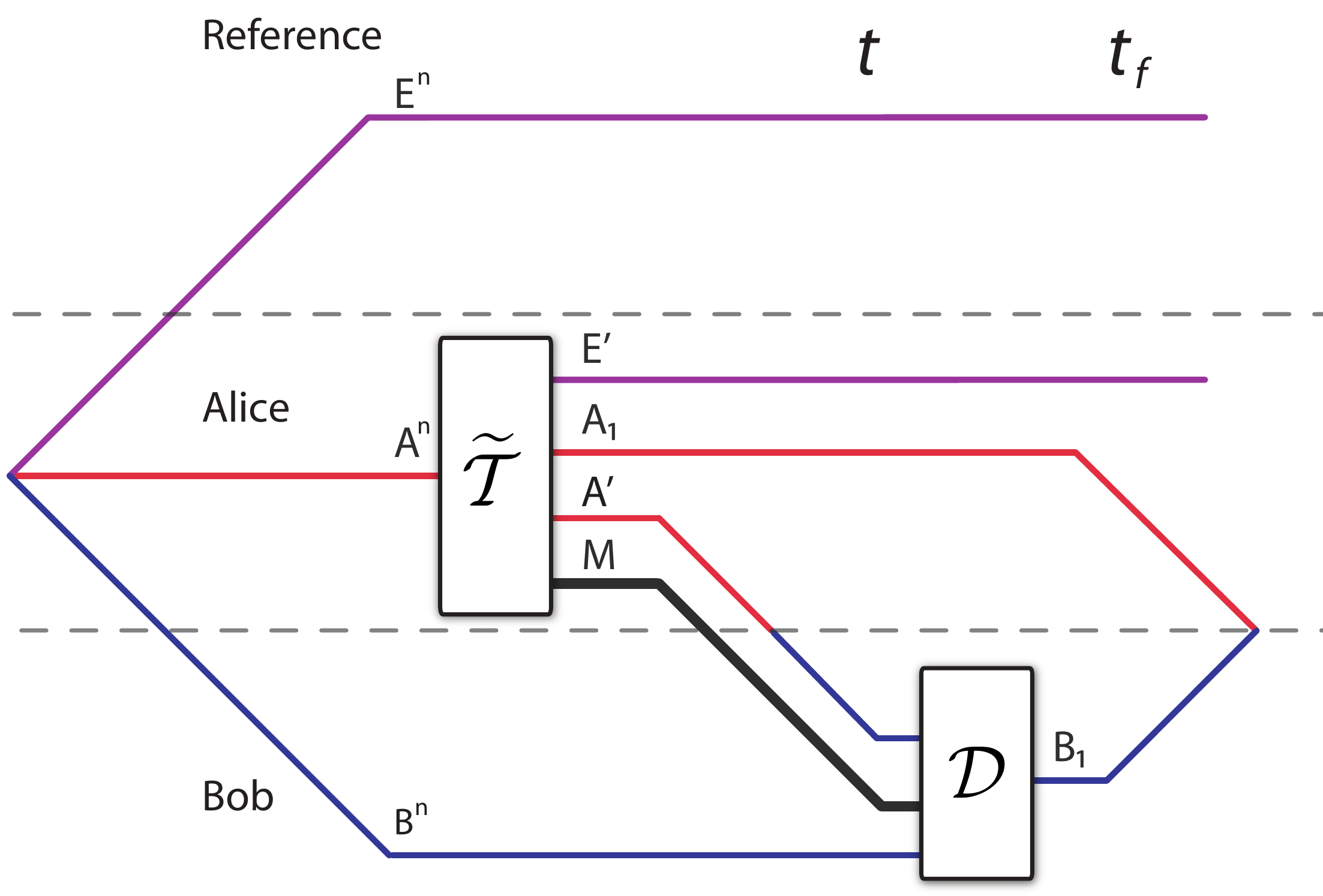}%
\caption{(Color online) The figure above illustrates the most general protocol
for classically- and quantum-communication-assisted entanglement distillation.
Alice, Bob, and the reference system share a state $\left\vert \psi
\right\rangle ^{A^{n}B^{n}E^{n}}$. Alice performs a quantum instrument
$\mathcal{T}^{A^{n}\rightarrow A^{\prime}A_{1}ME^{\prime}}$ on her system
$A^{n}$. Alice transmits $M$ and $A_{1}$ to Bob. Bob performs a decoding
operation $\mathcal{D}^{B^{n}MA^{\prime}\rightarrow B_{1}}$ that outputs the
system $B_{1}$. The result of the protocol is a state close to the maximally
entangled state $\Phi^{A_{1}B_{1}}$ on systems $A_{1}$ and $B_{1}$.}%
\label{fig:class-assist-mother}%
\end{center}
\end{figure}
%EndExpansion

Suppose Alice and Bob share many copies $\rho^{A^{n}B^{n}}$ of the noisy
bipartite state $\rho^{AB}$. The purification of the state $\rho^{A^{n}B^{n}}$
is $\psi^{A^{n}B^{n}E^{n}}$ where $E$ is the reference party. Alice performs a
quantum instrument $\widetilde{\mathcal{T}}^{A^{n}\rightarrow A^{\prime}%
A_{1}M}$ on her system $A^{n}$ and produces the quantum systems $A^{\prime}$
and $A_{1}$ and the classical system $M$. The quantum system $A^{\prime}$ has
size $2^{nQ}$, $A_{1}$ has size $2^{nE}$, and $M$ has size $2^{nC}$. We
consider an extension $\widetilde{\mathcal{T}}^{A^{n}\rightarrow A^{\prime
}A_{1}ME^{\prime}}$ of the quantum instrument $\widetilde{\mathcal{T}}%
^{A^{n}\rightarrow A^{\prime}A_{1}M}$ in what follows (see the discussion of
the CP\ formalism in Ref.~\cite{DHW05RI}). Alice sends $A^{\prime}$ through a
noiseless quantum channel and sends $M$ through a noiseless classical channel.
The resulting state, at time $t$ in Figure~\ref{fig:class-assist-mother}, is
as follows:%
\[
\omega^{A^{\prime}A_{1}ME^{\prime}B^{n}E^{n}}\equiv\widetilde{\mathcal{T}%
}^{A^{n}\rightarrow A^{\prime}A_{1}ME^{\prime}}\left(  \psi^{A^{n}B^{n}E^{n}%
}\right)  .
\]
Define $A\equiv A^{\prime}A_{1}$ so that the above state is a particular
$n^{\text{th}}$ extension of the state in (\ref{eq:instrument-state}). Bob
receives the systems $A^{\prime}$ and $M$. The most general decoding operation
that Bob can perform on his three registers $A^{\prime}$, $B^{n}$, and $M$ is
a conditional quantum decoder $\mathcal{D}^{MA^{\prime}B^{n}\rightarrow B_{1}%
}$ consisting of a collection $\{\mathcal{D}_{m}^{A^{\prime}B^{n}\rightarrow
B_{1}}\}_{m}$\ of CPTP maps. The state of Bob's system after the conditional
quantum decoder $\mathcal{D}^{MA^{\prime}B^{n}\rightarrow B_{1}}$ (at time
$t_{f}$ in Figure~\ref{fig:class-assist-mother}) is as follows:%
\[
\left(  \omega^{\prime}\right)  ^{A_{1}E^{\prime}B_{1}E^{n}}\equiv
\mathcal{D}^{MA^{\prime}B^{n}\rightarrow B_{1}}(\omega^{A^{\prime}%
A_{1}ME^{\prime}B^{n}E^{n}}).
\]

Suppose that an $\left(  n,C+\delta,Q+\delta,E-\delta,\epsilon\right)  $
classically-assisted state redistribution protocol as given above exists. We
prove that the following bounds apply to the elements of its rate triple
$\left(  C+\delta,Q+\delta,E-\delta\right)  $:%
\begin{align}
C+2Q+\delta &  \geq\frac{1}{n}\left(  I(M;E^{n}|B^{n})_{\omega}+I(A^{\prime
}A_{1};E^{n}|E^{\prime}M)_{\omega}\right)  ,\label{cgf1}\\
E-\delta &  \leq C+Q\label{cgf2}\\
&  +\frac{1}{n}\left(  I(A^{\prime}A_{1}\rangle B^{n}M)_{\omega}%
-I(M;E^{n}|B^{n})_{\omega}\right)  ,\nonumber\\
E-\delta &  \leq\frac{1}{n}I(A^{\prime}A_{1}\rangle B^{n}M)_{\omega
}+Q,\label{cgf4}%
\end{align}
for any $\epsilon,\delta>0$ and all sufficiently large $n$. We redefine the
system $A^{\prime}\equiv A^{\prime}A_{1}$ so that our expression above matches
that in the direct static capacity theorem.

In the ideal case, the protocol produces the maximally entangled state
$\Phi^{A_{1}B_{1}}$. So for our case, the following inequality%
\begin{equation}
\left\Vert \left(  \omega^{\prime}\right)  ^{A_{1}B_{1}}-\Phi^{A_{1}B_{1}%
}\right\Vert _{1}\leq\epsilon, \label{eq:converse-good-q-comm}%
\end{equation}
holds because the protocol is $\epsilon$-good for entanglement generation.

We first prove the bound in (\ref{cgf4}):%
\begin{align*}
n(E-\delta)  &  =I(A_{1}\rangle B_{1})_{\Phi^{A_{1}B_{1}}}\\
&  \leq I(A_{1}\rangle B_{1})_{\omega^{\prime}}+n\delta^{\prime}\\
&  \leq I(A_{1}\rangle A^{\prime}B^{n}M)_{\omega}+n\delta^{\prime}\\
&  \leq I(A_{1}A^{\prime}\rangle B^{n}M)_{\omega}+H\left(  A^{\prime
}|M\right)  _{\omega}+n\delta^{\prime}\\
&  \leq I(A_{1}A^{\prime}\rangle B^{n}M)_{\omega}+nQ+n(\delta^{\prime}+\delta)
\end{align*}
The first equality follows from evaluating the coherent information of the
maximally entangled state $\Phi^{A_{1}B_{1}}$, the first inequality follows
from the Alicki-Fannes' inequality \cite{0305-4470-37-5-L01} where we define
$\delta^{\prime}\equiv4E\epsilon+H(\epsilon)/n$, the second inequality follows
from quantum data processing \cite{SN96}, the third inequality follows because
conditioning reduces entropy ($H\left(  A^{\prime}|B^{n}M\right)  _{\omega
}\leq H\left(  A^{\prime}|M\right)  _{\omega}$), and the last inequality
follows because $H\left(  A^{\prime}|M\right)  _{\omega}\leq H(A^{\prime
})_{\omega}\leq n(Q+\delta)$ (the entropy of system $A^{\prime}$ has to be
less than the log of the dimension of the sytem).

We next prove the lower bound in (\ref{cgf1}) on the classical and quantum
communication consumption rate:%
\begin{align*}
&  n(C+2Q+\delta)\\
&  \geq H\left(  M\right)  _{\omega}+Q+E-I\left(  A^{\prime}A_{1}\rangle
B^{n}M\right)  _{\omega}\\
&  \geq H\left(  M|B^{n}\right)  _{\omega}+H\left(  A^{\prime}|M\right)
_{\omega}+H\left(  A_{1}|M\right)  _{\omega}\\
&  \ \ \ \ \ \ \ \ -H\left(  B^{n}|M\right)  _{\omega}+H\left(  A^{\prime
}A_{1}B^{n}|M\right)  _{\omega}\\
&  \geq H\left(  M|B^{n}\right)  _{\omega}-H\left(  M|B^{n}E^{n}\right)
_{\omega}+H\left(  A^{\prime}A_{1}|M\right)  _{\omega}\\
&  \ \ \ \ \ \ \ \ -H\left(  A^{\prime}A_{1}E^{n}E^{\prime}|M\right)
_{\omega}+H\left(  E^{n}E^{\prime}|M\right)  _{\omega}\\
&  =I\left(  M;E^{n}|B^{n}\right)  _{\omega}+I\left(  A^{\prime}A_{1}%
;E^{n}E^{\prime}|M\right)  _{\omega}\\
&  \geq I\left(  M;E^{n}|B^{n}\right)  _{\omega}+I\left(  A^{\prime}%
A_{1};E^{n}|E^{\prime}M\right)  _{\omega}.
\end{align*}
The first inequality follows because the entropy of $M$ is less than that of
the uniform distribution and by exploiting the inequality in (\ref{cgf4}). The
second inequality follows because conditioning reduces entropy ($H\left(
M\right)  \geq H\left(  M|B^{n}\right)  $), because $Q\geq H\left(  A^{\prime
}|M\right)  _{\omega}$ and $E\geq H\left(  A_{1}|M\right)  _{\omega}$, and by
expanding the coherent information $I\left(  A^{\prime}A_{1}\rangle
B^{n}M\right)  _{\omega}$. The third inequality follows because $H\left(
M|B^{n}E^{n}\right)  _{\omega}\geq0$, from subadditivity, and because the
state on $B^{n}A^{\prime}A_{1}E^{n}E^{\prime}$ is pure when conditioned on the
classical variable $M$. The sole equality follows by collecting terms, and the
last inequality follows because $I\left(  A^{\prime}A_{1};E^{n}E^{\prime
}|M\right)  _{\omega}=I\left(  A^{\prime}A_{1};E^{\prime}|M\right)  _{\omega
}+I\left(  A^{\prime}A_{1};E^{n}|E^{\prime}M\right)  _{\omega}$ $I\left(
A^{\prime}A_{1};E^{\prime}|M\right)  _{\omega}\geq0$.

The inequality in (\ref{cgf1}) follows by exploiting (\ref{cgf4}) and that
$nC\geq H\left(  M\right)  _{\omega}\geq I\left(  M;E^{n}|B^{n}\right)
_{\omega}$.

\section{$\boldsymbol{(-,+,+)}$ Octant of the Direct Static Capacity Region}

\label{sec:-++static}The technique for handling this octant is similar to the
technique for handling the octant $\left(  +,-,+\right)  $. We give the full
proof for completeness. Define%
\[
\mathcal{C}_{\text{DS}}^{-++}\left(  \rho\right)  \equiv\mathcal{C}%
_{\text{DS}}\left(  \rho\right)  \cap O^{-++},
\]
and recall the definition of $\mathcal{C}_{\text{DS}}^{-+0}\left(
\rho\right)  $ in (\ref{eq:def-cap-DS--+0}). Recall the line of teleportation
$L_{\text{TP}}$ as defined in (\ref{eq:line-TP}).

Define the following maps%
\begin{align*}
f  &  :S\rightarrow(S+L_{\text{TP}})\cap Q^{-+0},\\
\hat{f}  &  :S\rightarrow(S-L_{\text{TP}})\cap O^{-++}.
\end{align*}
The map $f$ translates the set $S$ in the teleportation direction and keeps
the points in the $Q^{-+0}$ quadrant. The map $\hat{f}$, in a sense, undoes
the effect of $f$ by moving the set $S$ back to the octant $O^{-++}$.

The inclusion $\mathcal{C}_{\text{DS}}^{-++}\left(  \rho\right)  \subseteq
\hat{f}(f(\mathcal{C}_{\text{DS}}^{-++}\left(  \rho\right)  ))$ holds because%
\begin{align}
&  \mathcal{C}_{\text{DS}}^{-++}\left(  \rho\right) \nonumber\\
&  =\mathcal{C}_{\text{DS}}^{-++}\left(  \rho\right)  \cap O^{-++}\nonumber\\
&  \subseteq(((\mathcal{C}_{\text{DS}}^{-++}\left(  \rho\right)
+L_{\text{TP}})\cap Q^{-+0})-L_{\text{TP}})\cap O^{-++}\nonumber\\
&  =(f(\mathcal{C}_{\text{DS}}^{-++}\left(  \rho\right)  )-L_{\text{TP}})\cap
O^{-++}\nonumber\\
&  =\hat{f}(f(\mathcal{C}_{\text{DS}}^{-++}\left(  \rho\right)  )).
\label{DSnpp1}%
\end{align}
The first set equivalence is obvious from the definition of $\mathcal{C}%
_{\text{DS}}^{-++}\left(  \rho\right)  $. The first inclusion follows from the
following logic. Pick any point $a\equiv(C,Q,E)\in\mathcal{C}_{\text{DS}%
}^{-++}\left(  \rho\right)  \cap O^{-++}$ and a particular point
$b\equiv\left(  -2E,E,-E\right)  \in L_{\text{TP}}$. It follows that the point
$a+b\in(\mathcal{C}_{\text{DS}}^{-++}\left(  \rho\right)  +L_{\text{TP}})\cap
Q^{-+0}$. We then pick a point $-b\equiv\left(  2E,-E,E\right)  \in
-L_{\text{TP}}$. It follows that $a+b-b\in(((\mathcal{C}_{\text{DS}}%
^{-++}\left(  \rho\right)  +L_{\text{TP}})\cap Q^{-+0})-L_{\text{TP}})\cap
O^{-++}$ and that $a+b-b=\left(  C,Q,E\right)  =a$. The first inclusion thus
holds because every point in $\mathcal{C}_{\text{DS}}^{-++}\left(
\rho\right)  \cap O^{-++}$ is in $(((\mathcal{C}_{\text{DS}}^{-++}\left(
\rho\right)  +L_{\text{TP}})\cap Q^{-+0})-L_{\text{TP}})\cap O^{-++}$. The
second set equivalence follows from the definition of $f$ and the third set
equivalence follows from the definition of $\hat{f}$.

It is operationally clear that the following inclusion holds
\begin{equation}
f(\mathcal{C}_{\text{DS}}^{-++}\left(  \rho\right)  )\subseteq\mathcal{C}%
_{\text{DS}}^{-+0}\left(  \rho\right)  , \label{DSnpp2}%
\end{equation}
because the mapping $f$ converts any achievable point $a\in\mathcal{C}%
_{\text{DS}}^{-++}\left(  \rho\right)  $ to an achievable point in
$\mathcal{C}_{\text{DS}}^{-+0}\left(  \rho\right)  $ by consuming all of the
entanglement in $a$ with teleportation.

The converse proof of the noisy teleportation \cite{DHW05RI} protocol is
useful for us:
\begin{equation}
\mathcal{C}_{\text{DS}}^{-+0}\left(  \rho\right)  \subseteq\widetilde
{\mathcal{C}}_{\text{DS}}^{-+0}\left(  \rho\right)  . \label{DSnpp4}%
\end{equation}
%by combining the relations in (\ref{eq:def-ach-DS--+0}),
%(\ref{eq:def-cap-DS--+0}), and (\ref{eq:mother-NTP}).

The inclusion $\hat{f}(\widetilde{\mathcal{C}}_{\text{DS}}^{-+0}\left(
\rho\right)  )\subseteq\widetilde{\mathcal{C}}_{\text{DS}}^{-++}\left(
\rho\right)  $ holds because%
\begin{align}
&  \hat{f}(\widetilde{\mathcal{C}}_{\text{DS}}^{-+0}\left(  \rho\right)
)\nonumber\\
&  =(((\widetilde{\mathcal{C}}_{\text{CASR}}\left(  \rho\right)
+L_{\text{TP}})\cap Q^{-+0})-L_{\text{TP}})\cap O^{-++}\nonumber\\
&  \subseteq((\widetilde{\mathcal{C}}_{\text{CASR}}\left(  \rho\right)
+L_{\text{TP}})-L_{\text{TP}})\cap O^{-++}\nonumber\\
&  =((\widetilde{\mathcal{C}}_{\text{CASR}}\left(  \rho\right)  +L_{\text{TP}%
})\cap O^{-++})\cup\nonumber\\
&  \ \ \ \ \ \ \ \ \ \ \ \ \ \ \ \ \ \ \ \ \ \ \ \ \ \ \ \ \ \ ((\widetilde
{\mathcal{C}}_{\text{CASR}}\left(  \rho\right)  -L_{\text{TP}})\cap
O^{-++})\nonumber\\
&  \subseteq\widetilde{\mathcal{C}}_{\text{DS}}^{-++}\left(  \rho\right)  .
\label{DSnpp3}%
\end{align}
The first set equivalence follows by definition. The first inclusion follows
by dropping the intersection with $Q^{-+0}$. The second set equivalence
follows because $(\widetilde{\mathcal{C}}_{\text{CASR}}\left(  \rho\right)
+L_{\text{TP}})-L_{\text{TP}}=(\widetilde{\mathcal{C}}_{\text{CASR}}\left(
\rho\right)  +L_{\text{TP}})\cup(\widetilde{\mathcal{C}}_{\text{CASR}}\left(
\rho\right)  -L_{\text{TP}})$, and the last inclusion follows because
$(\widetilde{\mathcal{C}}_{\text{CASR}}\left(  \rho\right)  -L_{\text{TP}%
})\cap O^{-++}=\left(  0,0,0\right)  $.

Putting (\ref{DSnpp1}), (\ref{DSnpp2}), (\ref{DSnpp4}), and (\ref{DSnpp3})
together, we have the following inclusion:%
\begin{multline*}
\mathcal{C}_{\text{DS}}^{-++}\left(  \rho\right)  \subseteq\hat{f}%
(f(\mathcal{C}_{\text{DS}}^{-++}\left(  \rho\right)  ))\\
\subseteq\hat{f}(\mathcal{C}_{\text{DS}}^{-+0}\left(  \rho\right)
)\subseteq\hat{f}(\widetilde{\mathcal{C}}_{\text{DS}}^{-+0}\left(
\rho\right)  )\subseteq\widetilde{\mathcal{C}}_{\text{DS}}^{-++}\left(
\rho\right)  .
\end{multline*}
The above inclusion $\mathcal{C}_{\text{DS}}^{-++}\subseteq\widetilde
{\mathcal{C}}_{\text{DS}}^{-++}$ is the statement of the converse theorem for
this octant.

\section{$\boldsymbol{(-,+,-)}$ Octant of the Direct Static Capacity Region}

\label{sec:-+-static}The proof of this octant is similar to the proof of the
octant $(+,-,-)$. We first need the following additivity lemma.

\begin{lemma}
\label{lem:additivity--+0}The following inclusion holds%
\[
\mathcal{C}_{\text{DS}}^{-+0}(\rho\otimes\Phi^{|E|})\subseteq\widetilde
{\mathcal{C}}_{\text{DS}}^{-+0}(\rho)+\widetilde{\mathcal{C}}_{\text{DS}%
}^{-+0}(\Phi^{|E|}).
\]

\end{lemma}

\begin{IEEEproof}
Teleportation induces a linear bijection $f:\widetilde{\mathcal{C}}_{\text{DS}}%
^{-0+}\left(  \rho\right)  \rightarrow\widetilde{\mathcal{C}}_{\text{DS}}^{-+0}\left(
\rho\right)  $ between the entanglement distillation achievable region
$\widetilde{\mathcal{C}}_{\text{DS}}^{-0+}\left(  \rho\right)  $\ and the noisy
teleportation achievable region $\widetilde{\mathcal{C}}_{\text{DS}}^{-+0}\left(
\rho\right)  $ \cite{DHW05RI}.\ The bijection $f$ behaves as follows for every
point $(C,0,E)\in\widetilde{\mathcal{C}}_{\text{DS}}^{-0+}\left(  \rho\right)  $,%
\[
f:(C,0,E)\rightarrow(C-2E,E,0).
\]
The following relation holds%
\begin{equation}
f(\widetilde{\mathcal{C}}_{\text{DS}}^{-0+}\left(  \rho\right)  )=\widetilde
{\mathcal{C}}_{\text{DS}}^{-+0}\left(  \rho\right)  , \label{DSnpn2}%
\end{equation}
because applying teleportation to entanglement distillation gives noisy
teleportation \cite{DHW05RI}. %The following relations hold%
%\begin{equation}%
%\begin{split}
%\mathcal{C}_{\text{DS}}^{-0+}\left(  \rho\right)   &  =\widetilde{\mathcal{C}}_{\text{DS}}%
%^{-0+}\left(  \rho\right)  ,\\
%\mathcal{C}_{\text{DS}}^{-+0}\left(  \rho\right)   &  =\widetilde{\mathcal{C}}_{\text{DS}}%
%^{-+0}\left(  \rho\right)  ,
%\end{split}
%\label{DSnpn1}%
%\end{equation}
%because the respective entanglement distillation and noisy teleportation
%achievable regions are optimal \cite{DHW05RI}.
The inclusion $\mathcal{C}_{\text{DS}}^{-0+}(\rho\otimes\Phi^{|E|})\subseteq
f^{-1}(\widetilde{\mathcal{C}}_{\text{DS}}^{-+0}(\rho)+\widetilde{\mathcal{C}}_{\text{DS}}%
^{-+0}(\Phi^{|E|}))$ holds because%
\begin{align*}
\mathcal{C}_{\text{DS}}^{-0+}(\rho\otimes\Phi^{|E|})  &  =\mathcal{C}_{\text{DS}}^{-0+}%
(\rho)+(0,0,E)\\
&  \subseteq \mathcal{C}_{\text{DS}}^{-0+}(\rho)+\mathcal{C}_{\text{DS}}^{-0+}(\Phi^{|E|})\\
&  =\widetilde{\mathcal{C}}_{\text{DS}}^{-0+}(\rho)+\widetilde{\mathcal{C}}_{\text{DS}}^{-0+}%
(\Phi^{|E|})\\
&  =f^{-1}(\widetilde{\mathcal{C}}_{\text{DS}}^{-+0}(\rho))+f^{-1}(\widetilde
{\mathcal{C}}_{\text{DS}}^{-+0}(\Phi^{|E|}))\\
&  =f^{-1}(\widetilde{\mathcal{C}}_{\text{DS}}^{-+0}(\rho)+\widetilde{\mathcal{C}}_{\text{DS}%
}^{-+0}(\Phi^{|E|})).
\end{align*}
The first set equivalence follows because the capacity region of the noisy
resource state $\rho$ combined with a rate $E$ maximally entangled state is
equivalent to a translation of the capacity region of the noisy resource state
$\rho$. The first inclusion follows because the capacity region of a rate $E$
maximally entangled state contains the rate triple $(0,0,E)$. The second set
equivalence follows from (\ref{eq:ED-capacity}), the third set equivalence from (\ref{DSnpn2}), and the fourth set equivalence from linearity of the map $f$.
The above inclusion implies the following one:
\[
f(\mathcal{C}_{\text{DS}}^{-0+}(\rho\otimes\Phi^{|E|}))\subseteq\widetilde{C}%
_{\text{DS}}^{-+0}(\rho)+\widetilde{\mathcal{C}}_{\text{DS}}^{-+0}(\Phi^{|E|}).
\]
The lemma follows because%
\begin{align*}
f(\mathcal{C}_{\text{DS}}^{-0+}(\rho\otimes\Phi^{|E|}))  &  =f(\widetilde{\mathcal{C}}_{\text{DS}%
}^{-0+}(\rho\otimes\Phi^{|E|}))\\
&  =\widetilde{\mathcal{C}}_{\text{DS}}^{-+0}(\rho\otimes\Phi^{|E|})\\
&  =\mathcal{C}_{\text{DS}}^{-+0}(\rho\otimes\Phi^{|E|}),
\end{align*}
where we apply the relations in (\ref{DSnpn2}) and (\ref{eq:ED-capacity}).
\end{IEEEproof}

Observe that%
\begin{equation}
\widetilde{\mathcal{C}}_{\text{DS}}^{-+0}(\Phi^{|E|})=\widetilde{\mathcal{C}%
}_{\text{U}}^{-+E}. \label{eq:-+0phi=-+EU}%
\end{equation}
Hence for all $E\leq0$,%
\begin{equation}
\mathcal{C}_{\text{DS}}^{-+E}(\rho)=\mathcal{C}_{\text{DS}}^{-+0}(\rho
\otimes\Phi^{|E|})\subseteq\widetilde{\mathcal{C}}_{\text{DS}}^{-+0}%
(\rho)+\widetilde{\mathcal{C}}_{\text{U}}^{-+E}, \label{DSnpn3}%
\end{equation}
where we apply Lemma~\ref{lem:additivity--+0} and (\ref{eq:-+0phi=-+EU}).
Thus,%
\begin{align*}
\mathcal{C}_{\text{DS}}^{-+-}(\rho)  &  =\bigcup_{E\leq0}\widetilde
{\mathcal{C}}_{\text{DS}}^{-+E}(\rho)\\
&  \subseteq\bigcup_{E\leq0}\left(  \widetilde{\mathcal{C}}_{\text{DS}}%
^{-+0}(\rho)+\widetilde{\mathcal{C}}_{\text{U}}^{-+E}\right) \\
&  =(\widetilde{\mathcal{C}}_{\text{DS}}^{-+0}(\rho)+\widetilde{\mathcal{C}%
}_{\text{U}})\cap O^{-+-}\\
&  \subseteq(\widetilde{\mathcal{C}}_{\text{CASR}}(\rho)+\widetilde
{\mathcal{C}}_{\text{U}})\cap O^{-+-}\\
&  =\widetilde{\mathcal{C}}_{\text{DS}}^{-+-}(\rho).
\end{align*}
The first set equivalence holds by definition. The first inclusion follows
from (\ref{DSnpn3}). The second set equivalence follows because $\bigcup
_{E\leq0}\widetilde{\mathcal{C}}_{\text{U}}^{-+E}=\widetilde{\mathcal{C}%
}_{\text{U}}\cap O^{-+-}$. The second inclusion holds because $\widetilde
{\mathcal{C}}_{\text{DS}}^{-+0}(\rho)$ is equivalent to noisy teleportation
and the classically-assisted state redistribution combined with the unit
resource region generates noisy teleportation. The above inclusion
$\mathcal{C}_{\text{DS}}^{-+-}(\rho)\subseteq\widetilde{\mathcal{C}%
}_{\text{DS}}^{-+-}(\rho)$ is the statement of the converse theorem for this octant.

\section{$\boldsymbol{(-,+,+)}$ Octant of the Direct Dynamic Capacity Region}

\label{sec:-++dynamic}The proof of this octant is similar to the proof of the
octant $\left(  +,+,+\right)  $. Define
\begin{align*}
\mathcal{C}_{\text{DD}}^{-++}\left(  \mathcal{N}\right)   &  \equiv
\mathcal{C}_{\text{DD}}\left(  \mathcal{N}\right)  \cap O^{-++},\\
\mathcal{C}_{\text{DD}}^{-0+}\left(  \mathcal{N}\right)   &  \equiv
\mathcal{C}_{\text{DD}}\left(  \mathcal{N}\right)  \cap Q^{-0+}.
\end{align*}
Recall the definition of the line of entanglement distribution $L_{\text{ED}}$
in (\ref{eq:line-ED}). Define the following maps
\begin{align*}
f  &  :S\rightarrow(S+L_{\text{ED}})\cap Q^{-0+},\\
\hat{f}  &  :S\rightarrow(S-L_{\text{ED}})\cap O^{-++}.
\end{align*}
The map $f$ translates the set $S$ in the entanglement distribution direction
and keeps the points that lie on the $Q^{-0+}$ quadrant. The map $\hat{f}$, in
a sense, undoes the effect of $f$ by moving the set $S$ back to the $O^{-++}$ octant.

The inclusion $\mathcal{C}_{\text{DD}}^{-++}\left(  \mathcal{N}\right)
\subseteq\hat{f}(f(\mathcal{C}_{\text{DD}}^{-++}\left(  \mathcal{N}\right)
))$ holds because%
\begin{align}
&  \mathcal{C}_{\text{DD}}^{-++}\left(  \mathcal{N}\right) \nonumber\\
&  =\mathcal{C}_{\text{DD}}^{-++}\left(  \mathcal{N}\right)  \cap
O^{-++}\nonumber\\
&  \subseteq(((\mathcal{C}_{\text{DD}}^{-++}\left(  \mathcal{N}\right)
+L_{\text{ED}})\cap Q^{-0+})-L_{\text{ED}})\cap O^{-++}\nonumber\\
&  =(f(\mathcal{C}_{\text{DD}}^{-++}\left(  \mathcal{N}\right)  )-L_{\text{ED}%
})\cap O^{-++}\nonumber\\
&  =\hat{f}(f(\mathcal{C}_{\text{DD}}^{-++}\left(  \mathcal{N}\right)  )).
\label{DDnpp1}%
\end{align}
The first set equivalence is obvious from the definition of $\mathcal{C}%
_{\text{DD}}^{-++}$. The first inclusion follows from the following logic.
Pick any point $a\equiv\left(  C,Q,E\right)  \in\mathcal{C}_{\text{DD}}%
^{-++}\left(  \mathcal{N}\right)  \cap O^{-++}$ and a particular point
$b\equiv\left(  0,-Q,Q\right)  \in L_{\text{ED}}$. It follows that the point
$a+b=\left(  C,0,E+Q\right)  \in(\mathcal{C}_{\text{DD}}^{-++}\left(
\mathcal{N}\right)  +L_{\text{ED}})\cap Q^{-0+}$. We then pick a point
$-b=\left(  0,Q,-Q\right)  \in-L_{\text{ED}}$. It follows that $a+b-b\in
(((\mathcal{C}_{\text{DD}}^{-++}\left(  \mathcal{N}\right)  +L_{\text{ED}%
})\cap Q^{-0+})-L_{\text{ED}})\cap O^{-++}$ and that $a+b-b=\left(
C,Q,E\right)  =a$. Thus, the first inclusion follows because every point in
$\mathcal{C}_{\text{DD}}^{-++}\cap O^{-++}$ belongs to $(((\mathcal{C}%
_{\text{DD}}^{-++}\left(  \mathcal{N}\right)  +L_{\text{ED}})\cap
Q^{-0+})-L_{\text{ED}})\cap O^{-++}$. The second set equivalence follows from
the definition of $f$, and the third set equivalence follows from the
definition of $\hat{f}$.

It is operationally clear that the following inclusion holds%
\begin{equation}
f(\mathcal{C}_{\text{DD}}^{-++}\left(  \mathcal{N}\right)  )\subseteq
\mathcal{C}_{\text{DD}}^{-0+}\left(  \mathcal{N}\right)  , \label{DDnpp2}%
\end{equation}
because the mapping $f$ converts any achievable point in $\mathcal{C}%
_{\text{DD}}^{-++}\left(  \mathcal{N}\right)  $ to an achievable point in
$\mathcal{C}_{\text{DD}}^{-0+}\left(  \mathcal{N}\right)  $ by combining it
with entanglement distribution.

Forward classical communication does not increase the entanglement generation
capacity \cite{BDSW96,BKN98}. Thus, the following result from
(\ref{eq:converse-forward-class-comm-ent-gen}) applies%
\begin{equation}
\mathcal{C}_{\text{DD}}^{-0+}(\mathcal{N})=\widetilde{\mathcal{C}}_{\text{DD}%
}^{-0+}(\mathcal{N}). \label{DDnpp4}%
\end{equation}
It then follows that%
\begin{align}
&  \widetilde{\mathcal{C}}_{\text{DD}}^{-0+}(\mathcal{N})\nonumber\\
&  =\widetilde{\mathcal{C}}_{\text{DD}}^{00+}(\mathcal{N})+L^{-00}\nonumber\\
&  =(((\widetilde{\mathcal{C}}_{\text{CEF}}(\mathcal{N})\cap Q^{0+-}%
)+L_{\text{ED}})\cap L^{00+})+L^{-00}\nonumber\\
&  \subseteq\widetilde{\mathcal{C}}_{\text{CEF}}(\mathcal{N})+L_{\text{ED}%
}+L^{-00}\nonumber\\
&  =\widetilde{\mathcal{C}}_{\text{CEF}}(\mathcal{N})+L_{\text{ED}%
}+L_{\text{TP}}. \label{DDnpp5}%
\end{align}
The first set equivalence follows from
(\ref{eq:converse-forward-class-comm-ent-gen}). The second set equivalence
follows from (\ref{eq:CAEG-from-father-ED}). The first inclusion follows by
dropping the intersections with $Q^{0+-}$ and $L^{00+}$, and the last
inclusion follows because entanglement distribution and teleportation can
generate any point along $L^{-00}$.

The inclusion $\hat{f}(\widetilde{\mathcal{C}}_{\text{DD}}^{-0+}%
(\mathcal{N}))\subseteq\widetilde{\mathcal{C}}_{\text{DD}}^{-++}(\mathcal{N})$
holds because
\begin{align}
&  \hat{f}(\widetilde{\mathcal{C}}_{\text{DD}}^{-0+}(\mathcal{N}))\nonumber\\
&  \subseteq((\widetilde{\mathcal{C}}_{\text{CEF}}(\mathcal{N})+L_{\text{TP}%
}+L_{\text{ED}})-L_{\text{ED}})\cap O^{-++}\nonumber\\
&  =((\widetilde{\mathcal{C}}_{\text{CEF}}(\mathcal{N})+L_{\text{TP}%
}+L_{\text{ED}})\cap O^{-++})\cup\nonumber\\
&  \ \ \ \ \ \ \ \ \ \ \ \ \ \ \ \ \ \ \ \ \ \ \ \ \ \ ((\widetilde
{\mathcal{C}}_{\text{CEF}}(\mathcal{N})+L_{\text{TP}}-L_{\text{ED}})\cap
O^{-++})\nonumber\\
&  \subseteq\widetilde{\mathcal{C}}_{\text{DD}}^{-++}\left(  \mathcal{N}%
\right)  . \label{DDnpp3}%
\end{align}
The first inclusion follows from (\ref{DDnpp5}) and the definition of $\hat
{f}$. The first set equivalence follows because $(\widetilde{\mathcal{C}%
}_{\text{CEF}}(\mathcal{N})+L_{\text{TP}}+L_{\text{ED}})-L_{\text{ED}%
}=(\widetilde{\mathcal{C}}_{\text{CEF}}(\mathcal{N})+L_{\text{TP}%
}+L_{\text{ED}})\cup(\widetilde{\mathcal{C}}_{\text{CEF}}(\mathcal{N}%
)+L_{\text{TP}}-L_{\text{ED}})$, and the last inclusion follows because
$(\widetilde{\mathcal{C}}_{\text{CEF}}(\mathcal{N})+L_{\text{TP}}%
-L_{\text{ED}})\cap O^{-++}=(0,0,0)$ and (\ref{CEF_DD}).

Putting (\ref{DDnpp1}), (\ref{DDnpp2}), (\ref{DDnpp4}), and (\ref{DDnpp3})
together, the following inclusion holds%
\begin{multline*}
\mathcal{C}_{\text{DD}}^{-++}\left(  \mathcal{N}\right)  \subseteq\hat
{f}(f(\mathcal{C}_{\text{DD}}^{-++}\left(  \mathcal{N}\right)  ))\\
\subseteq\hat{f}(\mathcal{C}_{\text{DD}}^{-0+}\left(  \mathcal{N}\right)
)\subseteq\hat{f}(\widetilde{\mathcal{C}}_{\text{DD}}^{-0+}\left(
\mathcal{N}\right)  )\subseteq\widetilde{\mathcal{C}}_{\text{DD}}^{-++}\left(
\mathcal{N}\right)  .
\end{multline*}
The above inclusion $\mathcal{C}_{\text{DD}}^{-++}\left(  \mathcal{N}\right)
\subseteq\widetilde{\mathcal{C}}_{\text{DD}}^{-++}\left(  \mathcal{N}\right)
$ is the statement of the converse theorem for this octant.

\section{$\boldsymbol{(-,-,+)}$ Octant of the Direct Dynamic Capacity Region}

\label{sec:--+dynamic}The proof technique for this octant is similar to that
for the $\left(  +,-,-\right)  $ octant in the static case. We exploit the
bijection between the quantum communication achievable rate region
$\widetilde{\mathcal{C}}_{\text{DD}}^{0+0}$ and the entanglement generation
achievable rate region $\widetilde{\mathcal{C}}_{\text{DD}}^{00+}$. We need
the following lemma.

\begin{lemma}
\label{lem:additivity--0+}The following inclusion holds%
\[
\mathcal{C}_{\text{DD}}^{-0+}(\mathcal{N}\otimes\text{id}^{\otimes
|Q|})\subseteq\widetilde{\mathcal{C}}_{\text{DD}}^{-0+}(\mathcal{N}%
)+\widetilde{\mathcal{C}}_{\text{DD}}^{-0+}(\text{id}^{\otimes|Q|}).
\]

\end{lemma}

\begin{IEEEproof}
Entanglement distribution induces a bijective mapping $f:\mathcal{C}_{\text{DD}}%
^{-+0}\left(  \mathcal{N}\right)  \rightarrow \mathcal{C}_{\text{DD}}^{-0+}\left(
\mathcal{N}\right)  $ between the classically-assisted quantum communication
achievable region and the classically-assisted entanglement generation
achievable region. It behaves as follows for every point $(C,Q,0)\in
\mathcal{C}_{\text{DD}}^{-+0}\left(  \mathcal{N}\right)  $:%
\[
f:(C,Q,0)\rightarrow(C,0,Q).
\]
The following relation holds%
\begin{equation}
f(\widetilde{\mathcal{C}}_{\text{DD}}^{-+0}\left(  \mathcal{N}\right)  )=\widetilde
{\mathcal{C}}_{\text{DD}}^{-0+}\left(  \mathcal{N}\right)  , \label{eq:ED-CAQ-CAE}%
\end{equation}
because applying entanglement distribution to the classically-assisted quantum
communication protocol produces classically-assisted entanglement generation.
The inclusion $\mathcal{C}_{\text{DD}}^{-+0}(\mathcal{N}\otimes$id$^{\otimes
|Q|})\subseteq f^{-1}(\widetilde{\mathcal{C}}_{\text{DD}}^{-0+}(\mathcal{N}%
)+\widetilde{\mathcal{C}}_{\text{DD}}^{-0+}($id$^{\otimes|Q|}))$ holds because%
\begin{align*}
&  \mathcal{C}_{\text{DD}}^{-+0}(\mathcal{N}\otimes\text{id}^{\otimes|Q|})\\
&  =\mathcal{C}_{\text{DD}}^{-+0}(\mathcal{N})+(0,Q,0)\\
&  \subseteq \mathcal{C}_{\text{DD}}^{-+0}(\mathcal{N})+\mathcal{C}_{\text{DD}}^{-+0}%
(\text{id}^{\otimes|Q|})\\
&  =\widetilde{\mathcal{C}}_{\text{DD}}^{-+0}(\mathcal{N})+\widetilde{\mathcal{C}}_{\text{DD}%
}^{-+0}(\text{id}^{\otimes|Q|})\\
&  =f^{-1}(\widetilde{\mathcal{C}}_{\text{DD}}^{-0+}(\mathcal{N}))+f^{-1}(\widetilde
{\mathcal{C}}_{\text{DD}}^{-0+}(\text{id}^{\otimes|Q|}))\\
&  =f^{-1}(\widetilde{\mathcal{C}}_{\text{DD}}^{-0+}(\mathcal{N})+\widetilde
{\mathcal{C}}_{\text{DD}}^{-0+}(\text{id}^{\otimes|Q|})).
\end{align*}
The first set equivalence follows because the capacity region of the noisy
channel $\mathcal{N}$ combined with a rate $Q$ noiseless qubit channel is
equivalent to a translation of the capacity region of the noisy channel
$\mathcal{N}$. The first inclusion follows because the capacity region of a
rate $Q$ noiseless qubit channel contains the rate triple $(0,Q,0)$. The
second set equivalence follows from the classically-assisted quantum
communication theorem in (\ref{eq:converse-forward-class-comm-ent-gen}), the
third set equivalence from (\ref{eq:ED-CAQ-CAE}), and the fourth set
equivalence from linearity of the map $f$. The above inclusion implies the
following one:
\[
f(\mathcal{C}_{\text{DD}}^{-+0}(\mathcal{N}\otimes\text{id}^{\otimes|Q|}))\subseteq
\widetilde{\mathcal{C}}_{\text{DD}}^{-0+}(\mathcal{N})+\widetilde{\mathcal{C}}_{\text{DD}}%
^{-0+}(\text{id}^{\otimes|Q|}).
\]
The lemma follows because%
\begin{align*}
f(\mathcal{C}_{\text{DD}}^{-+0}(\mathcal{N}\otimes\text{id}^{\otimes|Q|}))  &
=f(\widetilde{\mathcal{C}}_{\text{DD}}^{-+0}(\mathcal{N}\otimes\text{id}^{\otimes
|Q|}))\\
&  =\widetilde{\mathcal{C}}_{\text{DD}}^{-0+}(\mathcal{N}\otimes\text{id}^{\otimes
|Q|})\\
&  =\mathcal{C}_{\text{DD}}^{-0+}(\mathcal{N}\otimes\text{id}^{\otimes|Q|}),
\end{align*}
where we apply the relations in (\ref{eq:converse-forward-class-comm-ent-gen}%
), (\ref{eq:ED-CAQ-CAE}), and (\ref{DDNP0CR}).
\end{IEEEproof}

Observe that
\begin{equation}
\widetilde{\mathcal{C}}_{\text{DD}}^{-0+}(\text{id}^{\otimes|Q|}%
)=\widetilde{\mathcal{C}}_{\text{U}}^{-Q+}. \label{DDnnp1}%
\end{equation}
Hence, for all $Q\leq0$,
\begin{equation}
\mathcal{C}_{\text{DD}}^{-Q+}(\mathcal{N})=\mathcal{C}_{\text{DD}}%
^{-0+}(\mathcal{N}\otimes\text{id}^{\otimes|Q|})\subseteq\widetilde
{\mathcal{C}}_{\text{DD}}^{-0+}(\mathcal{N})+\widetilde{\mathcal{C}}%
_{\text{U}}^{-Q+}, \label{DDnnp2}%
\end{equation}
where we apply Lemma \ref{lem:additivity--0+} and (\ref{DDnnp1}). The
inclusion $\mathcal{C}_{\text{DD}}^{--+}(\mathcal{N})\subseteq\widetilde
{\mathcal{C}}_{\text{DD}}^{--+}(\mathcal{N})$ holds because%
\begin{align*}
\mathcal{C}_{\text{DD}}^{--+}(\mathcal{N})  &  =\bigcup_{Q\leq0}%
\mathcal{C}_{\text{DD}}^{-Q+}(\mathcal{N})\\
&  \subseteq\bigcup_{Q\leq0}(\widetilde{\mathcal{C}}_{\text{DD}}%
^{-0+}(\mathcal{N})+\widetilde{\mathcal{C}}_{\text{U}}^{-Q+})\\
&  =(\widetilde{\mathcal{C}}_{\text{DD}}^{-0+}(\mathcal{N})+\widetilde
{\mathcal{C}}_{\text{U}})\cap O^{--+}\\
&  \subseteq(\widetilde{\mathcal{C}}_{\text{CEF}}(\mathcal{N})+\widetilde
{\mathcal{C}}_{\text{U}})\cap O^{--+}\\
&  =\widetilde{\mathcal{C}}_{\text{DD}}^{--+}(\mathcal{N}).
\end{align*}
The first set equivalence holds by definition. The first inclusion follows
from (\ref{DDnnp2}). The second set equivalence follows because $\bigcup
_{Q\leq0}\widetilde{\mathcal{C}}_{\text{U}}^{-Q+}=\widetilde{\mathcal{C}%
}_{\text{U}}\cap O^{--+}$. The second inclusion follows because combining the
classically-enhanced father region with entanglement distribution and
teleportation gives the region for classically-assisted entanglement
generation. The above inclusion $\mathcal{C}_{\text{DD}}^{--+}(\mathcal{N}%
)\subseteq\widetilde{\mathcal{C}}_{\text{DD}}^{--+}(\mathcal{N})$ is the
statement of the converse theorem for this octant.

\section{Information-Theoretic Argument for the Converse of the
$\boldsymbol{(-,+,-)}$ Dynamic Octant}

\label{sec:converse_-+-_dynamic}We provide an information-theoretic proof of
the following bounds for all rate triples $\left(  -\left\vert C\right\vert
,Q,-\left\vert E\right\vert \right)  $ in the $\left(  -,+,-\right)  $ dynamic
octant (classical- and entanglement-assisted quantum communication):%
\begin{align}
2Q &  \leq I\left(  AX;B\right)  +\left\vert C\right\vert ,\label{eq:-+-1}\\
Q &  \leq I\left(  A\rangle BX\right)  +\left\vert E\right\vert
.\label{eq:-+-2}%
\end{align}
%

%TCIMACRO{\FRAME{ftbpFU}{3.5397in}{2.4007in}{0pt}{\Qcb{(Color online) The most
%general protocol for quantum communication with the help of a noisy channel,
%noiseless entanglement, and noiseless classical communication. Alice wishes to
%communicate a quantum register $A_{1}$ to Bob. She shares entanglement with
%Bob in the form of maximally entangled states. Her half of the entanglement is
%in system $T_{A}$ and Bob's half is in the system $T_{B}$. Alice performs some
%quantum instrument $T$ on her quantum register and her half of the
%entanglement. The output of this instrument is a classical message in some
%register $M$ and a large number of systems $A^{\prime n}$ that are input to
%the channel. She transmits $A^{\prime n}$ through the noisy channel and the
%system $M$ over noiseless classical channels. Bob receives the outputs $B^{n}$
%of the channel and the register $M$ from the noiseless classical channels. He
%combines these with his half of the entanglement and decodes the quantum state
%that Alice transmits.}}{\Qlb{fig:CEA-quantum-comm}}{cea-quantum-comm.pdf}%
%{\special{ language "Scientific Word";  type "GRAPHIC";
%maintain-aspect-ratio TRUE;  display "USEDEF";  valid_file "F";
%width 3.5397in;  height 2.4007in;  depth 0pt;  original-width 8.6533in;
%original-height 5.8461in;  cropleft "0";  croptop "1";  cropright "1";
%cropbottom "0";  filename 'CEA-quantum-comm.pdf';file-properties "XNPEU";}}}%
%BeginExpansion
\begin{figure}
[ptb]
\begin{center}
\includegraphics[
natheight=5.846100in,
natwidth=8.653300in,
height=2.4007in,
width=3.5397in
]%
{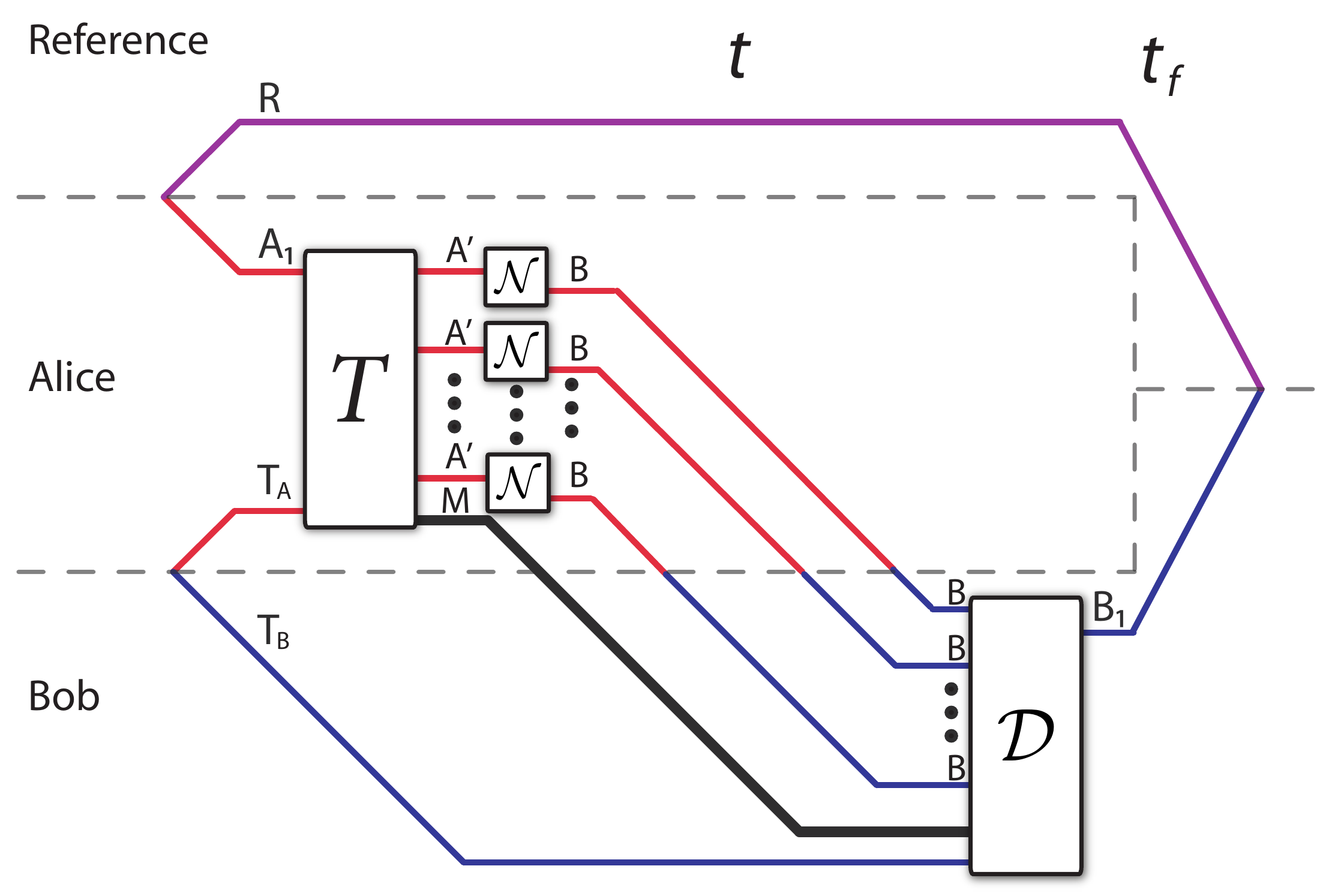}%
\caption{(Color online) The most general protocol for quantum communication
with the help of a noisy channel, noiseless entanglement, and noiseless
classical communication. Alice wishes to communicate a quantum register
$A_{1}$ to Bob. She shares entanglement with Bob in the form of maximally
entangled states. Her half of the entanglement is in system $T_{A}$ and Bob's
half is in the system $T_{B}$. Alice performs some quantum instrument $T$ on
her quantum register and her half of the entanglement. The output of this
instrument is a classical message in some register $M$ and a large number of
systems $A^{\prime n}$ that are input to the channel. She transmits $A^{\prime
n}$ through the noisy channel and the system $M$ over noiseless classical
channels. Bob receives the outputs $B^{n}$ of the channel and the register $M$
from the noiseless classical channels. He combines these with his half of the
entanglement and decodes the quantum state that Alice transmits.}%
\label{fig:CEA-quantum-comm}%
\end{center}
\end{figure}
%EndExpansion
Figure~\ref{fig:CEA-quantum-comm}\ depicts the most general protocol for
classical- and entanglement-assisted quantum communication. Alice wishes to
transmit the $A_{1}$ system of a maximally entangled state $\Phi^{RA_{1}}$ and
shares a maximally entangled state $\Phi^{T_{A}T_{B}}$\ with Bob on systems
$T_{A}$ and $T_{B}$. Her initial state is as follows:%
\[
\omega^{RA_{1}T_{A}T_{B}}\equiv\Phi^{RA_{1}}\otimes\Phi^{T_{A}T_{B}}.
\]
She performs a quantum instrument $T^{A_{1}T_{A}\rightarrow A^{\prime n}M}%
$\ on systems $A_{1}$ and $T_{A}$ to produce a quantum system $A^{\prime n}$
and a classical system $M$ where $A^{\prime n}$ goes to the noisy quantum
channel and $M$ goes to the noiseless classical channel. The state is then:%
\[
\omega^{RA^{\prime n}MT_{B}}\equiv T^{A_{1}T_{A}\rightarrow A^{\prime n}%
M}(\omega^{RA_{1}T_{A}T_{B}}).
\]
The channel $\mathcal{N}$\ transforms $A^{\prime n}$ to $B^{n}$ to produce the
following state:%
\begin{equation}
\omega^{RB^{n}MT_{B}}\equiv\mathcal{N}^{A^{\prime n}\rightarrow B^{n}}%
(\omega^{RA^{\prime n}MT_{B}}). \label{eq:CEA-quantum-state}%
\end{equation}
At this point, the state is almost a state of the form in (\ref{DD_sigma})
with $A\equiv RT_{B}$ (more on this later). Bob combines the classical system
$M$ and the quantum systems $B^{n}$ and $T_{B}$ at a conditional quantum
channel $\mathcal{D}^{MB^{n}T_{B}\rightarrow B_{1}}$\ to produce the state
$B_{1}$ giving the following state:%
\[
(\omega^{\prime})^{RB_{1}}\equiv\mathcal{D}^{MB^{n}T_{B}\rightarrow B_{1}%
}(\omega^{RB^{n}MT_{B}}).
\]
The protocol is $\epsilon$-good if the state $\omega^{\prime}$ is close in
trace distance to Alice's original state $\Phi^{RA_{1}}$:%
\begin{equation}
\left\Vert (\omega^{\prime})^{RB_{1}}-\Phi^{RA_{1}}\right\Vert _{1}%
\leq\epsilon. \label{eq:good-CEA-quantum-code}%
\end{equation}

We first prove the bound in (\ref{eq:-+-1}). Consider the following chain of inequalities:%

\begin{align*}
n2Q  &  =I\left(  R;B_{1}\right)  _{\Phi}\\
&  \leq I\left(  R;B_{1}\right)  _{\omega^{\prime}}+n\delta^{\prime}\\
&  \leq I\left(  R;B^{n}T_{B}M\right)  _{\omega}+n\delta^{\prime}\\
&  =I\left(  R;B^{n}|T_{B}M\right)  _{\omega}+I\left(  R;T_{B}M\right)
_{\omega}+n\delta^{\prime}\\
&  =I\left(  RT_{B}M;B^{n}\right)  _{\omega}-I\left(  T_{B}M;B^{n}\right)
_{\omega}\\
&  \ \ \ \ \ \ +I\left(  R;T_{B}M\right)  _{\omega}+n\delta^{\prime}\\
&  \leq I\left(  RT_{B}M;B^{n}\right)  _{\omega}+I\left(  R;T_{B}\right)
_{\omega}\\
&  \ \ \ \ \ \ +I\left(  R;M|T_{B}\right)  _{\omega}+n\delta^{\prime}\\
&  =I\left(  RT_{B}M;B^{n}\right)  _{\omega}+H\left(  M|T_{B}\right)
_{\omega}\\
&  \ \ \ \ \ \ -H\left(  M|T_{B}R\right)  _{\omega}+n\delta^{\prime}\\
&  \leq I\left(  AM;B^{n}\right)  _{\omega}+n\left\vert C\right\vert
+n\delta^{\prime}%
\end{align*}
The first equality follows by evaluating the quantum mutual information on the
maximally entangled state $\Phi^{RB_{1}}$. The first inequality follows from
the condition in (\ref{eq:good-CEA-quantum-code}) and from a variation of the
Alicki-Fannes' inequality with $\delta^{\prime}\equiv5\left\vert Q\right\vert
\epsilon+3H_{2}\left(  \epsilon\right)  /n$ (Corollary~1 of
Ref.~\cite{HW08GFP}). The second inequality follows from the quantum data
processing inequality \cite{NC00}. The third and fourth equalities follow by
expanding the quantum mutual information $I\left(  R;B^{n}T_{B}M\right)
_{\omega}$\ with the chain rule. The third inequality follows because
$I\left(  T_{B}M;B^{n}\right)  _{\omega}\geq0$ and by expanding the quantum
mutual information $I\left(  R;T_{B}M\right)  _{\omega}$ with the chain rule.
The fourth equality follows because $I\left(  R;T_{B}\right)  _{\omega}=0$ for
this protocol and by rewriting the mutual information $I\left(  R;M|T_{B}%
\right)  _{\omega}$. The last inequality follows because $H\left(
M|T_{B}\right)  _{\omega}\leq n\left\vert C\right\vert $ and $I\left(
A_{1};T_{B}|M\right)  _{\omega}=H\left(  A_{1}|M\right)  _{\omega}-H\left(
A_{1}|T_{B}M\right)  _{\omega}$. The final inequality follows from the
definition $A\equiv A_{1}T_{B}$ and because $H\left(  M|T_{B}R\right)
_{\omega}\geq0$.

We now prove the bound in (\ref{eq:-+-1}). Consider the following chain of
inequalities:%
\begin{align*}
nQ  &  =I\left(  R\rangle B_{1}\right)  _{\Phi}\\
&  \leq I\left(  R\rangle B_{1}\right)  _{\omega^{\prime}}+n\delta^{\prime}\\
&  \leq I\left(  R\rangle B^{n}T_{B}M\right)  _{\omega}+n\delta^{\prime}\\
&  =H\left(  B^{n}T_{B}M\right)  _{\omega}-H\left(  RB^{n}T_{B}M\right)
_{\omega}+n\delta^{\prime}\\
&  =H\left(  B^{n}M\right)  _{\omega}+H\left(  T_{B}|B^{n}M\right)  _{\omega
}\\
&  \ \ \ \ \ \ -H\left(  RB^{n}T_{B}M\right)  _{\omega}+n\delta^{\prime}\\
&  \leq I\left(  A\rangle B^{n}M\right)  _{\omega}+n\left\vert E\right\vert
+n\delta^{\prime}.
\end{align*}
The first equality follows by evaluating the coherent information of the
maximally entangled state $\Phi^{RB_{1}}$. The first inequality follows from
the condition in (\ref{eq:good-CEA-quantum-code}) and from the Alicki-Fannes'
inequality with $\delta^{\prime}\equiv4\left\vert Q\right\vert \epsilon
+2H_{2}\left(  \epsilon\right)  /n$. The third inequality follows from quantum
data processing \cite{NC00}. The next two equalities follow by expanding the
coherent information. The final inequality follows from the definition
$A\equiv RT_{B}$ and because $H\left(  T_{B}|B^{n}M\right)  _{\omega}\leq
n\left\vert E\right\vert $.

We should make some final statements concerning this proof. The state in
(\ref{eq:CEA-quantum-state}) as we have defined it is not quite a state of the
form in (\ref{DD_sigma}) because the instrument has an environment. Though, a
few arguments demonstrate that a particular type of instrument works just as
well as a general instrument, and it then follows that the state in
(\ref{eq:CEA-quantum-state}) is of the form in (\ref{DD_sigma}). First,
consider that a general instrument $T^{A_{1}T_{A}\rightarrow A^{\prime n}M}%
$\ has a realization as an isometry $U_{T}^{A_{1}T_{A}\rightarrow A^{\prime
n}E^{\prime}E_{M}}$ followed by a von Neumann measurement of the system
$E_{M}$ in the basis $\{\left\vert m\right\rangle \left\langle m\right\vert
^{M}\}$ (see the discussion of the CP\ formalism in Ref.~\cite{DHW05RI}). The
system $E^{\prime}$ is not involved in any of the entropic expressions in
(\ref{eq:-+-1}-\ref{eq:-+-2}). Thus, Alice can measure the system $E^{\prime}$
in some classical basis $\left\vert l\right\rangle \left\langle l\right\vert
$, obtaining a classical variable $L$, and a new state $\sigma$ whose
entropies in (\ref{eq:-+-1}-\ref{eq:-+-2})\ are the same as those of the
original state $\omega$. Additionally, the action of Alice's von Neumann
measurement of $E^{\prime}$ makes the state $\sigma$ be a state of the form in
(\ref{DD_sigma}). The following inequalities then hold by the quantum data
processing inequality:%
\begin{align*}
I\left(  AM;B^{n}\right)  _{\omega}  &  \leq I\left(  AML;B^{n}\right)
_{\sigma},\\
I\left(  A\rangle B^{n}M\right)  _{\omega}  &  \leq I\left(  A\rangle
B^{n}ML\right)  _{\sigma},
\end{align*}
demonstrating that it is sufficient to consider states of the form in
(\ref{DD_sigma}) for determining the capacity region for this octant.

\bibliographystyle{IEEEtran}
\bibliography{Ref}

\end{document}